\documentclass[a4paper,oribibl,orivec]{llncs} 
\usepackage{proof} 
\usepackage{fullpage}
\usepackage{cite}
\usepackage{abbrevs}
\usepackage{xspace}
\usepackage{pslatex}
\def\Bscr{{\mathcal B}}

\def\Gscr{{\mathcal G}}

\def\Iscr{{\mathcal I}}

\def\Lscr{{\mathcal L}}

\def\Rscr{{\mathcal R}}
\def\Sscr{{\mathcal S}}

\newcommand{\Momega}{\mathcal{M}_{\omega}}

\newcommand{\FOL   }{FO\lambda}

\newcommand{\FOLDN }{\FOL^{\Delta\N}}

\newcommand{\FOLNb }{\FOL^{\nabla}}

\newcommand{\Linc}{{\rm Linc}^-}

\newcommand{\N}{{\rm I} \! {\rm N}}
\newcommand{\Seq}[2]{#1\longrightarrow #2}

\newcommand{\botL}{\bot{\cal L}}
\newcommand{\bulletL}{\bullet{\cal L}}

\newcommand{\cL}{\hbox{\sl c}{\cal L}}
\newcommand{\circL}{\circ{\cal L}}
\newcommand{\circR}{\circ{\cal R}}

\newcommand{\defL}{\hbox{\sl def}{\cal L}}

\newcommand{\defeq}{\stackrel{\scriptscriptstyle\triangle}{=}}
\newcommand{\defmu}{\stackrel{\mu}{=}}
\newcommand{\defnu}{\stackrel{\nu}{=}}

\newcommand{\eqL}{{\rm eq}{\cal L}}
\newcommand{\eqR}{{\rm eq}{\cal R}}
\newcommand{\eval}[2]{#1\mathop{\Downarrow}#2}

\newcommand{\existsL}{\exists{\cal L}}
\newcommand{\existsR}{\exists{\cal R}}

\newcommand{\forallL}{\forall{\cal L}}
\newcommand{\forallR}{\forall{\cal R}}

\newcommand{\indR}{{\rm I}{\cal R}}
\newcommand{\indL}{{\rm I}{\cal L}}
\newcommand{\indm}[1]{{\rm indm}(#1)}
\newcommand{\coindR}{{\rm CI}{\cal R}}
\newcommand{\coindL}{{\rm CI}{\cal L}}

\newcommand{\init}{\hbox{\sl init}}

\newcommand{\lam}{{\rm lam}}
\newcommand{\landL}{\land{\cal L}}
\newcommand{\landR}{\land{\cal R}}
\newcommand{\level}[1]{{\rm lvl}(#1)}

\newcommand{\lorL}{\lor{\cal L}}
\newcommand{\lorR}{\lor{\cal R}}

\newcommand{\lub}[1]{{\rm lub}(#1)}

\newcommand{\mc}{\hbox{\sl mc}}
\newcommand{\measure}[1]{{\rm ht}(#1)}

\newcommand{\oimpL}{\oimp{\cal L}}
\newcommand{\oimpR}{\oimp{\cal R}}
\newcommand{\oimp}{\supset}

\newcommand{\ra}{\to}

\newcommand{\simm}[2]{\hbox{\sl sim}~#1~#2}

\newcommand{\topR}{\top{\cal R}}

\newcommand{\wL}{\hbox{\sl w}{\cal L}}

\def\RED{{\mathbf{RED}}}
\def\NM{{\mathbf{NM}}}
\def\idrv{{\mathrm{Id}}}

\newcommand{\conc}[2]{#1\!::\!#2}
\newcommand{\append}[3]{{\rm app}~#1~#2~#3}
\newcommand{\coappend}[3]{{\rm coapp}~#1~#2~#3}
\newcommand{\nil}{{\rm nil}}

                         {\endgroup]\end{quote}}
\long\def\ignore#1{}
\newlength{\infwidthi}
\newlength{\infwidthii}

\title{Induction and Co-induction in Sequent Calculus}

\author{Alwen Tiu\inst{1} \and Alberto Momigliano\inst{2}}

\institute{
  The Australian National University\\
  \email{Alwen.Tiu@rsise.anu.edu.au}\\
  \and
  LFCS, University of Edinburgh \\
  \email{amomigl1@inf.ed.ac.uk}\\
}

\begin{document}
\maketitle

\begin{abstract}
  Proof search has been used to specify a wide range of computation
  systems.  In order to build a framework for reasoning about such
  specifications, we make use of a sequent calculus involving
  induction and co-induction.  These proof principles are based on a
  proof theoretic (rather than set-theoretic) notion of
  \emph{definition}~\cite{PID,eriksson91elp,schroeder-heister93lics,mcdowell00tcs}.
  Definitions are akin to (stratified) logic programs, where the left
  and right rules for defined atoms allow one to view theories as
  ``closed'' or defining fixed points.  The use of definitions makes
  it possible to reason intensionally about syntax, in particular
  enforcing free equality via unification. We add in a consistent way
  rules for pre and post fixed points, thus allowing the user to
  reason inductively and co-inductively about properties of
  computational system making full use of higher-order abstract
  syntax.  Consistency is guaranteed via  cut-elimination,
  where we give the first, to our knowledge, cut-elimination procedure
  in the presence of general inductive and co-inductive definitions.
 
\end{abstract}





\section{Introduction}
\label{sec:intro}
A common approach to specifying computation systems is via deductive
systems. Those are used to specify and reason about various logics, as
well as aspects of programming languages such as operational
semantics, type theories, abstract machines \etc.  Such specifications
can be represented as logical theories in a suitably expressive formal
logic where \emph{proof-search} can then be used to model the
computation.  A logic used as a specification language is known as a
\emph{logical frameworks} \cite{pfenning01handbook}, which comes
equipped with a representation methodology.  The encoding of the
syntax of deductive systems inside formal logic can benefit from the
use of \emph{higher-order abstract syntax} (HOAS) \cite{PfenningE88},
a high-level and declarative treatment of object-level bound variables
and substitution. At the same time, we want to use such a logic in
order to reason over the \emph{meta-theoretical} properties of object
languages, for example type preservation in operational
semantics~\cite{mcdowell02tocl}, soundness and completeness of
compilation~\cite{Momigliano03fos} or congruence of bisimulation in
transition systems~\cite{mcdowell03tcs}. Typically this involves
reasoning by (structural) induction and, when dealing with infinite
behavior, co-induction \cite{Jacobs97}.

The need to support both inductive and co-inductive reasoning and some
form of HOAS requires some careful design decisions, since the two are
prima facie notoriously incompatible. While any meta-language based on
a $\lambda$-calculus can be used to specify and animate HOAS
encodings, meta-reasoning has traditionally involved (co)inductive
specifications both at the level of the syntax and of the judgements
--- which are of course unified at the type-theoretic level. The first
provides crucial freeness properties for datatypes constructors, while
the second offers principle of case analysis and (co)induction. This
is well-known to be problematic, since HOAS specifications lead to
non-monotone (co)inductive operators, which by cardinality and
consistency reasons are not permitted in inductive logical
frameworks. Moreover, even when HOAS is weakened so as to be made
compatible with standard proof assistants~\cite{despeyroux94lpar} such
as HOL or Coq, the latter suffer the fate of allowing the existence of
too many functions and yielding the so called \emph{exotic}
terms. Those are canonical terms in the signature of an HOAS encoding
that do not correspond to any term in the deductive system under
study. This causes a loss of adequacy in HOAS specifications, which is
one of the pillar of formal verification, and it undermines the trust
in formal derivations. On the other hand, logics such as
LF~\cite{harper93jacm} that are weak by design \cite{DeBruijn91lf} in
order to support this style of syntax are not directly endowed with
(co)induction principles.


The contribution of this paper lies in the design of a new logic,
called $\Linc$ (for a logic with $\lambda$-terms, induction and
co-induction),\footnote{The ``minus'' in the terminology refers to the
  lack of the $\nabla$ quantifier \wrt the eponymous logic in Tiu's
  thesis~\cite{tiu04phd}.} which carefully adds principles of
induction and co-induction to a higher-order intuitionistic logic
based on a proof theoretic notion of \emph{definition}, following on
work (among others) Lars Halln{\"{a}}s \cite{PID}, Eriksson
\cite{eriksson91elp}, Schroeder-Heister~\cite{schroeder-heister93lics}
and McDowell and Miller~\cite{mcdowell00tcs}.  Definitions are akin to
logic programs, but allow us to view theories as ``closed'' or defining
fixed points.  This alone allows us to perform case analysis
independently from induction principles.  Our
approach to formalizing induction and co-induction is via the least
and greatest solutions of the fixed point equations specified by the
definitions.  Such least and greatest solutions are guaranteed to
exist by imposing a stratification condition on definitions (which basically
ensures monotonicity).  The proof rules for induction and co-induction
makes use of the notion of \emph{pre-fixed points} and
\emph{post-fixed points} respectively. In the inductive case, this
corresponds to the induction invariant, while in the co-inductive one
to the so-called simulation.

The simply typed language underlying $\Linc$ and the notion of
definition make it possible to reason \emph{intensionally} about
syntax, in particular enforcing \emph{free} equality via unification,
which can be used on first-order terms or higher-order
$\lambda$-terms. In fact, we can support HOAS encodings of constants
without requiring them to be the constructors of a (recursive)
datatype, which could not exist for cardinality reasons.  In
particular we can \emph{prove} the freeness properties of those
constructors, namely injectivity, distinctness and case exhaustion.
Judgements are encoded as definitions accordingly to their informal
semantics, either inductive or co-inductive.  Definitions that are
true in every fixed point will not be given here special consideration.


$\Linc$ can be proved to be a conservative extension of
$\FOLDN$~\cite{mcdowell00tcs} and a generalization with a higher-order
language of Martin-L\"of \cite{martin-lof71sls} first-order theory of
iterated inductive definitions. Moreover, to the best of our
knowledge, it is the first sequent calculus with a syntactical
cut-elimination theorem for co-inductive definitions.  In recent
years, several logical systems have been designed that build on the
core features of $\Linc$. In particular, one interesting, and
orthogonal, extension is the addition of the
$\nabla$-quantifier~\cite{miller05tocl,tiu04phd,Tiu07,gacek08lics},
which allows one to reason about the intentional aspects of
\emph{names and bindings} in object syntax specifications (see, e.g.,
\cite{tiu05entcs,tiu05concur,AbellaSOS}).  The cut elimination proof
presented in this paper can be used as a springboard towards cut
elimination procedures for more expressive (conservative) extensions
of $\Linc$ such as the ones with $\nabla$. Here lies the added value
of the present paper, which extends and revises a conference paper
published in the proceedings of TYPES 2003
\cite{Momigliano03TYPES}. In the conference version, the co-inductive
rule had a technical side condition that is restrictive and
unnatural. The restriction was essentially imposed by the particular
cut elimination proof technique outlined in that paper.  This
restriction has been removed in the present version, and as such the
cut elimination proof itself has consequently been significantly
revised.


The rest of the paper is organized as follows. Section~\ref{sec:linc}
introduces the sequent calculus for the logic $\Linc$.
Section~\ref{sec:ex} shows some examples of using induction and
co-induction to prove properties of list-related predicates and
the lazy $\lambda$-calculus. Section~\ref{sec:drv} studies several
properties of derivations in $\Linc$ that will be used
extensively in the cut-elimination proof (Section~\ref{sec:cut-elim}).
Section~\ref{sec:lrel} surveys the related work and
Section~\ref{sec:conc} concludes this paper.

\section{The Logic $\Linc$}
\label{sec:linc}

\begin{figure}
\paragraph{Core rules:}
$$
\quad \quad \quad
\begin{array}{cc}
\infer[\cL]{\Seq{B,\Gamma}{C}}
        {\Seq{B,B,\Gamma}{C}}
\quad
\infer[\wL]{\Seq{B,\Gamma}{C}}{\Seq{\Gamma}{C}}
&
\infer[\botL]{\Seq{\bot,\Gamma}{B}}{\rule{0pt}{6pt}}
\quad \infer[\topR]{\Seq{\Gamma}{\top}}{}
\\ \\
\infer[\landL]{\Seq{B \land C,\Gamma}{D}}
        {\Seq{B,\Gamma}{D}}
\quad \infer[\landL]{\Seq{B \land C,\Gamma}{D}}
        {\Seq{C,\Gamma}{D}}
& 
\infer[\landR]{\Seq{\Gamma}{B \land C}}
        {\Seq{\Gamma}{B}
        & \Seq{\Gamma}{C}}
\\ \\
\infer[\lorL]{\Seq{B \lor C,\Gamma}{D}}
        {\Seq{B,\Gamma}{D}
        & \Seq{C,\Gamma}{D}}
& 
\infer[\lorR]{\Seq{\Gamma}{B \lor C}}
        {\Seq{\Gamma}{B}}
\quad \infer[\lorR]{\Seq{\Gamma}{B \lor C}}
        {\Seq{\Gamma}{C}}
\\ \\
\infer[\forallL]{\Seq{\forall x.B\,x,\Gamma}{C}}
        {\Seq{B\,t,\Gamma}{C}}
\quad \infer[\forallR]{\Seq{\Gamma}{\forall x.B\,x}}
        {\Seq{\Gamma}{B\,y}}
& 
\infer[\existsL]{\Seq{\exists x.B\,x,\Gamma}{C}}
        {\Seq{B\,y,\Gamma}{C}}
\quad \infer[\existsR]{\Seq{\Gamma}{\exists x.B\,x}}
        {\Seq{\Gamma}{B\,t}}
\\ \\
\infer[\oimpL]{\Seq{B \oimp C,\Gamma}{D}}
        {\Seq{\Gamma}{B}
        & \Seq{C,\Gamma}{D}}
& \infer[\oimpR]{\Seq{\Gamma}{B \oimp C}}
        {\Seq{B,\Gamma}{C}}\\ \\
\multicolumn{2}{c}{        
\infer[init]{\Seq{C}{C}}{} \quad
        \infer[\begin{array}{l}
                 \mc, 
                 \mbox{ where } n > 0 
               \end{array}]
        {\Seq{\Delta_1,\dots,\Delta_n, \Gamma}{C}}
        {\Seq{\Delta_1}{B_1}
        & \cdots &
        \Seq{\Delta_n}{B_n} &
        \Seq{B_1,\dots,B_n, \Gamma}{C}}}        
\end{array}
$$
\dotfill\\[1em]
\paragraph{Equality rules:}
$$
\infer[\eqL]
{\Seq{s = t, \Gamma}{C}}
{
\{\Seq{\Gamma\rho}{C\rho}~\mid~s\rho =_{\beta\eta} t\rho \}
}
\qquad
\infer[\eqR]
{\Seq{\Gamma}{t = t}}
{}
$$
\dotfill\\[1em]
\paragraph{Induction rules:}
$$
\infer[\indL, p\,\vec x \defmu B\,p\,\vec x]
{\Seq{\Gamma, p\,\vec{t}}{C}}
{
\Seq{B\,S \, \vec{y}}{S\,\vec{y}}
& 
\Seq{\Gamma, S\,\vec{t}}{C}
}
\qquad
\infer[\indR, p\,\vec x \defmu B\,p\,\vec x]
{\Seq{\Gamma}{p\,\vec{t}}}
{\Seq{\Gamma}{B\,p\,\vec{t}}}
$$
\dotfill\\[1em]
\paragraph{Co-induction rules:}
$$
\infer[\coindL, p\,\vec x \defnu B\,p\,\vec x]
{\Seq{p\,\vec{t}, \Gamma}{C}}
{
                \Seq{B\,p\,\vec{t}, \Gamma}{C}             
}
\quad
\infer[\coindR, p\,\vec x \defnu B\,p\,\vec x]
{\Seq{\Gamma}{p\,\vec{t}}}
{
\Seq{\Gamma}{S\,\vec{t}}
&
\Seq{S\,\vec{y}}{B\,S\,\vec{y}}
}
$$
\caption{The inference rules of $\Linc$}
\label{fig:linc}
\end{figure}

The logic $\Linc$ shares the core fragment of $\FOLDN$, which is an
intuitionistic version of Church's Simple Theory of Types.  Formulae
in the logic are built from predicate symbols and the usual logical
connectives $\bot$, $\top$, $\land$, $\lor$, $\oimp$, $\forall_\tau$
and $\exists_\tau$.  Following Church, formulae will be given type
$o$.  The quantification type $\tau$ (omitted in the rest of the
paper) can have base or higher types, but those are restricted not to 
contain $o$.  Thus the logic has a first-order proof theory but allows
the encoding of higher-order abstract syntax. 

We assume the usual notion of capture-avoiding
substitutions. Substitutions are ranged over by lower-case Greek
letters, e.g., $\theta$, $\rho$ and $\sigma$. Application of
substitution is written in postfix notation, \eg $t\theta$ denotes
the term resulting from an application of substitution $\theta$ to
$t$. Composition of substitutions, denoted by $\circ$, is defined as
$t (\theta \circ \rho) = (t\theta)\rho$.

The whole logic is presented in the sequent calculus in
Figure~\ref{fig:linc}. A sequent is denoted by $\Seq{\Gamma}{C}$ where
$C$ is a formula and $\Gamma$ is a multiset of formulae.  Notice that
in the presentation of the rule schemes, we make use of HOAS, e.g., in
the application $B\,x$ it is implicit that $B$ has no free occurrence
of $x$.  In particular we work modulo $\alpha$-conversion without
further notice. In the $\forallR$ and $\existsL$ rules, $y$ is an
eigenvariable that is not free in the lower sequent of the
rule. Whenever we write  a sequent, it is assumed implicitly that
the formulae are well-typed and in $\beta\eta$-long normal forms: the
type context, i.e., the types of the constants and the eigenvariables
used in the sequent, is left implicit as well. The $\mc$ rule is a
generalization of the cut rule that simplifies the presentation of the
cut-elimination proof.

We extend the core fragment with a proof theoretic notion of
equality and fixed points. Each of these extensions are discussed
below. 

\subsection{Equality} 

The right introduction rule for equality is the standard one, that is, 
it recognizes that two terms are syntactically equal. 
The left introduction rule is  more interesting.  The
substitution $\rho$ in $\eqL$ is a \emph{unifier} of $s$ and $t$. Note
that we specify the premise of $\eqL$ as a set, with the intention
that every sequent in the set is a premise of the rule.  This set is
of course infinite, since for every unifier of $(s,t)$, we can extend
it to another unifier (e.g., by adding substitution pairs for
variables not in the terms).  However, in many cases, it is sufficient
to consider a particular set of unifiers, which is often called a
\emph{complete set of unifiers (CSU)} \cite{BaaderS01}, from which any
unifier can be obtained by composing a member of the CSU set with a
substitution.  In the case where the terms are first-order terms, or
higher-order terms with the pattern restriction~\cite{Miller91elp}, the
set CSU is a singleton, i.e., there exists a most general unifier
(MGU) for the terms.

In examples and applications, we shall use a more restricted 
version of $\eqL$ using CSU:
$$
\infer[\eqL_{CSU}]
{\Seq{s = t, \Gamma}{C}}
{
\{\Seq{\Gamma\rho}{C\rho}~\mid~s\rho =_{\beta\eta} t\rho, \rho \in CSU(s,t) \}
}
$$
Replacing $\eqL$ with $\eqL_{CSU}$ does not change the class of
provable formulae, as shown in \cite{tiu04phd}.
Note that in applying $\eqL$ and $\eqL_{CSU}$, eigenvariables can be
instantiated as a result. Note also that if the premise set of $\eqL$
and $\eqL_{CSU}$ are empty, then the sequent in the conclusion is
considered proved.

Our treatment of equality implicitly assumes the notion of  \emph{free
  equality} as commonly found in logic programming. More specifically,
 the axioms of free equality \cite{clark78}, that is, injectivity of
function symbols, inequality between distinct function symbols, and
the ``occur-check'' are enforced via unification in the $\eqL$-rule.
For instance, given a base type $nt$ (for natural numbers) and the
constants $z : nt$ (zero) and $s : nt \rightarrow nt$ (successor), we
can derive $\forall x.\  z = (s~x) \oimp \bot$ as follows:
$$
\infer[\forallR]
{\Seq {} {\forall x.\  z = (s~x) \oimp \bot}}
{
 \infer[\oimpR]
 {\Seq{}{z = (s~x) \oimp \bot}}
 {
  \infer[\eqL]
  {\Seq {z = (s~x)} {\bot}}
  {}
 }
}
$$
Since $z$ and $s~x$ are not unifiable, the $\eqL$ rule above has empty
premise, thus concluding the derivation.  We can also prove the
injectivity of the successor function, \ie $\forall x \forall y. (s~x)
= (s~y) \oimp x = y$.

This proof theoretic notion of equality has been considered in several previous
work \eg by by Schroeder-Heister \cite{schroeder-heister93lics},
and McDowell and Miller~\cite{mcdowell00tcs}.

\subsection{Induction and co-induction}

One way of adding induction and co-induction is to introduce fixed
point expressions and their associated introduction rules, \ie using
the $\mu$ and $\nu$ operators of the (first-order)
$\mu$-calculus. This is essentially what we shall follow here, but
with a different notation. Instead of using a ``nameless'' notation
using $\mu$ and $\nu$ to express fixed points, we associate a fixed
point equation with an atomic formula. That is, we associate certain
designated predicates with a \emph{definition}.  This notation is
clearer and more convenient as far as  our examples and
applications are concerned. For the proof system using nameless notation for
inductive and co-inductive predicates, the interested reader is
referred to a recent work by Baelde and Miller \cite{baelde07lpar}.

\begin{definition}
\label{def:def-clause}
An \emph{inductive definition clause} is written $\forall\vec{x}.p \,
\vec{x} \defmu B\,\vec{x}$, where $p$ is a predicate constant and
$\vec x$ is a sequence of variables.
The atomic formula $p \, \vec{x}$ is called the \emph{head} of the clause, and the
formula $B\,\vec{x}$, where $B$ is a closed term, is called the {\em
  body}.  Similarly, a \emph{co-inductive definition clause} is written
$\forall\vec{x}.p \, \vec{x} \defnu B\,\vec{x}$.  The symbols $\defmu$
and $\defnu$ are used simply to indicate a definition clause: they are
not a logical connective.  A {\em definition} is a set of definition
clauses. 
\end{definition}

It is technically convenient to bundle up all the definitional clause
for a given predicate in a single clause, so that a  predicate may
occur only at most once in the heads of the clauses of a
definition, following the same principles of the \emph{iff-completion}
in logic programming \cite{SchroederHeister93elp}. Further, in order
to simplify the presentation of some rules that involve predicate
substitutions, we sometimes denote a definition using an abstraction
over predicates, that is
$$
\forall \vec x.\  p\,\vec x \defmu B\,p\,\vec x
$$
where $B$ is an abstraction with no free occurrence of predicate
symbol $p$ and variables $\vec x$.  Substitution of $p$ in the body of
the clause with a formula $S$ can then be written simply as
$B\,S\,\vec x$.  When writing definition clauses, we often omit the
outermost universal quantifiers, with the assumption that free
variables in a clause are universally quantified (such variables will
often be denoted with capital letters).  We shall write $\forall \vec
x.\ p\,\vec x \defeq B\,p\,\vec x$ to denote a definition clause
generally, i.e., when we are not interested in the details of whether
it is an inductive or a co-inductive definition.

The introduction rules for (co-)inductively defined atoms are given at
the bottom of 
Figure~\ref{fig:linc}.  The abstraction $S$ is an invariant of the
(co-)induction rule, which is of the same type as $p$.  The variables
$\vec{y}$ are new eigenvariables.  For the induction rule $\indL$, $S$
denotes a pre-fixed point of the underlying fixed point
operator. Similarly, for the co-induction rule $\coindL$, $S$ can be
seen as denoting a {post-fixed point} of the same 
operator.  Here, we use a characterization of induction and
co-induction proof rules as, respectively, the least and the greatest
solutions to a fixed point equation. To guarantee soundness of these
rules, we shall restrict the (co)inductive definitions to ones which
are monotone. In this case, the Knaster-Tarski fixed point theorems tell
us that the existence of a pre-fixed point (respectively, post-fixed
point) implies the existence of a least (resp., greatest) fixed point.
Monotonicity is enforced by a syntactic condition on definitions, as
it is used for the logic $\FOLDN$\cite{mcdowell00tcs}: we rule out
definitions with circular calling through implications (negations)
that can lead to inconsistency~\cite{schroeder-heister92nlip}.  The
notion of \emph{level} of a formula allows us to define a proper
stratification on definitions.
\begin{definition}
\label{def:level}
To each predicate $p$ we associate a natural number $\level{p}$, the
level of $p$.  Given a formula $B$, its \emph{level} $\level{B}$ is
defined as follows:
\begin{enumerate}
\item $\level{p \, \vec{t}} = \level{p}$,
\item $\level{\bot} = \level{\top} = 0,$ 
\item $\level{B \land C} = \level{B \lor C} = \max(\level{B},\level{C})$
\item $\level{B \oimp C} = \max(\level{B}+1,\level{C})$
\item $\level{\forall x.\ B\,x} = \level{\exists x.\ B\,x}
       = \level{B\,t}$, for any term $t$. 
\end{enumerate}
The level of a sequent $\Seq{\Gamma}{C}$ is the level of $C$.  
A formula $B$ is said to be \emph{dominated} by a predicate symbol $p$,
if $\level B \leq \level p$ and 
 $\level {B[\lambda \vec x. \top/p]} < \level p$, where $\lambda \vec x.\top$
is of the same type as $p$. 
A definition clause $\forall\vec{x}.\ p\, \vec{x} \defeq B\,\vec{x}$ is
\emph{stratified} if $B\,\vec x$ is dominated by $p$.  
\end{definition}
Note that when $p$ is vacuous in $B$ and $p$ dominates $B$, we
obviously have $\level{B} < \level p$.

From now on, we shall be concerned only with stratified definitions.
An occurrence of a formula $A$ in a formula $C$ is \emph{strictly
  positive} if that particular occurrence of $A$ is not to the left of
any implication in $C$. Stratification then implies that all
occurrences of the head in the body are strictly positive, and that
there is no mutual recursion between different definition clauses.
This restriction to non-mutual recursion is just for the sake of
simplicity in the presentation of the underlying idea of the cut
elimination proof.  This proof (Section~\ref{sec:cut-elim}) can be
extended to handle mutually recursive definitions with some
straightforward, albeit tedious, modifications.  In the first-order
case, the restriction to non-mutual recursion is immaterial, since
one can easily encode mutually recursive predicates as a single
predicate with an extra argument.  For example, consider the following
mutual recursive definitions for even and odd numbers.
$$
\begin{array}{lll}
  even ~ X & \defmu & X = z \lor \exists y. y = (s~X) \land odd ~ y. \\
  odd ~ X & \defmu & \exists y. y = (s~X) \land even ~ y.
\end{array}
$$
We can collapse these two definition clauses into a single one, with a
parameter that takes a constant $e$ (for `even') or $o$ (for `odd'):
$$
\begin{array}{ll}
  evod ~ W ~ X ~ \defmu ~ & 
  [W = e \land (X = z \lor \exists y.\ y = (s~X) \land evod~ o ~ y)]
  \lor\mbox{} \\ 
  & 
  [W = o \land (\exists y.\ y = (s~X) \land evod ~ e ~ y)].
\end{array}
$$
We then define even and odd as follows:
$$
\begin{array}{ll}
even ~ X \defmu &evod ~ e ~ X. 
\\
odd ~ X \defmu & evod ~ o ~ X.
\end{array}$$
This definition can be stratified by assigning levels to the predicate symbols
such that 
$$\level {evod} < \level{even} < \level{odd}.$$


\section{Examples}
\label{sec:ex}
We now give some examples, starting with some that make essential use
of HOAS.

\subsection{Lazy $\lambda$-Calculus}

We consider an untyped version of the pure $\lambda$-calculus with
lazy evaluation, following the usual HOAS style, i.e.,
object-level $\lambda$-operator and application are encoded as
constants ${\rm lam} : (tm \ra tm) \ra tm$ and ${\rm @} : tm \ra tm
\ra tm$, where $tm$ is the syntactic category of object-level
$\lambda$-terms.  The evaluation relation is encoded as the following
inductive definition
$$
\begin{array}{lcl}
\eval{M}{N} & \defmu & [\exists M'.\ (M = \lam\,M') \land (M = N)]
\lor\mbox{} \\ 
                        & &     [\exists M_1 \exists M_2 \exists P.\ 
                                   (M = M_1\, @\, M_2) \land
                                   \eval{M_1}{\lam\,P} \land \eval{(P\,M_2)}{N}
]\end{array}                                
\enspace 
$$
Notice that object-level substitution is realized via
$\beta$-reduction in the meta-logic.

The notion of \emph{applicative simulation} of $\lambda$-expressions
\cite{Ong} can be encoded as the (stratified) co-inductive definition
$$
\simm{R}{S} \defnu 
   \forall T.\ \eval{R}{\lam\, T} \oimp   
      \exists U.\ \eval{S}{\lam\, U} \land 
          \forall P. \simm{(T\,P)}{(U\,P)}.  
          $$
Given this encoding, we can prove the reflexivity property of simulation, i.e., 
$\forall s.\ \simm{s}{s}$. This is proved co-inductively by
using the simulation $\lambda x \lambda y.\ x = y$.  After
applying $\forallR$ and $\coindR$, it remains to prove 
the sequents $\Seq{}{s = s}$, and
$$               \Seq{x = y}{    
                   \forall x_1.\ \eval{x}{\lam\, x_1}
                   \oimp   
                      (\exists x_2.\ \eval{y}{\lam\, x_2}
                           \land \forall x_3. (x_1 \, x_3) = (x_2 \, x_3)
                      )
                }
\enspace 
$$
The first sequent is provable by an application of $\eqR$ rule.
The second sequent is proved as follows.
{\small
$$
        \infer[\forallR]
        {
                \Seq{x = y}{    
                   \forall x_1.\eval{x}{\lam\, x_1}
                   \oimp   
                      (\exists x_2.\eval{y}{\lam\, x_2}
                           \land \forall x_3. (x_1 \, x_3) = (x_2 \, x_3)
                      )
                }
        }
        {
                \infer[\oimpR]
                {
                  \Seq{x = y}{    
                     \eval{x}{\lam\, x_1}
                     \oimp   
                      (\exists x_2.\eval{y}{\lam\, x_2}
                           \land \forall x_3. (x_1 \, x_3) = (x_2 \, x_3)
                      )
                  }
                }
                {
                  \infer[\eqL]
                  {
                    \Seq{x = y, \eval{x}{\lam\, x_1}}{    
                      (\exists x_2.\eval{y}{\lam\, x_2}
                           \land \forall x_3. (x_1 \, x_3) = (x_2 \, x_3)
                      )
                    }
                  }
                  {
                          \infer[\existsR]
                          {
                            \Seq{\eval{z}{\lam\, x_1}}{    
                              (\exists x_2.\eval{z}{\lam\, x_2}
                                   \land \forall x_3. (x_1 \, x_3) = (x_2 \, x_3)
                              )
                            }             
                          }
                          {
                          \infer[\landR]
                          {
                            \Seq{\eval{z}{\lam\, x_1}}{    
                              (\eval{z}{\lam\, x_1}  
                                   \land \forall x_3. (x_1 \, x_3) = (x_1 \, x_3)
                              )
                            }                                             
                          }
                          {
                          \infer[\init]
                          {\Seq{\eval{z}{\lam\,x_1}}{\eval{z}{\lam\,x_1}}}
                          {  }
                          &
                          \infer[\forallR]
                          {\Seq{\eval{z}{\lam\,x_1}}{\forall x_3. (x_1 \, x_3) = (x_1 \, x_3)}}
                          {
                                  \infer[\eqR]
                                  {\Seq{\eval{z}{\lam\,x_1}}{(x_1 \, x_3) = (x_1 \, x_3)}}
                                  {              
                                  }
                          }
                          }
                          }               
                  }             
                }
        }
$$
}


The transitivity property is expressed as 
$
\forall r \forall s \forall t. \simm{r}{s} \land \simm{s}{t} \oimp \simm{r}{t}.
$
Its proof involves co-induction on $\simm{r}{t}$ with the 
simulation
$
\lambda u \lambda v. \exists w. \simm{u}{w} \land \simm{w}{v},
$
followed by case analysis (i.e., $\defL$ and $\eqL$ rules) on
$\simm{r}{s}$ and $\simm{s}{t}$. The rest of the proof is purely logical.

We can also show the existence of divergent terms. 
Divergence is encoded as follows.
$$
\begin{array}{lcl}
{\rm divrg}~T & \defnu &
   [
      \exists T_1 \exists T_2.\ T = (T_1 @\, T_2) \land {\rm divrg\ }
   T_1   ]     \lor\mbox{} \\
    & &   
       [\exists T_1 \exists T_2. \ T = (T_1 @ T_2) \land
       \exists E.\ \eval{T_1}{\lam\, E} \land {\rm divrg\ } (E\, T_2)  
     ]. 
\end{array}
$$
Let $\Omega$ be the term $(\lam\, x.(x\,@\,x))\, @ \, (\lam\,
x.(x\,@\,x))$.  We show that ${\rm divrg}~\Omega$ holds. The proof is
straightforward by co-induction using the simulation $S := \lambda s.\ 
s = \Omega$.  Applying the $\coindR$ produces the sequents
$\Seq{}{\Omega = \Omega}$ and $ \Seq{T = \Omega} { S_1~\lor~S_2} $
where
$$S_1 := \exists T_1 \exists T_2.\  T = (T_1 @\, T_2) \land (S\, T_1), \mbox{ and }$$
$$
S_2 :=        \exists T_1 \exists T_2. \ T = (T_1 @ T_2) \land
       \exists E.\ \eval{T_1}{\lam\, E} \land S\, (E\, T_2).
       $$
       Clearly, only the second disjunct is provable, i.e., by
       instantiating $T_1$ and $T_2$ with the same term $\lam\,
       x.(x\,@\,x)$, and $E$ with the function $\lambda x.(x\, @ \,
       x)$.

\subsection{Lists}

\newcommand\tlist{\hbox{\sl lst}} Lists over some fixed type $\alpha$
are encoded as the type $\tlist$, with the usual constructor $\nil :
\tlist$ for empty list and $::$ of type $ \alpha\ra\tlist\ra\tlist$.
We consider here the append predicate for both the finite and infinite
case.

\paragraph{Finite lists}
The usual append predicate on finite lists can be encoded as the
inductive definition
$$
\begin{array}{lcl}
\append{L_1}{L_2}{L_3} & \defmu & 
  [(L_1 = \nil) \land (L_2 = L_3)] \lor\mbox{}\\
  & & [\exists x\exists L_1'\exists L_3'. \ (L_1 = \conc{x}{L_1'})
  \land (L_3 = \conc{x}{L_3'}) 
\land \append{L_1'}{L_2}{L_3'}].
\end{array}
$$
Associativity of append is stated formally as
$$
\begin{array}{c}
\forall l_1 \forall l_2 \forall l_{12} \forall l_3 \forall l_4. 
 (\append{l_1}{l_2}{l_{12}} \land \append{l_{12}}{l_3}{l_4}) \oimp 
  \forall l_{23}. \append{l_2}{l_3}{l_{23}} \oimp \append{l_1}{l_{23}}{l_4}.
\end{array}
$$
Proving this formula requires us to prove first that the definition of
append is functional, that is,
$$
\forall l_1 \forall l_2 \forall l_3 \forall l_4.
   \append{l_1}{l_2}{l_3} \land \append{l_1}{l_2}{l_4} \oimp l_3 = l_4.
   $$
   This is done by induction on $l_1$, i.e., we apply the $\indL$
   rule on $\append{l_1}{l_2}{l_3}$, after the introduction rules for
   $\forall$ and $\oimp$, of course. The invariant in this case is
$$
S := \lambda r_1 \lambda r_2 \lambda r_3. \forall r. \append{r_1}{r_2}{r} \oimp r = r_3.
$$
It is a simple case analysis to check that this is the right
invariant.  Back to our original problem: after applying the
introduction rules for the logical connectives in the formula, the
problem of associativity is reduced to the following sequent
\begin{equation}
\label{eq:app1}
\Seq{\append{l_1}{l_2}{l_{12}},~\append{l_{12}}{l_3}{l_4},~\append{l_2}{l_{3}}{l_{23}}}
    {\append{l_1}{l_{23}}{l_4}.}
\end{equation}    
We then proceed by induction on the list $l_1$, that is, we apply
the $\indL$ rule to the hypothesis $\append{l_1}{l_2}{l_{12}}$.
The invariant is simply
$$
S  := \lambda l_1 \lambda l_2 \lambda l_{12}. 
     \forall l_3 \forall l_4. \append{l_{12}}{l_3}{l_4} \oimp 
     \forall l_{23}. \append{l_2}{l_3}{l_{23}} \oimp 
      \append{l_1}{l_{23}}{l_4}.
$$
Applying the $\indL$ rule, followed by $\lorL$, to sequent
(\ref{eq:app1}) reduces the sequent to the following sub-goals
\begin{description}
\item[$(i)$] $\Seq{S~l_1\,l_2\,l_{12},~\append{l_{12}}{l_3}{l_4},~
     \append{l_2}{l_{3}}{l_{23}}}{\append{l_1}{l_{23}}{l_4}}$,
\item[$(ii)$] $\Seq{(l_1 = \nil \land l_2 = l_3)}{S~l_1\,l_2\,l_3}$,
\item[$(iii)$] $\Seq{\exists x, l_1', l_3'. l_1 = \conc{x}{l_1'} \land l_3 = \conc{x}{l_3'}
  \land S~l_1'\,l_2\,l_3'}{S~l_1\,l_2\,l_3}$
\end{description}
The proof for the second sequent is straightforward. The first sequent
reduces to
$$
\Seq
{\append{l_{12}}{l_3}{l_4}, \append{l_{12}}{l_3}{l_{23}}}
{\append{\nil}{l_{23}}{l_4}}.
$$
This follows from the functionality of append and $\indR$. The
third sequent follows by case analysis. Of
course, these proofs could have been simplified by using a
\emph{derived} principle of \emph{structural} induction. While this is
easy to do, we have preferred here to use the primitive $\indL$ rule.

\paragraph{Infinite lists}

The append predicate on infinite lists is defined via co-recursion,
that is, we define the behavior of \emph{destructor operations} on
lists (i.e., taking the head and the tail of the list).  In this case
we never construct explicitly the result of appending two lists,
rather the head and the tail of the resulting lists are computed as
needed.  The co-recursive append requires case analysis on all
arguments.
$$
\begin{array}{lcl}
\coappend{L_1}{L_2}{L_3} & \defnu & 
  [(L_1 = \nil) \land (L_2 = \nil) \land (L_3 = \nil)]\lor\mbox{} \\
  & & [(L_1 = \nil) \land \exists x \exists L_2' \exists L_3'.\
      (L_2 = \conc{x}{L_2'}) \land (L3 = \conc{x}{L_3'}) 
  ~\land~ \coappend{ \nil}{L_2'}{L_3'}] \lor\mbox{}\\
  & &  [\exists x \exists L_1' \exists L_3'.\ (L_1 = \conc{x}{L_1'}) \land 
       (L_3 = \conc{x}{L_3'})  ~\land~ \coappend{L_1'}{L_2}{L_3'}].
\end{array}
$$
The corresponding associativity property is stated analogously  to the
inductive one and 
the main statement reduces to proving the sequent
$$
\Seq{\coappend{l_1}{l_2}{l_{12}},~ \coappend{l_{12}}{l_3}{l_4}, ~
     \coappend{l_2}{l_{3}}{l_{23}}}
    {\coappend{l_1}{l_{23}}{l_4}.}
$$
We apply the $\coindR$ rule to $\coappend{l_1}{l_{23}}{l_4}$, using
the simulation
$$
S := \lambda l_1 \lambda l_2 \lambda l_{12}.
   \exists l_{23} \exists l_3 \exists l_4. 
   \coappend{l_{12}}{l_3}{l_{4}} \land~ 
  \coappend{l_{2}}{l_3}{l_{23}} \land~ 
   \coappend{l_1}{l_{23}}{l_4}.
   $$
   Subsequent steps of the proof involve mainly case analysis on
   $\coappend{l_{12}}{l_3}{l_4}$. As in the inductive case, we have to
   prove the sub-cases when $l_{12}$ is $\nil$. However, unlike in the
   former case, case analysis on the arguments of ${\rm coapp}$
   suffices.


\section{Properties of derivations}
\label{sec:drv}

We discuss several properties of derivations in $\Linc$.  Some of them
involve transformations on derivations which will be used extensively
in the cut-elimination proof in Section~\ref{sec:cut-elim}.  Before we
proceed, some remarks on the use of eigenvariables in derivations are
useful. In proof search involving $\forallR$, $\existsL$ $\indL$,
$\coindR$ or $\eqL$, new eigenvariables can be introduced in the
premises of the rules. Let us refer to such variables as internal
eigenvariables, since they occur only in the premise derivations.  We
view the choice of such eigenvariables as arbitrary and therefore we
identify derivations that differ only in the choice of the
eigenvariables introduced by those rules. Another way to look at it is
to consider eigenvariables as proof-level binders. Hence when we work
with a derivation, we actually work with an equivalence class of
derivations modulo renaming of internal eigenvariables.

\subsection{Instantiating derivations}
\label{sec:subst}

The following definition extends substitutions to apply to
derivations.  Since we identify derivations that differ only in the
choice of variables that are not free in the end-sequent, we will
assume that such variables are chosen to be distinct from the
variables in the domain of the substitution and from the free
variables of the range of the substitution.  Thus applying a
substitution to a derivation will only affect the variables free in
the end-sequent.

\begin{definition}
  \label{def:subst}
  If $\Pi$ is a derivation of $\Seq{\Gamma}{C}$ and $\theta$ is a
  substitution, then we define the derivation $\Pi\theta$ of
  $\Seq{\Gamma\theta}{C\theta}$ as follows:
  \begin{enumerate}
  \item Suppose $\Pi$ ends with the $\eqL$ rule
    \begin{displaymath}
      \infer[\eqL]{\Seq {s = t,\Gamma'}{C} }
      {\left\{\raisebox{-1.5ex}
          {\deduce{\Seq{\Gamma'\rho}{C\rho}}
            {\Pi^{\rho}}}
        \right\}_{\rho}}
      \enspace 
    \end{displaymath}
    where $s\rho =_{\beta\eta} t\rho$. Observe that any unifier for
    the pair $(s\theta, t\theta)$ can be transformed to another
    unifier for $(s, t)$, by composing the unifier with $\theta$.
    Thus $\Pi\theta$ is
    \begin{displaymath}
      \infer[\eqL]{\Seq{s\theta = t\theta,\Gamma'\theta}{C\theta}}
      {\left\{\raisebox{-1.5ex}
          {\deduce{\Seq{\Gamma'\theta\rho'}{C\theta\rho'}}
            {\Pi^{\theta\circ\rho'}}}
	\right\}_{\rho'}}
      \enspace ,
    \end{displaymath}
    where $s\theta\rho' =_{\beta\eta} t\theta\rho'$.

  \item If $\Pi$ ends with any other rule and has premise derivations
    $\Pi_1, \ldots, \Pi_n$, then $\Pi\theta$ also ends with the same
    rule and has premise derivations $\Pi_1\theta, \ldots,
    \Pi_n\theta$.
  \end{enumerate}
\end{definition}

Among the premises of the inference rules of $\Linc$, certain premises
share the same right-hand side formula with the sequent in the
conclusion.  We refer to such premises as major premises. This notion
of major premise will be useful in proving cut-elimination, as certain
proof transformations involve only major premises.
\begin{definition}
  \label{def:major-premise}
  Given an inference rule $R$ with one or more premise sequents, we
  define its major premise sequents as follows.
  \begin{enumerate}
  \item If $R$ is either $\oimpL, \mc$ or $\indL$, then its rightmost
    premise is the major premise
  \item If $R$ is $\coindR$ then its left premise is the major
    premise.
  \item Otherwise, all the premises of $R$ are major premises.
  \end{enumerate}
  A \emph{minor premise} of a rule $R$ is a premise of $R$ which is not
  a major premise.  The definition extends to derivations by replacing
  premise sequents with premise derivations.
\end{definition}

The following two measures on derivations will be useful later in
proving many properties of the logic. Given a set of measures $\Sscr$,
we denote with $\lub \Sscr$ the least upper bound of $\Sscr$.

\begin{definition}
  \label{def:mu}
  Given a derivation $\Pi$ with premise derivations $\{\Pi_i\}_i$, the
  measure $\measure{\Pi}$ is 
  $\lub {\{\measure{\Pi_i}\}_i} + 1$.
\end{definition}
\begin{definition}
  \label{def:num-ind}
  Given a derivation $\Pi$ with premise derivations $\{\Pi_i\}_i$, the
  measure $\indm{\Pi}$ is defined as follows
$$
\indm{\Pi} = \left\{ {\begin{array}{l}
      \lub{\{\indm{\Pi_i}\}_i} + 1, \mbox{ if $\Pi$ ends with $\indL$, }\\
      \lub{\{\indm{\Pi_i}\}_i}, \mbox{ otherwise. }
    \end{array}}
\right.
$$
\end{definition}

Note that given the possible infinite branching of $\eqL$ rule, these
measures in general can be ordinals.  Therefore in proofs involving
induction on those measures, transfinite induction is needed. However,
in most of the inductive proofs to follow, we often do case analysis
on the last rule of a derivation. In such a situation, the inductive
cases for both successor ordinals and limit ordinals are basically
covered by the case analysis on the inference figures involved, and we
shall not make explicit use of transfinite induction.

\begin{lemma}
  \label{lm:subst}
  For any substitution $\theta$ and derivation $\Pi$ of
  $\Seq{\Gamma}{C}$, $\Pi\theta$ is a derivation of
  $\Seq{\Gamma\theta}{C\theta}$.
\end{lemma}
\begin{proof}
  This lemma states that Definition~\ref{def:subst} is
  well-constructed, and follows by induction on $\measure{\Pi}$.  \qed
\end{proof}

\begin{lemma}
  \label{lm:subst-mu}
  For any derivation $\Pi$ and substitution $\theta$, $\measure{\Pi}
  \geq \measure{\Pi\theta}$ and $\indm{\Pi} \geq \indm{\Pi\theta}$.
\end{lemma}
\begin{proof}
  By induction on $\measure{\Pi}$.  The measures may not be equal
  because in the case where the derivation ends with the $\eqL$ rule, some of the
  premise derivations of $\Pi$ may not be needed to construct the
  premise derivations of $\Pi\theta$.  \qed
\end{proof}

\begin{lemma}
  \label{lm:subst-drv-comp}
  For any derivation $\Pi$ and substitutions $\theta$ and $\rho$, the
  derivations $(\Pi\theta)\rho$ and $\Pi(\theta \circ \rho)$ are the
  same derivation.
\end{lemma}
\begin{proof}
  By induction on the measure $\measure{\Pi}$.  \qed
\end{proof}

\subsection{Atomic initial rule}

It is a common property of most logics that the initial rule can be
restricted to atomic form, that is, the rule
$$
\infer[\init] {\Seq{p\,\vec{t}}{p\,\vec{t}}} {}
$$
where $p$ is a predicate symbol. The more general rule is derived as
follows.

\begin{definition}
  \label{def:id-drv}
  We construct a derivation $\idrv_C$ of the sequent $\Seq{C}{C}$
  inductively as follows. The induction is on the size of $C$.  If $C$
  is an atomic formula we simply apply the atomic initial rule.
  Otherwise, we apply the left and right introduction rules for the
  topmost logical constant in $C$, probably with some instances of
  the contraction and the weakening rule.
\end{definition}

The proof of the following lemma is straightforward by induction on
$\measure{\idrv_C}$.

\begin{lemma}
\label{lm:idrv}
  For any formula $C$, it holds that $\indm{\idrv_C} = 0$.
\end{lemma}

Restricting the initial rule to atomic form will simplify some
technical definitions to follow. We shall use $\idrv$ instead of
$\idrv_C$ to denote identity derivations since the formula $C$ is
always known from context.


\subsection{Unfolding of derivations}
\label{sec:unfolding}

\begin{definition}
  \label{def:inductive-unfolding} \emph{Inductive unfolding.}
  Let $p\,\vec{x} \defmu B\,p\,\vec{x}$ be an inductive definition.
  Let $\Pi$ be a derivation of $\Seq{\Gamma}{C}$ where $p$ dominates
  $C$. Let $S$ be a closed term of the same type as $p$ and let
  $\Pi_S$ be a derivation of the sequent
$$\Seq{B\,S\,\vec{x}}{S\,\vec{x}}$$
where $\vec x$ are new eigenvariables not free in $\Gamma$ and $C$. We
define the derivation $\mu_C^p(\Pi,\Pi_S)$ of $\Seq{\Gamma}{C[S/p]}$
as follows.

If $p$ is vacuous in $C$, then $\mu_C^p(\Pi,\Pi_S) = \Pi$.  Otherwise,
we define $\mu_C^p(\Pi,\Pi_S)$ according to the last rule of $\Pi$.
\begin{enumerate}
\item Suppose $\Pi$ ends with $\init$
$$
\infer[\init] {\Seq{p\,\vec{t}}{p\,\vec{t}}} {}.
$$
Then $\mu_C^p(\Pi,\Pi_S)$ is the derivation
$$
\infer[\indL] {\Seq{p\,\vec{t}}{S\,\vec{t}}} {
  \deduce{\Seq{B\,S\,\vec{x}}{S\,\vec{x}}}{\Pi_S} & \deduce
  {\Seq{S\,\vec{t}}{S\,\vec{t}}} {\idrv} }
$$

\item Suppose $\Pi$ ends with $\oimpL$
$$
\infer[\oimpL] {\Seq{D_1\oimp D_2, \Gamma'}{C}} {
  \deduce{\Seq{\Gamma'}{D_1}}{\Pi_1} & \deduce{\Seq{D_2,
      \Gamma'}{C}}{\Pi_2} }
$$
Then $\mu_C^p(\Pi,\Pi_S)$ is the derivation
$$
\infer[\oimpL] {\Seq{D_1\oimp D_2, \Gamma'}{C[S/p]}} {
  \deduce{\Seq{\Gamma'}{D_1}}{\Pi_1} & \deduce{\Seq{D_2,
      \Gamma'}{C[S/p]}}{\mu_C^p(\Pi_2,\Pi_S)} }
$$

\item Suppose $\Pi$ ends with $\oimpR$
$$
\infer[\oimpR] {\Seq{\Gamma}{C_1 \oimp C_2}} { \deduce{\Seq{\Gamma,
      C_1}{C_2}}{\Pi'} }
$$
Note that since $p$ dominates $C$, it must be the case that $p$ does
not occur in $C_1$.  The derivation $\mu(\Pi, \Pi_S)$ is then defined
as follows.
$$
\infer[\oimpR] {\Seq{\Gamma}{C_1 \oimp C_2[S/p]}} { \deduce{\Seq
    {\Gamma, C_1 } {C_2[S/p]}} {\mu_{C_2}^p(\Pi',\Pi_S)} }
$$

\item Suppose $\Pi$ ends with $\mc$
$$
\infer[\mc] {\Seq{\Delta_1,\dots,\Delta_m, \Gamma'}{C}} {
  \deduce{\Seq{\Delta_1}{B_1}}{\Pi_1} & \ldots &
  \deduce{\Seq{\Delta_m}{B_m}}{\Pi_m} &
  \deduce{\Seq{B_1,\dots,B_m,\Gamma'}{C}}{\Pi'} }
$$
Then $\mu_C^p(\Pi,\Pi_S)$ is
$$
\infer[\mc] {\Seq{\Delta_1,\dots,\Delta_m, \Gamma'}{C[S/p]}} {
  \deduce{\Seq{\Delta_1}{B_1}}{\Pi_1} & \ldots &
  \deduce{\Seq{\Delta_m}{B_m}}{\Pi_m} &
  \deduce{\Seq{B_1,\dots,B_m,\Gamma'}{C[S/p]}}{\mu_C^p(\Pi',\Pi_S)} }
$$

\item Suppose $\Pi$ ends with $\indL$ on some predicate $q$ given a
  definition clause $q\,\vec{z} \defmu D\,q\,\vec{z}$.
$$
\infer[\indL] {\Seq{q\,\vec{t},\Gamma'}{C}} {
  \deduce{\Seq{D\,I\,\vec{z}}{I\,\vec{z}}}{\Psi} & \deduce
  {\Seq{I\,\vec{t},\Gamma'}{C}} {\Pi'} }
$$
Then $\mu_C^p(\Pi, \Pi_S)$ is the derivation
$$
\infer[\indL] {\Seq{q\,\vec{t},\Gamma'}{C[S/p]}} {
  \deduce{\Seq{D\,I\,\vec{z}}{I\,\vec{z}}}{\Psi} & \deduce
  {\Seq{I\,\vec{t},\Gamma'}{C[S/p]}} {\mu_C^p(\Pi',\Pi_S)} }
$$

\item Suppose $\Pi$ ends with $\indR$
$$
\infer[\indR .]  {\Seq{\Gamma}{p\,\vec{t}}} {
  \deduce{\Seq{\Gamma}{B\,p\,\vec{t}}}{\Pi'} }
$$
Then $\mu_C^p(\Pi,\Pi_S)$ is the derivation
$$
\infer[\mc .]  {\Seq{\Gamma}{S\,\vec{t}}} {
  \deduce{\Seq{\Gamma}{B\,S\,\vec{t}}}{\mu_{B\,p}^p(\Pi',\Pi_S)} &
  \deduce{\Seq{B\,S\,\vec{t}} {S\,\vec{t}}}{\Pi_S[\vec t/ \vec x]} }
$$

\item If $\Pi$ ends with any other rules, and has premise derivations
$$
\left\{ \raisebox{-1.5ex}{\deduce{\Seq{\Gamma_i}{C_i}}{\Pi_i}}
\right\}_{i \in \Iscr}
$$ 
for some index set $\Iscr$, then $\mu_C^p(\Pi,\Pi_S)$ also ends with
the same rule and has premise derivations
$\{\mu_{C_i}^p(\Pi_i,\Pi_S)\}_{i \in \Iscr}$.
\end{enumerate}

\end{definition}

\begin{definition}
  \label{def:coinductive-unfolding}
  \emph{Co-inductive unfolding.}  Let $p\,\vec{x} \defnu
  B\,p\,\vec{x}$ be a co-inductive definition.  Let $S$ be a closed
  term of the same type as $p$ and let $\Pi_S$ be a derivation of
$$
\Seq{S\,\vec{x}}{B\,S\,\vec{x}}.
$$
Let $C$ be a formula dominated by $p$, and let $\Pi$ be a derivation
of $\Seq{\Gamma}{C[S/p]}$.  We define the derivation
$\nu_C^p(\Pi,\Pi_S)$ of $\Seq{\Gamma}{C}$ as follows.

If $p$ is vacuous in $C$, then $\nu_C^p(\Pi,\Pi_S) = \Pi$.  If $C =
p\,\vec{t}$ then $C[S/p] = S\,\vec{t}$ and $\nu_C^p(\Pi,\Pi_S)$ is the
derivation
$$
\infer[\coindR] {\Seq{\Gamma}{p\,\vec{t}}} {
  \deduce{\Seq{\Gamma}{S\,\vec{t}}}{\Pi} &
  \deduce{\Seq{S\,\vec{x}}{B\,S\,\vec{x}}}{\Pi_S} }
$$
Otherwise, we define $\nu_C^p(\Pi,\Pi_S)$ based on the last rule in
$\Pi$.
\begin{enumerate}
\item Suppose $\Pi$ ends with $\oimpL$
$$
\infer[\oimpL] {\Seq{D_1\oimp D_2, \Gamma'} {C[S/p]}} {
  \deduce{\Seq{\Gamma'}{D_1}}{\Pi_1} & \deduce{\Seq{D_2,
      \Gamma'}{C[S/p]}}{\Pi_2} }
$$
Then $\nu_C^p(\Pi,\Pi_S)$ is the derivation
$$
\infer[\oimpL] {\Seq{D_1\oimp D_2, \Gamma'}{C}} {
  \deduce{\Seq{\Gamma'}{D_1}}{\Pi_1} & \deduce{\Seq{D_2,
      \Gamma'}{C}}{\nu_C^p(\Pi_2,\Pi_S)} }
$$

\item Suppose $\Pi$ ends with $\oimpR$
$$
\infer[\oimpR] {\Seq{\Gamma}{(C_1 \oimp C_2)[S/p]}} {
  \deduce{\Seq{\Gamma, C_1}{C_2[S/p]}}{\Pi'} }
$$
Note that since $p$ dominates $C$, it must be the case that $p$ is
vacuous in $C_1$.  Therefore we construct the derivation $\nu_C^p(\Pi,
\Pi_S)$ as follows.
$$
\infer[\oimpR] {\Seq{\Gamma}{C_1 \oimp C_2}} { \deduce{\Seq{\Gamma,
      C_1}{C_2}}{\nu_{C_2}^p(\Pi',\Pi_S)} }
$$

\item Suppose $\Pi$ ends with $\mc$
$$
\infer[\mc] {\Seq{\Delta_1,\dots,\Delta_m, \Gamma'}{C[S/p]}} {
  \deduce{\Seq{\Delta_1}{B_1}}{\Pi_1} & \ldots &
  \deduce{\Seq{\Delta_m}{B_m}}{\Pi_m} &
  \deduce{\Seq{B_1,\dots,B_m,\Gamma'}{C[S/p]}}{\Pi'} }
$$
Then $\nu_C^p(\Pi,\Pi_S)$ is
$$
\infer[\mc] {\Seq{\Delta_1,\dots,\Delta_m, \Gamma'}{C}} {
  \deduce{\Seq{\Delta_1}{B_1}}{\Pi_1} & \ldots &
  \deduce{\Seq{\Delta_m}{B_m}}{\Pi_m} &
  \deduce{\Seq{B_1,\dots,B_m,\Gamma'}{C}}{\nu_C^p(\Pi',\Pi_S)} }
$$

\item Suppose $\Pi$ ends with $\indL$ on a predicate $q\,\vec{t}$,
  given an inductive definition $q\,\vec{z} \defmu D\,q\,\vec{z}$.
$$
\infer[\indL] {\Seq{q\,\vec{t},\Gamma'}{C[S/p]}} {
  \deduce{\Seq{D\,I\,\vec{z}}{I\,\vec{z}}}{\Psi} & \deduce
  {\Seq{I\,\vec{t},\Gamma'}{C[S/p]}} {\Pi'} }
$$
Then $\nu_C^p(\Pi, \Pi_S)$ is the derivation
$$
\infer[\indL] {\Seq{q\,\vec{t},\Gamma'}{C}} {
  \deduce{\Seq{D\,I\,\vec{z}}{I\,\vec{z}}}{\Psi} & \deduce
  {\Seq{I\,\vec{t},\Gamma'}{C}} {\nu_C^p(\Pi',\Pi_S)} }
$$

\item If $\Pi$ ends with any other rules, and has premise derivations
$$
\left\{ \raisebox{-1.5ex}{\deduce{\Seq{\Gamma_i}{C_i[S/p]}}{\Pi_i}}
\right \}_{i \in \Iscr}
$$ 
for some index set $\Iscr$, then $\nu_C^p(\Pi,\Pi_S)$ also ends with
the same rule and has premise derivations $\{\nu_C^p(\Pi_i,\Pi_S)\}_{i
  \in \Iscr}$.
\end{enumerate}
\end{definition}

The following two lemmas state that substitutions commute with
unfolding of derivations. Their proofs follow straightforwardly from
the fact that the definitions of (co-)inductive unfolding depend only
on the logical structures of conclusions of sequents, hence is
orthogonal to substitutions of eigenvariables. In these lemmas, we
assume that the formulas $C$, $p$ and derivations $\Pi$ and $\Pi_S$
satisfy the conditions of Definition~\ref{def:inductive-unfolding} and
Definition~\ref{def:coinductive-unfolding}.

\begin{lemma}
  \label{lm:ind unfold subst}
  The derivations $\mu_C^p(\Pi,\Pi_S)\theta$ and
  $\mu_C^p(\Pi\theta,\Pi_S)$ are the same derivation.
\end{lemma}
\begin{lemma}
  \label{lm:coind unfold subst}
  The derivations $\nu_C^p(\Pi,\Pi_S)\theta$ and
  $\nu_C^p(\Pi\theta,\Pi_S)$ are the same derivation.
\end{lemma}


\newcommand\restrict[2]{{#1}\!\downarrow_{#2}}

\section{Cut elimination for $\Linc$}
\label{sec:cut-elim}

A central result of our work is cut-elimination, from which
consistency of the logic follows.  Gentzen's classic proof of
cut-elimination for first-order logic uses an induction on the size of
the cut formula, i.e., the number of logical connectives in the
formula. The cut-elimination procedure consists of a set of reduction
rules that reduce a cut of a compound formula to cuts on its
sub-formulae of smaller size.  In the case of $\Linc$, the use of
induction/co-induction complicates the reduction of cuts.  Consider
for example a cut involving the induction rules
$$
\infer[\mc] {\Seq{\Delta, \Gamma}{C} } {\infer[\indR]
  {\Seq{\Delta}{p\,t}} {\deduce{\Seq{\Delta}{B\,p\,t}}{\Pi_1}} &
  \infer[\indL] {\Seq{p\,t, \Gamma}{C}} {
    \deduce{\Seq{B\,S\,y}{S\,y}}{\Pi_B} & \deduce{\Seq{S\,t,
        \Gamma}{C}}{\Pi} } } \enspace
$$
There are at least two problems in reducing this cut. First, any
permutation upwards of the cut will necessarily involve a cut with $S$
that can be of larger size than $p$, and hence a simple induction on
the size of cut formula will not work.  Second, the invariant $S$ does
not appear in the conclusion of the left premise of the cut. The
latter means that we need to transform the left premise so that its
end sequent will agree with the right premise. Any such transformation
will most likely be \emph{global}, and hence simple induction on the
height of derivations will not work either.

We shall use the \emph{reducibility} techniques to prove cut
elimination.  More specifically, we shall build on the notion of
reducibility introduced by Martin-L\"of to prove normalization of an
intuitionistic logic with iterative inductive definition
\cite{martin-lof71sls}.  Martin-L\"of's proof has been adapted to
sequent calculus by McDowell and Miller \cite{mcdowell00tcs}, but in a
restricted setting where only natural number induction is
allowed. Since our logic involves arbitrary stratified inductive
definitions, which also includes iterative inductive definitions, we
shall need a more general cut reductions. But the real difficulty in
our case is really in establishing cut elimination in the presence of
co-inductive definitions, for which there is no known cut elimination
proof for the sequent calculus formulation.

The main part of the reducibility technique is a definition of the
family of reducible sets of derivations.  In Martin-L\"of's theory of
iterative inductive definition, this family of sets is defined
inductively by the level of the derivations they contain.  Extending
this definition of reducibility to $\Linc$ is not obvious.  In
particular, in establishing the reducibility of a derivation $\Xi$
ending with a $\coindR$ rule:
$$
\infer[\coindR, p\,\vec x \defnu B\,p\,\vec x] {\Seq \Gamma {p\,\vec
    t}} {\deduce{\Seq \Gamma {S\,\vec t}}{\Pi} & \deduce{\Seq
    {S\,\vec x}{B\,S\,\vec x}}{\Pi_S} }
$$
one must first establish the reducibility of its premise derivations.
But a naive definition of reducibility for $\Xi$, i.e., a definition
that postulates the reducibility of $\Xi$ from the reducibility of its
premises, is not a monotone definition, since the premise derivations
of $\Xi$ may be derivations that have a higher level than $\Xi$.

To define a proper notion of reducibility for the co-inductive cases,
we use a notion of \emph{parametric reducibility}, similar to that used
in the strong normalisation proof of System F~\cite{girard89book}.
The notion of a parameter in our case is essentially a coinductive
predicate. As with strong normalisation of System F, these parameters
are substituted with some ``reducibility candidates'', which in our
case are certain sets of derivations of a co-inductive invariant which
we call \emph{saturated sets}.  Let us say that a derivation $\Psi$ has
type $B$ if its end sequent is of the form $\Seq \Gamma B$, for some
$\Gamma$.  Roughly, a parametric reducibility set of type $C$, under a
parameter substitution $[S/p]$, where $p$ is a co-inductive predicate
and $S$ is an invariant of the same type as $S$, is a certain set of
derivations of type $C[S/p]$ satisfying some closure conditions which
are very similar to the definition of reducibility sets, but without
the co-inductive part.  The definition of reducibility in the case
involving co-induction rules, e.g., as in the derivation $\Xi$ above,
can then be defined in terms parametric reducibility sets, under
appropriate parameter substitutions. Details of the definition will be
given later in this section.



\subsection{Cut reduction}

We follow the idea of Martin-L\"of in using derivations directly as a
measure by defining a well-founded ordering on them. The basis for the
latter relation is a set of reduction rules (called the contraction
rules in \cite{martin-lof71sls}) that are used to eliminate the
applications of the cut rule.  For the cases involving logical
connectives, the cut-reduction rules used to prove the cut-elimination
for $\Linc$ are the same to those of $\FOLDN$. The crucial differences
are in the reduction rules involving induction and co-induction rules.


\begin{definition}
  \label{def:reduct}
  We define a \emph{reduction} relation between derivations.  The redex
  is always a derivation $\Xi$ ending with the multicut rule
  \begin{displaymath}
    \infer[\mc]{\Seq{\Delta_1,\ldots,\Delta_n,\Gamma}{C}}
    {\deduce{\Seq{\Delta_1}{B_1}}
      {\Pi_1}
      & \cdots
      & \deduce{\Seq{\Delta_n}{B_n}}
      {\Pi_n}
      & \deduce{\Seq{B_1,\ldots,B_n,\Gamma}{C}}
      {\Pi}}
    \enspace 
  \end{displaymath}
  We refer to the formulas $B_1,\dots,B_n$ produced by the $\mc$ as
  \emph{cut formulas}.

  If $n=0$, $\Xi$ reduces to the premise derivation $\Pi$.

  For $n > 0$ we specify the reduction relation based on the last rule
  of the premise derivations.  If the rightmost premise derivation
  $\Pi$ ends with a left rule acting on a cut formula $B_i$, then the
  last rule of $\Pi_i$ and the last rule of $\Pi$ together determine
  the reduction rules that apply.  We classify these rules according
  to the following criteria: we call the rule an \emph{essential} case
  when $\Pi_i$ ends with a right rule; if it ends with a left rule, it
  is a \emph{left-commutative} case; if $\Pi_i$ ends with the $\init$
  rule, then we have an \emph{axiom} case; a \emph{multicut} case arises
  when it ends with the $\mc$ rule.  When $\Pi$ does not end with a
  left rule acting on a cut formula, then its last rule is alone
  sufficient to determine the reduction rules that apply.  If $\Pi$
  ends in a rule acting on a formula other than a cut formula, then we
  call this a \emph{right-commutative} case.  A \emph{structural} case
  results when $\Pi$ ends with a contraction or weakening on a cut
  formula.  If $\Pi$ ends with the $\init$ rule, this is also an axiom
  case; similarly a multicut case arises if $\Pi$ ends in the $\mc$
  rule.

  For simplicity of presentation, we always show $i = 1$.

  \paragraph{Essential cases:}

  \begin{description}
  \item[$\landR/\landL$:] If $\Pi_1$ and $\Pi$ are
    \begin{displaymath}
      \infer[\landR]{\Seq{\Delta_1}{B_1' \land B_1''}}
      {\deduce{\Seq{\Delta_1}{B_1'}}
        {\Pi_1'}
	& \deduce{\Seq{\Delta_1}{B_1''}}
        {\Pi_1''}}
      \qquad\qquad\qquad
      \infer[\landL]{\Seq{B_1' \land B_1'',B_2,\ldots,B_n,\Gamma}{C}}
      {\deduce{\Seq{B_1',B_2,\ldots,B_n,\Gamma}{C}}
        {\Pi'}}
      \enspace ,
    \end{displaymath}
    then $\Xi$ reduces to
    \begin{displaymath}
      \infer[\mc]{\Seq{\Delta_1,\ldots,\Delta_n,\Gamma}{C}}
      {\deduce{\Seq{\Delta_1}{B_1'}}
        {\Pi_1'}
	& \deduce{\Seq{\Delta_2}{B_2}}
        {\Pi_2}
	& \cdots
	& \deduce{\Seq{\Delta_n}{B_n}}
        {\Pi_n}
	& \deduce{\Seq{B_1',B_2,\ldots,B_n,\Gamma}{C}}
        {\Pi'}}
      \enspace 
    \end{displaymath}
    The case for the other $\landL$ rule is symmetric.

  \item[$\lorR/\lorL$:] If $\Pi_1$ and $\Pi$ are
    \begin{displaymath}
      \infer[\lorR]{\Seq{\Delta_1}{B_1' \lor B_1''}}
      {\deduce{\Seq{\Delta_1}{B_1'}}
        {\Pi_1'}}
      \qquad\qquad\!\!\!
      \infer[\lorL]{\Seq{B_1' \lor B_1'',B_2,\ldots,B_n,\Gamma}{C}}
      {\deduce{\Seq{B_1',B_2,\ldots,B_n,\Gamma}{C}}
        {\Pi'}
	& \deduce{\Seq{B_1'',B_2,\ldots,B_n,\Gamma}{C}}
        {\Pi''}}
      \enspace ,
    \end{displaymath}
    then $\Xi$ reduces to
    \begin{displaymath}
      \infer[\mc]{\Seq{\Delta_1,\ldots,\Delta_n,\Gamma}{C}}
      {\deduce{\Seq{\Delta_1}{B_1'}}
        {\Pi_1'}
	& \deduce{\Seq{\Delta_2}{B_2}}
        {\Pi_2}
	& \cdots
	& \deduce{\Seq{\Delta_n}{B_n}}
        {\Pi_n}
	& \deduce{\Seq{B_1',B_2,\ldots,B_n,\Gamma}{C}}
        {\Pi'}}
      \enspace 
    \end{displaymath}
    The case for the other $\lorR$ rule is symmetric.

  \item[$\oimpR/\oimpL$:] Suppose $\Pi_1$ and $\Pi$ are
    \begin{displaymath}
      \infer[\oimpR]{\Seq{\Delta_1}{B_1' \oimp B_1''}}
      {\deduce{\Seq{B_1',\Delta_1}{B_1''}}
        {\Pi_1'}}
      \qquad\qquad\!\!
      \infer[\oimpL]{\Seq{B_1' \oimp B_1'',B_2,\ldots,B_n,\Gamma}{C}}
      {\deduce{\Seq{B_2,\ldots,B_n,\Gamma}{B_1'}}
        {\Pi'}
	& \deduce{\Seq{B_1'',B_2,\ldots,B_n,\Gamma}{C}}
        {\Pi''}}
      \enspace 
    \end{displaymath}
    Let $\Xi_1$ be
    \begin{displaymath}
      \infer[\mc]{\Seq{\Delta_1,\ldots,\Delta_n,\Gamma}{B_1''}}
      {\infer[\mc]{\Seq{\Delta_2,\ldots,\Delta_n,\Gamma}{B_1'}}
        {\left\{\raisebox{-1.5ex}{\deduce{\Seq{\Delta_i}{B_i}}
              {\Pi_i}}\right\}_{i \in \{2..n\}}
          & \raisebox{-2.5ex}{\deduce{\Seq{B_2,\ldots,B_n,\Gamma}{B_1'}}
            {\Pi'}}}
	& \deduce{\Seq{B_1',\Delta_1}{B_1''}}
        {\Pi_1'}}
      \enspace 
    \end{displaymath}
    Then $\Xi$ reduces to \settowidth{\infwidthi}
    {$\Seq{\Delta_1,\ldots,\Delta_n,\Gamma,\Delta_2,\ldots,\Delta_n,\Gamma}{C}$}
    \begin{displaymath}
      \infer=[\cL]
      {\Seq{\Delta_1,\ldots,\Delta_n,\Gamma}{C}}
      {
        \infer[\mc]
        {\Seq{\Delta_1,\ldots,\Delta_n,\Gamma, \Delta_2,\ldots,\Delta_n,\Gamma}{C}}
        {
          \raisebox{-2.5ex}{\deduce{\Seq{\ldots}{B_1''}}{\Xi_1}}
          & 
          \left\{\raisebox{-1.5ex}{\deduce{\Seq{\Delta_i}{B_i}}{\Pi_i}}\right\}_{i \in \{2..n\}}
          & \raisebox{-2.5ex}{\deduce{\Seq{B_1'',\{B_i\}_{i \in \{2..n\}},\Gamma}{C}}{\Pi''}}
        }
      }
      \enspace 
    \end{displaymath}
    We use the double horizontal lines to indicate that the relevant
    inference rule (in this case, $\cL$) may need to be applied zero
    or more times.

  \item[$\forallR/\forallL$:] If $\Pi_1$ and $\Pi$ are
    \begin{displaymath}
      \infer[\forallR]{\Seq{\Delta_1}{\forall x.B_1'}}
      {\deduce{\Seq{\Delta_1}{B_1'[y/x]}}
        {\Pi_1'}}
      \qquad\qquad\qquad
      \infer[\forallL]{\Seq{\forall x.B_1',B_2,\ldots,B_n,\Gamma}{C}}
      {\deduce{\Seq{B_1'[t/x],B_2,\ldots,B_n,\Gamma}{C}}
        {\Pi'}}
      \enspace ,
    \end{displaymath}
    then $\Xi$ reduces to
    \begin{displaymath}
      \infer[\mc]{\Seq{\Delta_1,\ldots,\Delta_n,\Gamma}{C}}
      {\raisebox{-2.5ex}{\deduce{\Seq{\Delta_1}{B_1'[t/x]}}
          {\Pi_1'[t/y]}}
	& \left\{\raisebox{-1.5ex}{\deduce{\Seq{\Delta_i}{B_i}}
            {\Pi_i}}\right\}_{i \in \{2..n\}}
	& \raisebox{-2.5ex}{\deduce{\Seq{\ldots}{C}}
          {\Pi'}}}
      \enspace 
    \end{displaymath}

  \item[$\existsR/\existsL$:] If $\Pi_1$ and $\Pi$ are
    \begin{displaymath}
      \infer[\existsR]{\Seq{\Delta_1}{\exists x.B_1'}}
      {\deduce{\Seq{\Delta_1}{B_1'[t/x]}}
        {\Pi_1'}}
      \qquad\qquad\qquad
      \infer[\existsL]{\Seq{\exists x.B_1',B_2,\ldots,
          B_n,\Gamma}{C}}
      {\deduce{\Seq{B_1'[y/x],B_2,\ldots,B_n,
	    \Gamma}{C}}
        {\Pi'}}
      \enspace ,
    \end{displaymath}
    then $\Xi$ reduces to
    \begin{displaymath}
      \infer[\mc]{\Seq{\Delta_1,\ldots,\Delta_n,\Gamma}{C}}
      {\deduce{\Seq{\Delta_1}{B_1'[t/x]}}
        {\Pi_1'}
        & \ldots
	& \deduce{\Seq{B_1'[t/x], B_2,\dots,\Gamma}{C}}
        {\Pi'[t/y]}
      }
      \enspace 
    \end{displaymath}

  \item[$*/\indL$:] Suppose $\Pi$ is the derivation
$$
\infer[\indL] {\Seq{p\,\vec{t}, B_2,\dots,B_n,\Gamma}{C}} {
  \deduce{\Seq{D\,S\,\vec{x}}{S\,\vec{x}}}{\Pi_S} &
  \deduce{\Seq{S\,\vec{t}, B_2,\dots,B_n, \Gamma} {C}}{\Pi'} }
$$
where $p\,\vec{x} \defmu B\,p\,\vec{x}$.  Then $\Xi$ reduces to
$$
\infer[\mc] {\Seq{\Delta_1,\dots,\Delta_n,\Gamma}{C}} {
  \deduce{\Seq{\Delta_1}{S\,\vec{t}}}{\mu_{p\,\vec t}^p(\Pi_1,\Pi_S)}
  & \ldots & \deduce{\Seq{S\,\vec{t},\dots,B_n,\Gamma}{C}}{\Pi'} }
$$

\item[$\coindR/\coindL$:] Suppose $\Pi_1$ and $\Pi$ are
$$
\infer[\coindR] {\Seq{\Delta_1}{p\,\vec{t}}} {
  \deduce{\Seq{\Delta_1}{S\,\vec{t}}}{\Pi_1'} &
  \deduce{\Seq{S\,\vec{x}}{D\,S\,\vec{x}}}{\Pi_S} } \qquad \qquad
\infer[\coindL] {\Seq{p\,\vec{t}, \dots, \Gamma}{C}}
{\deduce{\Seq{D\,p\,\vec{t},\dots, \Gamma}{C}}{\Pi'}}
$$
Let $\Xi_1$ be the derivation
$$
\infer[\mc] {\Seq{\Delta_1}{D\,S\,\vec{t}}} {
  \deduce{\Seq{\Delta_1}{S\,\vec{t}}}{\Pi_1'} &
  \deduce{\Seq{S\,\vec{t}} {D\,S\,\vec{t}}}{\Pi_S[\vec t/ \vec x]} }
$$
Then $\Xi$ reduces to
$$
\infer[\mc] {\Seq{\Delta_1,\dots,\Delta_n,\Gamma}{C}} {
  \raisebox{-1.5ex}{\deduce{\Seq{\Delta_1}
      {D\,p\,\vec{t}}}{\nu_{D\,p}^p(\Xi_1,\Pi_S)}} & {
    \left\{\raisebox{-1.5ex} {\deduce{\Seq{\Delta_j}{B_j}}{\Pi_j}}
    \right\}_{j \in \{2,\dots,n \}} } & \raisebox{-1.5ex}{
    \deduce{\Seq{D\,p\,\vec{t},\dots,\Gamma }{C}}{\Pi'}} }
$$


\item[$\eqR/\eqL$:] Suppose $\Pi_1$ and $\Pi$ are
  \begin{displaymath}
    \infer[\eqR]{\Seq{\Delta_1}{s = t}}
    {}
    \qquad\qquad\qquad
    \infer[\eqL]{\Seq{s=t,B_2,\ldots,B_n,\Gamma}{C}}
    {\left\{\raisebox{-1.5ex}
        {\deduce{\Seq{B_2\rho,\ldots,B_n\rho,\Gamma\rho}
            {C\rho}}
          {\Pi^\rho}}
      \right\}_\rho}
    \enspace 
  \end{displaymath}
  Then by the definition of $\eqR$ rule, $s$ and $t$ are equal terms
  (modulo $\lambda$-conversion), and hence are unifiable by the empty
  substitution.  Note that in this case $\Pi^\epsilon \in \{\Pi^\rho
  \}_\rho$.  Therefore $\Xi$ reduces to
  \begin{displaymath}
    \infer=[\wL]{\Seq{\Delta_1,\Delta_2,\ldots,\Delta_n,\Gamma}{C}}
    {
      \infer[\mc]
      {\Seq{\Delta_2,\ldots,\Delta_n,\Gamma}{C}}
      {	  
        \left\{\raisebox{-1.5ex}
          {\deduce{\Seq{\Delta_i}{B_i}}
            {\Pi_i}}
        \right\}_{i \in \{2..n\}}
        & \raisebox{-2.5ex}
        {\deduce{\Seq{B_2,\ldots,B_n,\Gamma}{C}}
          {\Pi^{\epsilon}}}    
      }
    }
    \enspace 
  \end{displaymath}

\end{description}

\paragraph{Left-commutative cases:}

In the following cases, we suppose that $\Pi$ ends with a left rule,
other than $\{\cL, \wL,\indL\}$, acting on $B_1$.

\begin{description}

\item[$\bulletL/\circL$:] Suppose $\Pi_1$ is
  \begin{displaymath}
    \infer[\bulletL]{\Seq{\Delta_1}{B_1}}
    {\left\{\raisebox{-1.5ex}{\deduce{\Seq{\Delta_1^i}{B_1}}
          {\Pi_1^i}}\right\}}
    \enspace ,
  \end{displaymath}
  where $\bulletL$ is any left rule except $\oimpL$, $\eqL$, or
  $\indL$.  Then $\Xi$ reduces to \settowidth{\infwidthi}
  {$\left\{\raisebox{-3.5ex}{\infer[\mc]{\Seq{\Delta_1^i,\Delta_2,\ldots,\Delta_n,\Gamma}{C}}
        {\raisebox{-2.5ex}{\deduce{\Seq{\Delta_1^i}{B_1}} {\Pi_1^i}} &
          \left\{\raisebox{-1.5ex}{\deduce{\Seq{\Delta_j}{B_j}}
              {\Pi_j}}\right\}_{j \in \{2..n\}} &
          \raisebox{-2.5ex}{\deduce{\Seq{B_1,\ldots,B_n,\Gamma}{C}}
            {\Pi}}}}\right\}$} \settowidth{\infwidthii}{\mc}
  \begin{displaymath}
    \infer[\bulletL]{\Seq{\Delta_1,\Delta_2,\ldots,\Delta_n,\Gamma}{C}}
    {\deduce{\makebox[\infwidthi]{}}
      {\left\{\raisebox{-3.5ex}{\infer[\mc]{\Seq{\Delta_1^i,\Delta_2,\ldots,\Delta_n,\Gamma}{C}}
            {\raisebox{-2.5ex}{\deduce{\Seq{\Delta_1^i}{B_1}}
                {\Pi_1^i}}
              & \left\{\raisebox{-1.5ex}{\deduce{\Seq{\Delta_j}{B_j}}
                  {\Pi_j}}\right\}_{j \in \{2..n\}}
              & \raisebox{-2.5ex}{\deduce{\Seq{B_1,\ldots,B_n,\Gamma}{C}}
                {\Pi}}}}\right\}\makebox[\infwidthii]{}}}
    \enspace 
  \end{displaymath}

\item[$\oimpL/\circL$:] Suppose $\Pi_1$ is
  \begin{displaymath}
    \infer[\oimpL]{\Seq{D_1' \oimp D_1'',\Delta_1'}{B_1}}
    {\deduce{\Seq{\Delta_1'}{D_1'}}
      {\Pi_1'}
      & \deduce{\Seq{D_1'',\Delta_1'}{B_1}}
      {\Pi_1''}}
    \enspace 
  \end{displaymath}
  Let $\Xi_1$ be
  \begin{displaymath}
    \infer[\mc]{\Seq{D_1'',\Delta_1',\Delta_2,\ldots,\Delta_n,\Gamma}{C}}
    {\deduce{\Seq{D_1'',\Delta_1'}{B_1}}
      {\Pi_1''}
      & \deduce{\Seq{\Delta_2}{B_2}}
      {\Pi_2}
      & \cdots
      & \deduce{\Seq{\Delta_n}{B_n}}
      {\Pi_n}
      & \deduce{\Seq{B_1,\ldots,B_n,\Gamma}{C}}
      {\Pi}}
    \enspace 
  \end{displaymath}
  Then $\Xi$ reduces to
  \begin{displaymath}
    \infer[\oimpL]
    {\Seq{D_1' \oimp D_1'',\Delta_1',\Delta_2,\ldots,\Delta_n,\Gamma}{C}}
    {
      \infer=[\wL]
      {\Seq{\Delta_1',\Delta_2,\ldots,\Delta_n,\Gamma}{D_1'}}
      {\deduce{\Seq{\Delta_1'}{D_1'}}{\Pi_1'} }
      & 
      \deduce{\Seq{D_1'',\Delta_1',\Delta_2,\ldots,\Delta_n,\Gamma}{C}}{\Xi_1}
    }
    \enspace 
  \end{displaymath}

\item[$\indL/\circL$:] Suppose $\Pi_1$ is
$$
\infer[\indL] {\Seq{p\,\vec{t}, \Delta_1'}{B_1}} {
  \deduce{\Seq{D\,S\,\vec{x}}{S\,\vec{x}}}{\Pi_S} &
  \deduce{\Seq{S\,\vec{t}, \Delta_1'}{B_1}}{\Pi_1'} }
$$      
where $p\,\vec{x} \defmu D\,p\,\vec{x}$.  Let $\Xi_1$ be
$$
\infer[\mc] {\Seq{S\,\vec{t}, \Delta_1',
    \Delta_2,\dots,\Delta_n,\Gamma}{C}} {
  \deduce{\Seq{S\,\vec{t},\Delta_1'}{B_1}}{\Pi_1'} & \ldots &
  \deduce{\Seq{\Delta_n}{B_n}}{\Pi_n} &
  \deduce{\Seq{B_1,\dots,B_n,\Gamma}{C}}{\Pi} } \enspace
$$      
Then $\Xi$ reduces to
$$
\infer[\indL] {\Seq{p\,\vec{t}, \Delta_1',\dots,\Delta_n}{C}} {
  \deduce{\Seq{D\,S\,\vec{x}}{S\,\vec{x}}}{\Pi_S} &
  \deduce{\Seq{S\,\vec{t}, \Delta_1',\dots,\Delta_n,\Gamma}{C}}{\Xi_1}
}
$$

\item[$\eqL/\circL$:] Suppose $\Pi_1$ is
  \begin{displaymath}
    \infer[\eqL]{\Seq{s=t,\Delta_1'}{B_1}}
    {\left\{\raisebox{-1.5ex}
        {\deduce{\Seq{\Delta_1'\rho}{B_1\rho}}
          {\Pi_1^{\rho}}}
      \right\}}
    \enspace ,
  \end{displaymath}
  then $\Xi$ reduces to \settowidth{\infwidthi}
  {$\left\{\raisebox{-3.5ex}{\infer[\mc]
        {\Seq{\Delta_1'\rho,\Delta_2\rho,\ldots,
            \Delta_n\rho,\Gamma\rho} {C\rho}}
        {\raisebox{-2.5ex}{\deduce{\Seq{\Delta_1'\rho}{B_1\rho}}
            {\Pi_1^{\rho}}} & \left\{\raisebox{-1.5ex}{
              \deduce{\Seq{\Delta_i\rho}{B_i\rho}}
              {\Pi_i\rho}}\right\}_{i \in \{2..n\}} &
          \raisebox{-2.5ex}{\deduce{\Seq{\ldots}{C\rho}}
            {\Pi\rho}}}}\right\}$} \settowidth{\infwidthii}{\mc}
  \begin{displaymath}
    \infer[\eqL]{\Seq{s=t,\Delta_1',\Delta_2,\ldots,\Delta_n,\Gamma}{C}}
    {\deduce{\makebox[\infwidthi]{}}
      {\left\{\raisebox{-3.5ex}{\infer[\mc]
            {\Seq{\Delta_1'\rho,\Delta_2\rho,\ldots,
                \Delta_n\rho,\Gamma\rho}
              {C\rho}}
            {\raisebox{-2.5ex}{
                \deduce{\Seq{\Delta_1'\rho}{B_1\rho}}
                {\Pi_1^{\rho}}}
              & \left\{\raisebox{-1.5ex}{
                  \deduce{\Seq{\Delta_i\rho}{B_i\rho}}
                  {\Pi_i\rho}}\right\}_{i \in \{2..n\}}
              & \raisebox{-2.5ex}{
                \deduce{\Seq{\ldots}{C\rho}}
                {\Pi\rho}}}}\right\}\makebox[\infwidthii]{}}}
    \enspace 
  \end{displaymath}
\end{description}

\paragraph{Right-commutative cases:}

\begin{description}

\item[$-/\circL$:] Suppose $\Pi$ is
  \begin{displaymath}
    \infer[\circL]{\Seq{B_1,\ldots,B_n,\Gamma}{C}}
    {\left\{\raisebox{-1.5ex}{\deduce{\Seq{B_1,\ldots,B_n,\Gamma^i}{C}}
          {\Pi^i}}\right\}}
    \enspace ,
  \end{displaymath}
  where $\circL$ is any left rule other than $\oimpL$, $\eqL$, or
  $\indL$ acting on a formula other than $B_1, \ldots, B_n$.
  The derivation $\Xi$ reduces to \settowidth{\infwidthi}
  {$\left\{\raisebox{-2.45ex}{
        \infer[\mc]{\Seq{\Delta_1,\ldots,\Delta_n,\Gamma^i}{C}}
        {\deduce{\Seq{\Delta_1}{B_1}} {\Pi_1} & \cdots &
          \deduce{\Seq{\Delta_n}{B_n}} {\Pi_n} &
          \deduce{\Seq{B_1,\ldots,B_n,\Gamma^i}{C}}
          {\Pi^i}}}\right\}$} \settowidth{\infwidthii}{\mc}
  \begin{displaymath}
    \infer[\circL]{\Seq{\Delta_1,\ldots,\Delta_n,\Gamma}{C}}
    {\deduce{\makebox[\infwidthi]{}}
      {\left\{\raisebox{-2.45ex}{
            \infer[\mc]{\Seq{\Delta_1,\ldots,\Delta_n,\Gamma^i}{C}}		 
            {\deduce{\Seq{\Delta_1}{B_1}}
              {\Pi_1}
              & \cdots
              & \deduce{\Seq{\Delta_n}{B_n'}}
              {\Pi_n}
              & \deduce{\Seq{B_1,\ldots,B_n,\Gamma^i}{C}}
              {\Pi^i}}}\right\}\makebox[\infwidthii]{}}}
    \enspace 
  \end{displaymath}

\item[$-/\oimpL$:] Suppose $\Pi$ is
  \begin{displaymath}
    \infer[\oimpL]{\Seq{B_1,\ldots,B_n, D' \oimp D'',\Gamma'}{C}}
    {\deduce{\Seq{B_1,\ldots,B_n,\Gamma'}{D'}}
      {\Pi'}
      & \deduce{\Seq{B_1,\ldots,B_n,D'',\Gamma'}{C}}
      {\Pi''}}
    \enspace 
  \end{displaymath}
  Let $\Xi_1$ be
  \begin{displaymath}
    \infer[\mc]{\Seq{\Delta_1,\ldots,\Delta_n,\Gamma'}{D'}}
    {\deduce{\Seq{\Delta_1}{B_1}}
      {\Pi_1}
      & \cdots
      & \deduce{\Seq{\Delta_n}{B_n}}
      {\Pi_n}
      & \deduce{\Seq{B_1,\ldots,B_n,\Gamma'}{D'}}
      {\Pi'}}
  \end{displaymath}
  and $\Xi_2$ be
  \begin{displaymath}
    \infer[\mc]{\Seq{\Delta_1,\ldots,\Delta_n,D'',\Gamma'}{C}}
    {\deduce{\Seq{\Delta_1}{B_1}}
      {\Pi_1}
      & \cdots
      & \deduce{\Seq{\Delta_n}{B_n}}
      {\Pi_n}
      & \deduce{\Seq{B_1,\ldots,B_n,D'',\Gamma'}{C}}
      {\Pi''}}
    \enspace
  \end{displaymath}
  Then $\Xi$ reduces to
  \begin{displaymath}
    \infer[\oimpL]{\Seq{\Delta_1,\ldots,\Delta_n,D' \oimp D'',\Gamma'}{C}}
    {\deduce{\Seq{\Delta_1,\ldots,\Delta_n,\Gamma'}{D'}}
      {\Xi_1}
      & \deduce{\Seq{\Delta_1,\ldots,\Delta_n,D'',\Gamma'}{C}}
      {\Xi_2}}
    \enspace 
  \end{displaymath}

\item[$-/\indL$:] Suppose $\Pi$ is
$$
\infer[\indL] {\Seq{B_1,\dots,B_n, p\,\vec{t},\Gamma'} {C}} {
  \deduce{\Seq{D\,S\,\vec{x}}{S\,\vec{x}}}{\Pi_S} &
  \deduce{\Seq{B_1,\dots,B_n, S\,\vec{t}, \Gamma'}{C}}{\Pi'} }
\enspace ,
$$      
where $p\,\vec{x} \defmu D\,p\,\vec{x}$.  Let $\Xi_1$ be
\begin{displaymath}
  \infer[\mc]{\Seq{\Delta_1,\ldots,\Delta_n, S\,\vec{t}, \Gamma'}{C}}
  {\deduce{\Seq{\Delta_1}{B_1}}
    {\Pi_1}
    & \cdots
    & \deduce{\Seq{\Delta_n}{B_n}}
    {\Pi_n}
    & \deduce{\Seq{B_1,\ldots,B_n,
        S\,\vec{t},\Gamma'}{C}}
    {\Pi'}}
  \enspace 
\end{displaymath}
Then $\Xi$ reduces to
$$
\infer[\indL] {\Seq{\Delta_1,\dots,\Delta_n, p\,\vec{t},\Gamma'}{C}} {
  \deduce{\Seq{D\,S\,\vec{x}}{S\,\vec{x}}}{\Pi_S} &
  \deduce{\Seq{\Delta_1,\dots,\Delta_n, S\,\vec{t}, \Gamma'}{C}}{\Xi}
} \enspace 
$$

\item[$-/\eqL$:] If $\Pi$ is
  \begin{displaymath}
    \infer[\eqL]{\Seq{B_1,\ldots,B_n,s=t,\Gamma'}{C}}
    {\left\{\raisebox{-1.5ex}
        {\deduce{\Seq{B_1\rho,\ldots, B_n\rho,\Gamma'\rho}{C\rho}}
          {\Pi^{\rho}}}\right\}}
    \enspace ,
  \end{displaymath}
  then $\Xi$ reduces to \settowidth{\infwidthi}
  {$\left\{\raisebox{-3.5ex}{\infer[\mc] {\Seq{\Delta_1\rho,\ldots,
            \Delta_n\rho,D\sigma,\Gamma'\rho} {C\rho}}
        {\left\{\raisebox{-1.5ex}
            {\deduce{\Seq{\Delta_i\rho}{B_i\rho}} {\Pi_i\rho}}
          \right\}_{i \in \{1..n\}} & \raisebox{-2.5ex}
          {\deduce{\Seq{\{B_i\rho\}_{i\in\{1..n\}},
                D\sigma,\Gamma'\rho} {C\rho}} {\Pi^{\rho,\sigma,D}}}}}
    \right\}$} \settowidth{\infwidthii}{\mc}
  \begin{displaymath}
    \infer[\eqL]{\Seq{\Delta_1,\ldots,\Delta_n,s=t,\Gamma'}{C}}
    {\deduce{\makebox[\infwidthi]{}}
      {\left\{\raisebox{-3.5ex}{\infer[\mc]
	    {\Seq{\Delta_1\rho,\ldots,\Delta_n\rho,\Gamma'\rho}{C\rho}}
            {\left\{\raisebox{-1.5ex}
                {\deduce{\Seq{\Delta_i\rho}{B_i\rho}}
                  {\Pi_i\rho}}
              \right\}_{i \in \{1..n\}}
              & \raisebox{-2.5ex}
              {\deduce{\Seq{B_i\rho, \ldots, \Gamma'\rho}{C\rho}}
                {\Pi^{\rho}}}}}
	\right\}\makebox[\infwidthii]{}}}
    \enspace 
  \end{displaymath}

\item[$-/\circR$:] If $\Pi$ is
  \begin{displaymath}
    \infer[\circR]{\Seq{B_1,\ldots,B_n,\Gamma}{C}}
    {\left\{\raisebox{-1.5ex}{
          \deduce{\Seq{B_1,\ldots,B_n,\Gamma^i}{C^i}}
          {\Pi^i}}\right\}}
    \enspace ,
  \end{displaymath}
  where $\circR$ is any right rule except $\coindR$, then $\Xi$
  reduces to \settowidth{\infwidthi} {$\left\{\raisebox{-2.45ex}{
        \infer[\mc]{\Seq{\Delta_1,\ldots,\Delta_n,\Gamma^i}{C^i}}
        {\deduce{\Seq{\Delta_1}{B_1}} {\Pi_1} & \cdots &
          \deduce{\Seq{\Delta_n}{B_n}} {\Pi_n} &
          \deduce{\Seq{B_1,\ldots,B_n,\Gamma^i}{C^i}}
          {\Pi^i}}}\right\}$} \settowidth{\infwidthii}{\mc}
  \begin{displaymath}
    \infer[\circR]{\Seq{\Delta_1,\ldots,\Delta_n,\Gamma}{C}}
    {\deduce{\makebox[\infwidthi]{}}
      {\left\{\raisebox{-2.45ex}{
	    \infer[\mc]{\Seq{\Delta_1,\ldots,\Delta_n,\Gamma^i}{C^i}}
            {\deduce{
                \Seq{\Delta_1}{B_1}}{\Pi_1}
              & \cdots
              & \deduce{\Seq{\Delta_n}{B_n}}{\Pi_n}
              & \deduce{\Seq{B_1,\ldots,B_n,\Gamma^i}{C^i}}
              {\Pi^i}}}\right\}\makebox[\infwidthii]{}}}
    \enspace ,
  \end{displaymath}

\item[$-/\coindR$:] Suppose $\Pi$ is
$$
\infer[\coindR] {\Seq{B_1,\dots,B_n,\Gamma}{p\,\vec{t}}} {
  \deduce{\Seq{B_1,\dots,B_n,\Gamma} {S\,\vec{t}}}{\Pi'} &
  \deduce{\Seq{S\,\vec{x}}{D\,S\,\vec{x}}}{\Pi_S} } \enspace ,
$$
where $p\,\vec{x} \defnu D\,p\,\vec{x}$.  Let $\Xi_1$ be
\begin{displaymath}
  \infer[\mc]{\Seq{\Delta_1,\ldots,\Delta_n,\Gamma}
    {S\,\vec{t}}}
  {\deduce{\Seq{\Delta_1}{B_1}}{\Pi_1}
    & \cdots
    & \deduce{\Seq{\Delta_n}{B_n}}{\Pi_n}
    & \deduce{\Seq{B_1,\ldots,B_n,\Gamma}{S\,\vec{t}}}
    {\Pi'}}
  \enspace 
\end{displaymath}
Then $\Xi$ reduces to
$$
\infer[\coindR] {\Seq{\Delta_1,\dots,\Delta_n,\Gamma}{p\,\vec{t}}} {
  \deduce{\Seq{\Delta_1,\dots,\Delta_n,\Gamma} {S\,\vec{t}}}{\Xi_1} &
  \deduce{\Seq{S\,\vec{x}}{D\,S\,\vec{x}}}{\Pi_S} } \enspace
$$
\end{description}

\paragraph{Multicut cases:}

\begin{description}

\item[$\mc/\circL$:] If $\Pi$ ends with a left rule, other than $\cL$,
  $\wL$ and $\indL$, acting on $B_1$ and $\Pi_1$ ends with a multicut
  and reduces to $\Pi_1'$, then $\Xi$ reduces to
  \begin{displaymath}
    \infer[\mc]{\Seq{\Delta_1,\ldots,\Delta_n,\Gamma}{C}}
    {\deduce{\Seq{\Delta_1}{B_1}}
      {\Pi_1'}
      & \deduce{\Seq{\Delta_2}{B_2}}
      {\Pi_2}
      & \cdots
      & \deduce{\Seq{\Delta_n}{B_n}}
      {\Pi_n}
      & \deduce{\Seq{B_1,\ldots,B_n,\Gamma}{C}}
      {\Pi}}
    \enspace 
  \end{displaymath}

\item[$-/\mc$:] Suppose $\Pi$ is
  \begin{displaymath}
    \infer[\mc]{\Seq{B_1,\ldots,B_n,\Gamma^1,\ldots,\Gamma^m,\Gamma'}{C}}
    {\left\{\raisebox{-1.5ex}{\deduce{\Seq{\{B_i\}_{i \in I^j},\Gamma^j}{D^j}}
          {\Pi^j}}\right\}_{j \in \{1..m\}}
      & \raisebox{-2.5ex}{\deduce{\Seq{\{D^j\}_{j \in \{1..m\}},\{B_i\}_{i \in I'},\Gamma'}{C}}
        {\Pi'}}}
    \enspace ,
  \end{displaymath}
  where $I^1,\ldots,I^m,I'$ partition the formulas $\{B_i\}_{i \in
    \{1..n\}}$ among the premise derivations $\Pi_1$, \ldots,
  $\Pi_m$,$\Pi'$.  For $1 \leq j \leq m$ let $\Xi^j$ be
  \begin{displaymath}
    \infer[\mc]{\Seq{\{\Delta_i\}_{i \in I^j},\Gamma^j}{D^j}}
    {\left\{\raisebox{-1.5ex}{\deduce{\Seq{\Delta_i}{B_i}}
          {\Pi_i}}\right\}_{i \in I^j}
      & \raisebox{-2.5ex}{\deduce{\Seq{\{B_i\}_{i \in I^j},\Gamma^j}{D^j}}
        {\Pi^j}}}
    \enspace 
  \end{displaymath}
  Then $\Xi$ reduces to
  \begin{displaymath}
    \infer[\mc]{\Seq{\Delta_1,\ldots,\Delta_n,\Gamma^1,\ldots\Gamma^m,\Gamma'}{C}}
    {\left\{\raisebox{-1.5ex}{\deduce{\Seq{\ldots}{D^j}}
          {\Xi^j}}\right\}_{j \in \{1..m\}}
      & \left\{\raisebox{-1.5ex}{\deduce{\Seq{\Delta_i}{B_i}}
          {\Pi_i}}\right\}_{i \in I'}
      & \raisebox{-2.5ex}{\deduce{\Seq{\ldots}{C}}
        {\Pi'}}}
    \enspace 
  \end{displaymath}

\end{description}

\paragraph{Structural cases:}

\begin{description}
\item[$-/\cL$:] If $\Pi$ is
  \begin{displaymath}
    \infer[\cL]{\Seq{B_1,B_2,\ldots,B_n,\Gamma}{C}}
    {\deduce{\Seq{B_1,B_1,B_2,\ldots,B_n,\Gamma}{C}}
      {\Pi'}}
    \enspace ,
  \end{displaymath}
  then $\Xi$ reduces to
  \begin{displaymath}
    \infer=[\cL]
    {\Seq{\Delta_1,\Delta_2,\ldots,\Delta_n,\Gamma}{C}}
    {
      \infer[\mc]
      {\Seq{\Delta_1,\Delta_1,\Delta_2,\ldots,\Delta_n,\Delta_n,\Gamma}{C}}
      {
        \raisebox{-2.5ex}{\deduce{\Seq{\Delta_1}{B_1}}{\Pi_1}}
        & 
        \left\{\raisebox{-1.5ex}{\deduce{\Seq{\Delta_i}{B_i}}{\Pi_i}}\right\}_{i \in \{1..n\}}
        & 
        \raisebox{-2.5ex}{\deduce{\Seq{B_1,B_1,B_2,\ldots,B_n,\Gamma}{C}}{\Pi'}}
      }
    }
    \enspace 
  \end{displaymath}

\item[$-/\wL$:] If $\Pi$ is
  \begin{displaymath}
    \infer[\wL]{\Seq{B_1,B_2,\ldots,B_n,\Gamma}{C}}
    {\deduce{\Seq{B_2,\ldots,B_n,\Gamma}{C}}
      {\Pi'}}
    \enspace ,
  \end{displaymath}
  then $\Xi$ reduces to
  \begin{displaymath}
    \infer=[\wL]
    {\Seq{\Delta_1,\Delta_2,\ldots,\Delta_n,\Gamma}{C}}
    {
      \infer[\mc]
      {\Seq{\Delta_2,\ldots,\Delta_n,\Gamma}{C}}
      {
        \deduce{\Seq{\Delta_2}{B_2}}{\Pi_2}
        & \ldots
        & \deduce{\Seq{\Delta_n}{B_n}}{\Pi_n}			
        & \deduce{\Seq{B_2,\ldots,B_n,\Gamma}{C}}{\Pi'}
      }
    }
    \enspace 
  \end{displaymath}

\end{description}

\paragraph{Axiom cases:}

\begin{description}
\item[$\init/\circL$:] Suppose $\Pi$ ends with a left-rule acting on
  $B_1$ and $\Pi_1$ ends with the $\init$ rule. Then it must be the
  case that $\Delta_1 = \{B_1\}$ and $\Xi$ reduces to
  \begin{displaymath}
    \infer[\mc]{\Seq{B_1,\Delta_2,\ldots,\Delta_n,\Gamma}{C}}
    {\deduce{\Seq{\Delta_2}{B_2}}
      {\Pi_2}
      & \cdots
      & \deduce{\Seq{\Delta_n}{B_n}}
      {\Pi_n}
      & \deduce{\Seq{B_1,B_2,\ldots,B_n,\Gamma}{C}}
      {\Pi}}		
    \enspace 
  \end{displaymath}

\item[$-/\init$:] If $\Pi$ ends with the $\init$ rule, then $n = 1$,
  $\Gamma$ is the empty multiset, and $C$ must be a cut formula, i.e.,
  $C = B_1$. Therefore $\Xi$ reduces to $\Pi_1$.
\end{description}

\end{definition}

Notice that the reductions in the essential case for induction and
co-induction are not symmetric. This is because we use an asymmetric
measure to show the termination of cut-reduction, that is, the
complexity of cut is always reduced on the right premise.  The
difficulty in getting a symmetric measure, in the presence of
contraction and implication (in the body of definition), is already
observed in logics with definitions but without (co-)induction
\cite{schroeder-heister92nlip}.

It is clear from an inspection of the rules of the logic and the
definition of cut reduction that every derivation ending with a
multicut has a reduct. But because we use multisets in sequents, there
may be some ambiguity as to whether a formula occurring on the left
side of the rightmost premise of a multicut rule is in fact a cut
formula, and if so, which of the left premises corresponds to it.  As
a result, several of the reduction rules may apply, and so a
derivation may have multiple redexes.

The following lemmas show that the reduction relation is preserved by
some of the transformations of derivations defined previously.

\begin{lemma}
  \label{lm:reduct subst}
  Let $\Pi$ be a derivation of $\Seq{\Gamma}{C}$ ending with a $\mc$
  and let $\theta$ be a substitution.  If $\Pi\theta$ reduces to $\Xi$
  then there exists a derivation $\Pi'$ such that $\Xi = \Pi'\theta$
  and $\Pi$ reduces to $\Pi'$.
\end{lemma}
\begin{proof}
  Observe that the redexes of a derivation are not affected by
  substitution, since the cut reduction rules are determined by the
  last rules of the premise derivations of the derivation, which are
  not changed by substitution. Therefore, any cut reduction rule that
  is applied to $\Pi\theta$ to get $\Xi$ can also be applied to
  $\Pi$. Suppose that $\Pi'$ is the reduct of $\Pi$ obtained this
  way. In all cases, except for the cases where the reduction rule
  applied is either $*/\indL$ or $\coindL/ \coindR$, it is a matter of
  routine to check that $\Pi'\theta = \Xi$. For the reduction rules
  $*/\indL$ and $\coindL/ \coindR$, we need Lemma~\ref{lm:ind unfold
    subst} and Lemma~\ref{lm:coind unfold subst} which show that
  substitution commutes with (co-)inductive unfolding. \qed
\end{proof}

\begin{lemma}
  \label{lm:unfold-reduct-ind}
  Let $p\,\vec{x} \defmu D\,p\,\vec{x}$ be an inductive definition and
  let $\Pi_S$ be a derivation of $\Seq{D\,S\,\vec{x}}{S\,\vec{x}}$ for
  some invariant $S$. Let $C$ be a non-atomic formula dominated by
  $p$.  Let $\Pi$ and $\Pi'$ be two derivations of the same sequent
  $\Seq{\Gamma}{C}$, and $\Pi$ ends with an $\mc$-rule.  If
  $\mu_C^p(\Pi,\Pi_S)$ reduces to $\Xi$ then there exists a derivation
  $\Pi'$ such that $\Xi = \mu_C^p(\Pi',\Pi_S)$ and $\Pi$ reduces to
  $\Pi'$.
\end{lemma}
\begin{proof}
  By case analysis on the reduction rules.  The case analysis can be
  much simplified by the following observations.  First, the reduction
  rules are driven only by outermost connectives in the formulas in
  the sequent. Second, the unfolding of a derivation affects only the
  right-hand-side of the sequents appearing in the derivation (or more
  specifically, only the branches containing major premises).  By a
  quick inspection on the definition of reduction rules in
  Definition~\ref{def:reduct}, we see that the only non-trivial case
  to consider is the right-commutative case $-/\circR$.  Since $C$ is
  non-atomic (and assuming that it has at least one occurrence of $p$,
  otherwise it is trivial since $\Pi = \mu_C^p(\Pi,\Pi_S)$ in this
  case), the only cases we need to verify is when its topmost logical
  connective is either $\land$, $\lor$, $\oimp$, $\forall$ and
  $\exists$.  In these cases, the unfolding does not change the
  topmost connective, therefore any reduction rule that applies to
  $\mu(\Pi,\Pi_S)$ also applies to $\Pi$.  Lemma~\ref{lm:ind unfold
    subst} and Lemma~\ref{lm:coind unfold subst} are used when
  substitutions are involved (right/left commutative cases with
  $\eqL$).  \qed
\end{proof}

\begin{lemma}
  \label{lm:unfold-reduct-ind2}
  Let $p\,\vec{x} \defmu D\,p\,\vec{x}$ be an inductive definition and
  let $\Pi_S$ be a derivation of $\Seq{D\,S\,\vec{x}}{S\,\vec{x}}$ for
  some invariant $S$. Let $\Pi$ be the derivation
$$
\infer[mc] {\Seq{\Delta_1, \ldots, \Delta_n, \Gamma}{p\,\vec t}} {
  \deduce{\Seq{\Delta_1}{B_1}}{\Pi_1} & \cdots &
  \deduce{\Seq{\Delta_n}{B_n}}{\Pi_n} &
  \deduce{\Seq{B_1,\ldots,B_n,\Gamma}{p\,\vec t}}{\Pi'} } \enspace
$$
Suppose that $\Pi'$ ends with a rule other than $\init$ and $\indR$.
If $\mu_{p\,\vec t}^p(\Pi,\Pi_S)$ reduces to $\Xi$ then there exists a
derivation $\Pi''$ such that $\Xi = \mu_{p\,\vec t}^p(\Pi'', \Pi_S)$
and $\Pi$ reduces to $\Pi''$.
\end{lemma}
\begin{proof}
  The proof is straightforward by inspection on the cut reduction
  rules and the definition of inductive unfolding. \qed
\end{proof}

\begin{lemma}
  \label{lm:unfold-reduct-coind}
  Let $p\,\vec{x} \defnu D\,p\,\vec{x}$ be a co-inductive definition
  and let $\Pi_S$ be a derivation of $\Seq{S\,\vec{x}}{D\,S\,\vec{x}}$
  for some invariant $S$. Let $C$ be a non-atomic formula dominated by
  $p$.  Let $\Pi$ and $\Pi'$ be two derivations of the sequent
  $\Seq{\Gamma}{C[S/p]}$, where $\Pi$ ends with a $\mc$ rule.  If
  $\nu_C^p(\Pi,\Pi_S)$ reduces to $\Xi$ then there exists a derivation
  $\Pi'$ such that $\Xi = \nu_C^p(\Pi',\Pi_S)$ and $\Pi$ reduces to
  $\Pi'$.
\end{lemma}
\begin{proof}
  Analogous to the proof of Lemma~\ref{lm:unfold-reduct-ind}. \qed
\end{proof}

\subsection{Normalizability}

\begin{definition}
  \label{def:norm}
  We define the set of \emph{normalizable} derivations to be the
  smallest set that satisfies the following conditions:
  \begin{enumerate}
  \item If a derivation $\Pi$ ends with a multicut, then it is
    normalizable if every reduct of $\Pi$ is normalizable.

  \item If a derivation ends with any rule other than a multicut, then
    it is normalizable if the premise derivations are normalizable.
  \end{enumerate}
\end{definition}

Following Martin-L\"of~\cite{martin-lof71sls}, instead of assigning
some ordinal measures to derivations and define an ordering on them,
we shall use the derivation figures themselves as a measure.  Each
clause in the definition of normalizability asserts that a derivation
is normalizable if certain (possibly infinitely many) other
derivations are normalizable. We call the latter the \emph{predecessors} of the former.  Thus a derivation is normalizable if
the tree of its successive predecessors is well-founded.  We refer to
this well-founded tree as its \emph{normalization}.

Since a normalization is well-founded, it has an associated induction
principle: for any property $P$ of derivations, if for every
derivation $\Pi$ in the normalization, $P$ holds for every predecessor
of $\Pi$ implies that $P$ holds for $\Pi$, then $P$ holds for every
derivation in the normalization.

The set of all normalizable derivations is denoted by $\NM$.

\begin{lemma}
  \label{lm:norm-cut-free}
  If there is a normalizable derivation of a sequent, then there is a
  cut-free derivation of the sequent.
\end{lemma}
\begin{proof} Let $\Pi$ be a normalizable derivation of the sequent
  $\Seq{\Gamma}{B}$.  We show by induction on the normalization of
  $\Pi$ that there is a cut-free derivation of $\Seq{\Gamma}{B}$.
  \begin{enumerate}

  \item If $\Pi$ ends with a multicut, then any of its reducts is one
    of its predecessors and so is normalizable.  But the reduct is
    also a derivation of $\Seq{\Gamma}{\Bscr}$, so by the induction
    hypothesis this sequent has a cut-free derivation.

  \item Suppose $\Pi$ ends with a rule other than multicut.  Since we
    are given that $\Pi$ is normalizable, by definition the premise
    derivations are normalizable.  These premise derivations are the
    predecessors of $\Pi$, so by the induction hypothesis there are
    cut-free derivations of the premises.  Thus there is a cut-free
    derivation of $\Seq{\Gamma}{\Bscr}$.
  \end{enumerate}
  \qed
\end{proof}

The next lemma states that normalization is closed under
substitutions.

\begin{lemma}
  \label{lm:subst-norm}
  If $\Pi$ is a normalizable derivation, then for any substitution
  $\theta$, $\Pi\theta$ is normalizable.
\end{lemma}
\begin{proof} We prove this lemma by induction on the normalization of
  $\Pi$.
  \begin{enumerate}
 
  \item If $\Pi$ ends with a multicut, then $\Pi\theta$ also ends with
    a multicut. By Lemma~\ref{lm:reduct subst} every reduct of
    $\Pi\theta$ corresponds to a reduct of $\Pi$, therefore by
    induction hypothesis every reduct of $\Pi\theta$ is normalizable,
    and hence $\Pi\theta$ is normalizable.
 
  \item Suppose $\Pi$ ends with a rule other than multicut and has
    premise derivations $\{\Pi_i\}$.  By Definition~\ref{def:subst}
    each premise derivation in $\Pi\theta$ is either $\Pi_i$ or
    $\Pi_i\theta$.  Since $\Pi$ is normalizable, $\Pi_i$ is
    normalizable, and so by the induction hypothesis $\Pi_i\theta$ is
    also normalizable.  Thus $\Pi\theta$ is normalizable.  \qed
  \end{enumerate}
\end{proof}


\subsection{Parametric reducibility}

Let us first define some terminology concerning derivations.  We say
that a derivation $\Pi$ has type $C$ if the end sequent of $\Pi$ is of
the form $\Seq{\Gamma}{C}$ for some $\Gamma$.  We say that a set of
derivations $\Sscr$ has type $C$, if every derivation $\Pi \in \Sscr$
has type $C$.  A set of derivations $\Rscr$ is \emph{closed under
  substitution} if for every $\Pi \in \Rscr$ and for every
substitution $\theta$, $\Pi \theta \in \Rscr$.

To simplify presentation, we shall use the following notations to
denote certain types of derivations.  The derivation
$$
\infer[mc] {\Seq{\Delta_1, \ldots, \Delta_n, \Gamma }{C}} {
  \deduce{\Seq {\Delta_1}{B_1}}{\Pi_1} & \cdots & \deduce{\Seq
    {\Delta_n}{B_n}}{\Pi_n} & \deduce{\Seq {\Gamma}{C}}{\Pi} }
\enspace
$$
is abbreviated as $mc(\Pi_1, \ldots, \Pi_n, \Pi )$. The derivation
$$
\infer[\indL] {\Seq{\Gamma, p\,\vec u}{C}} {\deduce{\Seq{B\,S\,\vec
      x}{S\,\vec x}}{\Pi_S} & \deduce{\Seq{\Gamma, S\,\vec u}{C}}{\Pi}
}
$$
is abbreviated as $ind(\Pi_S, \Pi)$, and the derivation
$$
\infer[\coindR] {\Seq{\Gamma}{p\,\vec u}} {
  \deduce{\Seq{\Gamma}{S\,\vec u}}{\Pi} & \deduce{\Seq{S\,\vec
      x}{B\,S\,\vec x}}{\Pi_S} }
$$
is abbreviated as $coind(\Pi,\Pi_S)$.

\begin{definition}
  \label{def:candidates}
  Let $F$ be a closed term of type $\alpha_1 \to \cdots \to \alpha_n
  \to o$.  A set of derivations $\Sscr$ is said to be {\em
    $F$-indexed} if every derivation in $\Sscr$ has type $F\,t_1
  \ldots t_n$ for some $t_1, \ldots, t_n$.
\end{definition}

Given a set $\Sscr$ of derivations and a formula $C$, we denote with
$\restrict \Sscr C$ the set
$$
\{\Pi \in \Sscr \mid \hbox{$\Pi$ is of type $C$ } \}.
$$ 

We shall now define a family of sets of derivations, which we call
parametric reducibility sets.

\begin{definition}
  \label{def:param red} \emph{Parametric Reduciblity.}
  Let $p\vec x \defnu B\,p\,\vec x$ be a co-inductive definition, let
  $I$ be a closed term of the same type as $p$, let $\Rscr$ be a set
  of derivations, and let $\Sscr$ be an $I$-indexed set of
  derivations.  Let $C$ be a formula dominated by $p$.  We define the
  \emph{parametric reducibility sets} $\RED^p_C[\Rscr, \Sscr] $,
  consisting of derivations of type $C[I/p]$, by induction on the size
  of $C$, as follows.  (In the following, we shall refer to $C$ as the
  \emph{type} of $\RED^p_C[\Rscr,\Sscr]$.)
  \begin{enumerate}
  \item If $p$ does not appear in $C$ then $\RED^p_C[\Rscr, \Sscr] =
    \restrict {\Rscr} {C}$.
  \item If $C = p\,\vec u$, for some $\vec u$, then $\RED^p_C[\Rscr,
    \Sscr] = \restrict {\Sscr} {I\,\vec u}$.
  \item Otherwise, the family of parametric reducibility sets $\{
    \RED^p_{C\theta}[\Rscr, \Sscr] \}_\theta$ is the smallest family
    that satisfies the following: for every $\theta$ and for every
    derivation $\Pi$ of type $C\theta[I/p]$, $\Pi \in
    \RED^p_{C\theta}[\Rscr, \Sscr]$ if one of the following holds:
    \begin{enumerate}
    \item $\Pi$ ends with $mc$, and all its reducts are in
      $\RED^p_{C\theta}[\Rscr, \Sscr]$.
    \item $\Pi$ ends with $\oimpR$, i.e.,
$$
\infer[\oimpR] {\Seq \Gamma {B \oimp D[ I/ p]}} {\deduce{\Seq
    {\Gamma, B}{D[ I/ p]}}{\Pi'} }
$$
$\Pi' \in \RED^p_D [\Rscr, \Sscr ] $, and for every substitution
$\rho$ and for every derivation $\Xi$ of $\Seq \Delta {B\rho }$ in
$\Rscr$, we have $mc(\Xi,\Pi'\rho) \in \RED^p_{D\rho}[\Rscr, \Sscr]$.

\item $\Pi$ ends with a rule $\rho$ other than $mc$ and $\oimpR$, the
  minor premise derivations of $\Pi$ are normalizable, and its major
  premise derivations are in the parametric reducibility sets of the
  appropriate types.
\end{enumerate}
\end{enumerate}

\end{definition}

From now on, when we write $\RED_C^p[\Rscr,\Sscr]$, it is understood
that $p$ is a co-inductive predicate, $C$ is dominated by $p$, $\Rscr$
is a set of derivations, and $\Sscr$ is an $I$-indexed set of
normalizable derivations, for some $I$.

Note that in Definition~\ref{def:param red} (3), we define
simultaneously the reducibility sets $\RED^P_{C\theta}[\Rscr,\Sscr]$
for all substitution $\theta$.  This is because in the case the
derivation $\Pi$ ends with $\eqL$, reducibility of $\Pi$ may depend on
the reducibility of (possibly infinitely many) derivations which are
in $\RED^p_{C\rho}[\Rscr, \Sscr]$ for some $\rho$. Since $C\rho$ is of
the same size as $C\theta$, its parametric reducibility set may not
yet be defined by induction on the size. We therefore need to define
this and other reducibility sets which are indexed by instances of $C$
simultaneously.

As with the definition of normalizability, clause (3) in
Definition~\ref{def:param red} defines a monotone fixed point operator
(assuming the parametric reducibility sets of smaller types have been
fixed), and it therefore induces a well-founded tree of derivations in
the family $\{\RED^p_{C\theta}[\Rscr,\Sscr] \}_\theta$.  It is
immediately clear from the definition that a derivation $\Pi'$ in the
family is a predecessor of $\Pi$ (in the same family) if either
\begin{itemize}
\item $\Pi$ ends with a left rule and $\Pi'$ is a major premise of
  $\Pi$, or
\item $\Pi$ ends with $mc$ and $\Pi'$ is a reduct of $\Pi$.
\end{itemize}
We shall call the well-founded tree of successive predecessors of a
derivation $\Pi$ in the family $\{\RED^p_{C\theta}[\Rscr,\Sscr]
\}_\theta$ the \emph{parametric reduction} of $\Pi$. As with the
normalization of a derivation, it has an associated induction
principle.  Note that, however, this ordering on derivations is
defined only in the case where $C$ satisfies the syntactic condition
defined in Definition~\ref{def:param red}(3), i.e., it contains at
least an occurrence of $p$ and is not an atomic formula.

The definition of parametric reducibility can be seen as defining a
function on $S$-indexed sets. In the case where the type of the
parametric reducibility set is the body of the co-inductive definition
for $p$, this function corresponds to the underlying fixed point
operator for $p$.  We shall now define a class of $S$-indexed sets
which are closed under this fixed point operator. These sets, called
saturated sets in the following, can be seen as post-fixed points of
the fixed point operator for the co-inductive definition for $p$.
They will be used in defining the reducibility of derivations
involving the co-induction rule $\coindR$.

\begin{definition}
  \label{def:saturated sets}
  Let $\forall \vec x.\ p\,\vec x \defnu B\,p\,\vec x$ be an
  co-inductive definition. Let $S$ be a closed term of the same type
  as $p$.  Let $\Pi_S$ be a derivation of $\Seq {S\,\vec x}
  {B\,S\,\vec x}$.  Let $\Rscr$ be a set of derivations.  An
  $S$-indexed set $\Sscr$ is a \emph{$(\Rscr,\Pi_S)$-saturated set} if
  the following hold:
  \begin{enumerate}
  \item Every derivation in $\Sscr$ is normalizable.
  \item If $\Pi \in \Sscr$ then $\Pi\theta \in \Sscr$ for any
    $\theta$.
  \item If $\Pi \in \Sscr$ and $\Pi$ is of type $S\,\vec u$ for some
    $\vec u$, then $mc(\Pi, \Pi_S[\vec u/\vec x]) \in \RED_{B\,p\,\vec
      u}^p[\Rscr,\Sscr]$.
  \end{enumerate}
\end{definition}

\subsection{Reducibility}

We now define a family of reducible sets $\RED_i$ of level $i$.

\begin{definition}
  \label{def:red} \emph{Reducibility.}
  We define the family $\{\RED_i \}_i$ of \emph{reducible sets} of
  level $i$ by induction on $i$. In defining the reducible set of
  level $i$, we assume that reducible sets of smaller levels have been
  defined.  Each set $\RED_i$ the smallest set that satisfies the
  following: For every derivation $\Pi$ of level $i$, $\Pi \in \RED_i$
  if one of the following holds:
  \begin{enumerate}
  \item $\Pi$ ends with $mc$ and all its reducts are in $\RED_i$.
  \item $\Pi$ is
$$
\infer[\oimpR,] {\Seq \Gamma {B \oimp D}} {\deduce{\Seq {\Gamma,
      B}{D}}{\Pi'} }
$$
$\Pi' \in \RED_{\level D}$, and for every substitution $\theta$ and
for every derivation $\Xi$ of $\Seq \Delta {B\theta }$ in
$\RED_{\level{B\theta}}$, we have $mc(\Xi,\Pi'\theta) \in
\RED_{\level{D\theta}}$.

\item $\Pi$ ends with $\coindR$, i.e., $\Pi$ is
$$
\infer[\coindR] {\Seq{\Gamma}{p\,\vec t}} {\deduce{\Seq \Gamma S\,
    \vec t}{\Pi'} & \deduce{\Seq {S\,\vec x} {B\,S\,\vec x}}{\Pi_S} }
$$
where $p\,\vec x \defnu B\,p \, \vec x,$ $\Pi'$ and $\Pi_S$ are
normalizable, and there exists a $(\Rscr,\Pi_S)$-saturated set
$\Sscr$, where $\Rscr = \bigcup \{\RED_j \mid j < i \},$ such that
$\Pi' \in \Sscr$.

\item $\Pi$ ends with a rule $\rho$ other than $mc$ and $\oimpR$, the
  minor premise derivations of $\Pi$ are normalizable, and its major
  premise derivations are in the reducibility sets of the appropriate
  levels.
\end{enumerate}
\end{definition}

As in the definition of normalizability, each clause in the definition
of reducibility asserts that a derivation is reducible provided that
certain other derivations, called the predecessors of the derivation,
are reducible. The definition of reducibility induces a well-founded
ordering on derivations in the reducibility sets. We shall refer to
this ordering as \emph{reducibility ordering} and the induced
well-founded tree as the \emph{reduction} of the derivation.  We say
that a derivation is \emph{reducible} if it is in $\RED_i$ for some
$i$.

\begin{lemma}
  \label{lm:red-norm}
  Every reducible derivation is normalizable.
\end{lemma}
\begin{proof}
  Given a reducible derivation $\Pi$, it is straightforward to show by
  induction on its reduction that it is normalizable.  In the case
  where $\Pi$ ends with $\coindR$, by the definition of saturated sets
  (Definition~\ref{def:saturated sets}) and reducibility
  (Definition~\ref{def:red}), its premise derivations are
  normalizable, and therefore $\Pi$ is also normalizable.  \qed
\end{proof}

\begin{lemma}
  \label{lm:red-subst}
  If $\Pi$ is reducible then for every derivation $\theta$, $\Pi
  \theta$ is also reducible.
\end{lemma}
\begin{proof}
  The proof is by induction on the reduction of $\Pi$. We consider two
  non-trivial cases here: the case where $\Pi$ ends with $mc$ and the
  case where it ends with $\coindR$.  For the former, suppose that
  $\Pi = mc(\Pi_1,\ldots,\Pi_n,\Pi')$.  By Lemma~\ref{lm:reduct
    subst}, every reduct of $\Pi\theta$, say $\Xi$, is the result of
  substituting a reduct of $\Pi$.  By induction hypothesis, every
  reduct of $\Pi\theta$ is reducible, hence $\Pi\theta$ is also
  reducible.

  We now consider the case $\Pi$ ends with $\coindR$, i.e., $\Pi$ is
$$
\infer[\coindR] {\Seq{\Gamma}{p\,\vec t}} {
  \deduce{\Seq{\Gamma}{S\,\vec t}}{\Pi'} & \deduce{\Seq{S\,\vec
      x}{B\,S\,\vec x}}{\Pi_S} }
$$
where $p\,\vec x \defnu B\,p\,\vec x$. Let $i$ be the level of $p$ and
let $\Rscr = \bigcup \{\RED_j \mid j < \level{p} \}$.  By the
definition of reducibility, we have that $\Pi'$ and $\Pi_S$ are both
normalizable, and moreover, there exists a $(\Rscr, \Pi_S)$-saturated
set $\Sscr$, such that $\Pi' \in \Sscr$.  Suppose that $\vec u = (\vec
t)\theta$.  To show that $\Pi\theta$ is reducible, we must first show
that both $\Pi'\theta$ and $\Pi_S$ are normalizable.  This is
straightforward from the fact that both $\Pi'$ and $\Pi_S$ are
normalizable and that normalisation is closed under substitutions
(Lemma~\ref{lm:subst-norm}).  It remains to show that there exists a
$(\Rscr,\Pi_S)$-saturated set $\Sscr'$ such that $\Pi'\theta \in
\Sscr'$. Let $\Sscr' = \Sscr$. Since saturated sets are closed under
substitution and $\Pi' \in \Sscr'$, we have $\Pi'\theta \in \Sscr'$.
\qed
\end{proof}

\begin{lemma}
  \label{lm:red equal param red}
  Let $p$ be a co-inductive predicate, let $S$ be a closed term of the
  same type as $p$. Let $\Rscr = \bigcup \{\RED_j \mid j <
  \level{p}\},$ let
$$\Sscr = \bigcup \{\Xi \mid \hbox{$\Xi$ is reducible and has type $S\,\vec t$ for some $\vec t$} \} $$
and let $C$ be a formula dominated by $p$.  Then for every reducible
derivation $\Pi$ of type $C[S/p]$, $\Pi \in \RED^p_{C}[\Rscr, \Sscr]$.
\end{lemma}
\begin{proof}
  By induction on the reduction of $\Pi$.  If $p$ does not occur in
  $C$ then $\Pi \in \Rscr$, since in this case $\level{C} < \level{p}$
  (recall that $C$ is dominated by $p$), therefore $\Pi \in
  \RED^p_C[\Rscr, \Sscr]$.  If $C = p$ then $\Pi \in \Sscr$ (since
  $\Pi$ is reducible), hence $\Pi \in \RED^p_C[\Rscr, \Sscr]$.  The
  other cases follow from straightforwardly from induction hypothesis.
  We show here the case where $\Pi$ ends with $\oimpR$.
$$
\infer[\oimpR] {\Seq{\Gamma}{B \oimp D[S/p]}} {\deduce{\Seq{\Gamma,
      B}{D[S/p]}}{\Pi'}}
$$
Note that in this case $C = B \oimp D$, and $p$ does not occur in $B$
by the restriction on $C$ ($p$ dominates $C$).  Since $\Pi$ is
reducible, we have that $\Pi'$ is a reducible predecessor of $\Pi$,
and for every substitution $\theta$ and every reducible derivation
$\Xi$ of type $B\theta$, we have $mc(\Xi, \Pi'\theta)$ is also a
reducible predecessor of $\Pi$.  It thus follows from induction
hypotheses that $ \Pi' \in \RED_D^p[\Rscr,\Sscr] $ and for every $\Xi
\in \Rscr$ of type $B\theta$ (which is reducible by the definition of
$\Rscr$), $mc(\Xi, \Pi'\theta) \in \RED_{D\theta}^p[\Rscr, \Sscr]$.
Therefore, by the definition of parametric reducibility, we have that
$\Pi \in \RED_C^p[\Rscr, \Sscr]$.  \qed
\end{proof}


\subsection{Reducibility of unfolded derivations}

The following lemmas state that reducibility is preserved by
(co)inductive unfolding, under certain assumptions.

\begin{lemma} \emph{Inductive unfolding.}
  \label{lm:ind unfold}
  Let $p\,\vec{x} \defmu B\,p\,\vec{x}$ be an inductive definition.
  Let $\Pi_S$ be a reducible derivation of $\Seq{B\,S\,
    \vec{x}}{S\,\vec{x}}$.  Let $\Pi$ be a reducible derivation of
  $\Seq{\Gamma}{C}$ such that $p$ dominates $C$. Suppose the following
  statements hold:
  \begin{enumerate}
  \item For every derivation $\Xi$ of $\Seq{\Delta}{B\,p\,\vec{u}}$,
    if $\mu(\Xi,\Pi_S)$ is reducible, then the derivation
    $mc(\mu(\Xi,\Pi_S), \Pi_S[\vec u/ \vec x])$ is reducible.

  \item For every reducible derivation $\Xi$ of $\Seq{\Delta}{S\,\vec
      u}$ the derivation
    $mc(\Xi,\idrv_{S\,\vec u})$ is reducible.

  \item The derivation $ind(\Pi_S, \idrv_{S\,\vec u})$
    is reducible, for every $\vec u$ of the appropriate types.

\end{enumerate}
Then the derivation $\mu^p_C(\Pi,\Pi_S)$ of $\Seq{\Gamma}{C[S/p]}$ is
reducible.
\end{lemma}

\begin{proof}
  By induction on the reduction of $\Pi$. We show the non-trivial
  cases, assuming that $p$ is not vacuous in $C$.  To simplify
  presentation, we shall write $\mu(.,.)$ instead of $\mu_F^p(.,.)$,
  since in each of the following cases, it is easy to infer from the
  context which $F$ we are referring to.

  \begin{enumerate}
  \item Suppose $\Pi$ ends with $\init$ rule on $p\,\vec{u}$.  Then
    $\mu(\Pi,\Pi_S) = ind(\Pi_S, \idrv_{S\,\vec u})$, which is
    reducible by assumption.


  \item Suppose $\Pi$ ends with $\oimpR$, that is, $C = C_1 \oimp
    C_2$.
$$
\infer[\oimpR] {\Seq{\Gamma}{C_1 \oimp C_2}}
{\deduce{\Seq{\Gamma,C_1}{C_2}}{\Pi'}} \enspace
$$
By the restriction on $C$, we know that $p$ is vacuous in $C_1$, hence
$C[S/p] = C_1 \oimp C_2[S/p]$.  By the definition of reducibility, the
derivation $\Pi'$ is reducible and for every substitution $\theta$ and
every reducible derivation $\Psi$ of $\Seq{\Delta}{C_1\theta}$, the
derivation $\Xi$
$$
\infer[\mc] {\Seq{\Delta, \Gamma\theta}{C_2\theta}} {
  \deduce{\Seq{\Delta}{C_1\theta}}{\Psi} & \quad & \deduce{\Seq
    {\Gamma\theta,C_1\theta} {C_2\theta}}{\Pi'\theta} }
$$
is reducible.  We want to show that the derivation $\mu(\Pi,\Pi_S)$
$$
\infer[\oimpR] {\Seq{\Gamma}{C_1\,p\oimp C_2[S/p]}}
{\deduce{\Seq{\Gamma,C_1\,p} {C_2\,S}}{\mu(\Pi',\Pi_S)}}
$$
is reducible. This reduces to showing that $\mu(\Pi',\Pi_S)$ is
reducible and that
$$
\infer[\mc] {\Seq{\Delta, \Gamma\theta}{C_2\theta [S/p]}} {
  \deduce{\Seq{\Delta}{C_1\theta}}{\Psi} & \quad &
  \deduce{\Seq{\Gamma\theta,C_1\theta} {C_2\theta
      [S/p]}}{\mu(\Pi',\Pi_S)\theta} }
$$
is reducible. The first follows from induction hypothesis on $\Pi'$.
For the second derivation, we know from Lemma~\ref{lm:ind unfold
  subst} that
$$
\mu(\Pi',\Pi_S)\theta = \mu(\Pi'\theta,\Pi_S).
$$
It follows from this and the definition of inductive unfolding
(Definition~\ref{def:inductive-unfolding}) that
$$mc(\Psi, \mu(\Pi',\Pi_S)\theta) = mc(\Psi, \mu(\Pi'\theta, \Pi_S)) = \mu(mc(\Psi, \Pi'\theta), \Pi_S) = \mu(\Xi, \Pi_S)$$
We can apply induction hypothesis on $\Xi$, since it is a predecessor
of $\Pi$, to establish the reducibility of $\mu(\Xi, \Pi_S)$. This,
together with reducibility of $\mu(\Pi',\Pi_S)$ implies that
$\mu(\Pi,\Pi_S)$ is reducible.

\item Suppose $\Pi$ ends with $\indR$ rule on $p\,\vec{u}$.
$$
\infer[\indR] {\Seq{\Gamma}{p\,\vec{u}}} {
  \deduce{\Seq{\Gamma}{B\,p\,\vec{u}}}{\Pi'} } \enspace
$$
Then $\mu(\Pi,\Pi_S)$ is the derivation
$$
\infer[\mc] {\Seq{\Gamma}{S\,\vec{u}}} {
  \deduce{\Seq{\Gamma}{B\,S\,\vec{u}}}{\mu(\Pi',\Pi_S)} & \quad &
  \deduce{\Seq{B\,S\,\vec{u}}{S\,\vec{u}}}{\Pi_S[\vec u/ \vec x]} }
\enspace
$$
The derivation $\mu(\Pi',\Pi_S)$ is reducible by induction hypothesis.
This, together with assumption (1) of the lemma, imply that
$\mu(\Pi,\Pi_S)$ is reducible.

\item Suppose $\Pi$ ends with $\mc$.
$$
\infer[\mc]{\Seq{\Delta_1,\ldots,\Delta_m,\Gamma'}{C}}
{\deduce{\Seq{\Delta_1}{D_1}} {\Pi_1} & \cdots &
  \deduce{\Seq{\Delta_n}{D_m}} {\Pi_n} &
  \deduce{\Seq{D_1,\ldots,D_m,\Gamma'}{C}} {\Pi'}} \enspace
$$
Then $\mu(\Pi,\Pi_S)$ is the derivation
$$
\infer[\mc]{\Seq{\Delta_1,\ldots,\Delta_m,\Gamma'}{C[S/p]}}
{\deduce{\Seq{\Delta_1}{D_1}} {\Pi_1} & \cdots &
  \deduce{\Seq{\Delta_n}{D_m}} {\Pi_n} &
  \deduce{\Seq{D_1,\ldots,D_m,\Gamma'}{C[S/p]}} {\mu(\Pi',\Pi_S)}}
\enspace
$$
By the definition of reducibility, every reduct of $\Pi$ is
reducible. We need to show that every reduct of $\mu(\Pi,\Pi_S)$ is
reducible.

From Lemma~\ref{lm:unfold-reduct-ind}, we know that for the case where
$C$ is not atomic every reduct of $\mu(\Pi,\Pi_S)$ corresponds to some
reduct of $\Pi$.  Similarly, for the case where $\Pi'$ ends with a
rule other than $\init$ or $\indR$, by
Lemma~\ref{lm:unfold-reduct-ind2}, the reducts of $\mu(\Pi, \Pi_S)$
are in one-to-one correspondence with the reducts of $\Pi$.  Therefore
in these cases, the inductive hypothesis can be applied to show the
reducibility of each reduct of $\mu(\Pi,\Pi_S)$.  This leaves us the
following two cases, where $C = p\,\vec u$ and $\Pi'$ ends with either
$\indR$ or $\init$ rules.
\begin{itemize}
\item Suppose $\Pi'$ is the derivation
$$
\infer[\indR] {\Seq{D_1,\dots,D_m,\Gamma'} {p\,\vec{u}}} {
  \deduce{\Seq{D_1,\dots,D_m,\Gamma'} {B\,p\,\vec{u}}}{\Pi''} }
\enspace
$$
Let $\Xi_1$ be the derivation
$$
\infer[\mc] {\Seq{\Delta_1,\dots,\Delta_m,\Gamma'} {B\,p\,\vec{u}}} {
  {\left\{\raisebox{-1.5ex}{\deduce{\Seq{\Delta_j}{D_j}}{\Pi_j}}
    \right\}_{j \in \{1,\dots,m\} }} & \raisebox{-1.5ex}{
    \deduce{\Seq{D_1,\dots,\Gamma'} {B\,p\,\vec{u}}}{\Pi''}} }
$$
then the derivation
$$
\infer[\indR] {\Seq{\Delta_1,\dots,\Delta_m,\Gamma'} {p\,\vec{u}}} {
  \deduce{\Seq{\Delta_1,\dots,\Delta_m,\Gamma'}
    {B\,p\,\vec{u}}}{\Xi_1} } \enspace
$$
is a reduct of $\Pi$ (by the reduction rule $-/\indR$), and therefore
by the definition of reducibility both this reduct and $\Xi_1$ are
reducible predecessors of $\Pi$.  Let $\Psi$ be the derivation
$$
\infer[\mc] {\Seq{D_1,\dots,\Gamma'} {S\,\vec{u}} } {
  \deduce{\Seq{D_1,\dots,\Gamma'} {B\,S\,\vec{u}}}{\mu(\Pi'',\Pi_S)} &
  \deduce{\Seq{B\,S\,\vec{u}} {S\,\vec{u}}}{\Pi_S'} }
$$
Then the derivation $\mu(\Pi,\Pi_S)$ is the following
$$
\infer[\mc] {\Seq{\Delta_1,\dots,\Delta_m,\Gamma'} {S\,\vec{u}}} {
  {\left\{\raisebox{-1.5ex}{\deduce{\Seq{\Delta_j}{D_j}}{\Pi_j}}
    \right\}_{j \in \{1,\dots,m\} }} & \raisebox{-1.5ex}{\deduce
    {\Seq{D_1,\dots,\Gamma'} {S\,\vec{u}} } {\Psi } } } \enspace
$$
The only applicable reduction rule to $\mu(\Pi,\Pi_S)$ is $-/\mc$,
which gives us the reduct $\Xi$
$$
\infer[\mc] {\Seq{\Delta_1,\dots,\Delta_m,\Gamma'} {S\,\vec{u}} } {
  \deduce {\Seq{\Delta_1,\dots,\Delta_m,\Gamma'} {B\,S\,\vec{u}}}
  {\Psi'} & \deduce{\Seq{B\,S\,\vec{u}} {S\,\vec{u}}}{\Pi_S'} }
\enspace ,
$$
where $\Psi'$ is the derivation
$$
\infer[\mc] {\Seq{\Delta_1,\dots,\Delta_m,\Gamma'} {B\,S\,\vec{u}}} {
  {\left\{\raisebox{-1.5ex}{\deduce{\Seq{\Delta_j}{D_j}}{\Pi_j}}
    \right\}_{j \in \{1,\dots,m\} }} & \raisebox{-1.5ex}{
    \deduce{\Seq{D_1,\dots,\Gamma'}
      {B\,S\,\vec{u}}}{\mu(\Pi'',\Pi_S)}} }
$$
Notice that $\Psi'$ is exactly $\mu(\Xi_1,\Pi_S)$, and is reducible by
inductive hypothesis. Therefore assumption (1) applies, and the reduct
$\Xi$ is reducible, hence $\mu(\Pi,\Pi_S)$ is also reducible.

\item Otherwise, suppose $\Pi'$ ends with $\init$, then $D_1 = p\,\vec
  u$ and $\Pi$ is the derivation
$$
\infer[\mc] {\Seq{\Delta_1}{p\,\vec{u}}} {\deduce{\Seq{\Delta_1}
    {p\,\vec{u}}}{\Pi_1} & \infer[\init]{\Seq{p\,\vec{u}}
    {p\,\vec{u}}}{} } \enspace
$$
The only reduct of $\Pi$ is $\Pi_1$ since the only applicable
reduction is $-/\init$.  On the other hand, the derivation
$\mu(\Pi,\Pi_S)$ is
$$
\infer[\mc] {\Seq{\Delta_1}{S\,\vec{u}}} {
  \deduce{\Seq{\Delta_1}{p\,\vec{u}}}{\Pi_1} & \infer[\indL]
  {\Seq{p\,\vec{u}}{S\,\vec{u}}} {
    \deduce{\Seq{B\,S\,\vec{x}}{S\,\vec{x}}}{\Pi_S} &
    \deduce{\Seq{S\,\vec{u}} {S\,\vec{u}}}{\idrv} } }
$$
Its only reduct is (by $*/\indL$)
$$
\infer[\mc] {\Seq{\Delta_1}{S\,\vec{u}}} {
  \deduce{\Seq{\Delta_1}{S\,\vec{u}}}{\mu(\Pi_1,\Pi_S)} &
  \deduce{\Seq{S\,\vec{u}} {S\,\vec{u}}}{\idrv} } \enspace
$$
The derivation $\mu(\Pi_1,\Pi_S)$ is reducible by inductive hypothesis
($\Pi_1$ is a predecessor of $\Pi$) and assumption (2) applies, and
the above reduct is reducible.
\end{itemize}
\end{enumerate}
\qed
\end{proof}

\begin{remark}
  Intuitively, condition (1) of Lemma~\ref{lm:ind unfold} can be seen
  as asserting that the set of reducible derivations whose types are
  instances of $S\,\vec x$ forms a pre-fixed point of the fixed point
  operator induced by the inductive definition of $p$.
\end{remark}

\begin{lemma} \emph{Co-inductive unfolding.}
  \label{lm:coind unfold}
  Let $p\,\vec{x} \defnu B\,p\,\vec{x}$ be a co-inductive definition.
  Let $\Pi_S$ be a normalizable derivation of $\Seq{S\,\vec{x}}{B\,S\,
    \vec{x}}$ for some invariant $S$.  Let $\Rscr = \{\RED_j \mid j <
  \level{p} \},$ and let $\Sscr$ be a $(\Rscr, \Pi_S)$-saturated set.
  Let $\Pi$ be a derivation of $\Seq{\Gamma}{C[S/p]}$ for some $C$
  dominated by $p$.  If $\Pi \in \RED_{C}[\Rscr, \Sscr]$ then
  $\nu_C^p(\Pi, \Pi_S)$ is reducible.
\end{lemma}

\begin{proof}
  By induction on the size of $C$, with sub-induction on the parametric
  reduction of $\Pi$.  As in the proof of inductive unfolding, we omit
  the subscript and superscript in the $\nu$ function to simplify the
  presentation of the proof.

  \begin{enumerate}
  \item If $p$ is not free in $C$, then $\nu(\Pi,\Pi_S) = \Pi$. Since
    $\Pi \in \RED_C[\Rscr, \Sscr]$, it follows from the definition of
    parametric reducibility that $\Pi \in \Rscr$, hence it is
    reducible by assumption.

  \item Suppose $C = p\,\vec{u}$.  Then $C[S/p] = S\,\vec{u}$ and
    $\nu(\Pi,\Pi_S)$ is the derivation
$$
\infer[\coindR] {\Seq{\Gamma}{p\,\vec{u}}} {
  \deduce{\Seq{\Gamma}{S\,\vec{u}}}{\Pi} &
  \deduce{\Seq{S\,\vec{x}}{B\,S\,\vec{x}}}{\Pi_S} } \enspace
$$
To show that this derivation is reducible, we first show that there
exist a $(\Rscr,\Pi_S)$-saturated set $\Sscr'$ such that $\Pi \in
\Sscr'$.  Since $\Pi \in \RED_{p\,\vec u}^p[\Rscr,\Sscr],$ by the
definition of parametric reducibility, we have $\Pi \in \Sscr$. Let
$\Sscr' = \Sscr$. Then $\Sscr'$ is indeed a $(\Rscr,\Pi_S)$-saturated
set containing $\Pi$. It remains to show that both $\Pi$ and $\Pi_S$
are normalizable. This follows from the assumption on $\Pi_S$ and the
fact that saturated sets contain only normalizable derivations.

\item Suppose $p$ occurs in $C$ but $C \not = p \, \vec u$ for any
  $\vec u$.  There are several subcases, depending on the last rule in
  $\Pi$.  Then we show by induction on parametric reducibility of
  $\Pi$ that it is also reducible.
  \begin{enumerate}
  \item The base cases are those where $\Pi$ ends with a rule with
    empty premises and where $\Pi$ ends with a right-introduction
    rule. In the former case, its reducibility is immediate from the
    definition of reducibility (Definition~\ref{def:red}). For the
    latter, in most cases, the reducibility of $\Pi$ follows from the
    outer induction hypothesis (since in this case, the premise
    derivations of $\Pi$ are in the parametric reducibility sets of
    smaller types) and Definition~\ref{def:red}. We show here a
    non-trivial case involving implication-right: Suppose $\Pi$ ends
    with $\oimpR$, i.e., $C = C_1 \oimp C_2$ for some $C_1$ and $C_2$.
$$
\infer[\oimpR] {\Seq{\Gamma}{C_1 \oimp C_2[S/p]}}
{\deduce{\Seq{\Gamma, C_1} {C_2[S/p]}}{\Pi'}}
$$
Note that $p$ is vacuous in $C_1$ by the restriction on $C$.  The
derivation $\nu(\Pi,\Pi_S)$ is
$$
\infer[\oimpR] {\Seq{\Gamma}{C_1\oimp C_2}} {\deduce{\Seq{\Gamma, C_1}
    {C_2}}{\nu(\Pi',\Pi_S)}} \enspace
$$
To show that $\nu(\Pi,\Pi_S)$ is reducible, we need to show that
$\nu(\Pi',\Pi_S)$ is reducible, and for every $\theta$ and every $\Psi
\in \RED_{C_1\theta}$, we have $mc(\Psi, \nu(\Pi',\Pi_S)\theta) \in
\RED_{C_2\theta}$.

The parametric reducibility of $\Pi$ implies that $\Pi' \in
\RED_{C_2}[\Rscr,\Sscr]$ and for every $\theta$ and every derivation
$\Psi' \in \Rscr$, $mc(\Psi', \Pi'\theta) \in
\RED_{C_2\theta}[\Rscr,\Sscr]$.  Note that $\Psi$ is in $\Rscr$ since
$\level{C_1\theta} < \level{p}$.  Therefore we also have $mc(\Psi,
\Pi'\theta) \in \RED_{C_2\theta}[\Rscr,\Sscr]$.  By the outer
induction hypothesis, we have that both
$$
\nu(\Pi',\Pi_S) \qquad \hbox{and} \qquad \nu(mc(\Psi, \Pi'\theta),
\Pi_S)$$ are reducible. It remains to show that the $mc(\Psi,
\nu(\Pi',\Pi_S)\theta)$ is reducible.  Note that by
Lemma~\ref{lm:coind unfold subst} this derivation is equivalent to
$mc(\Psi, \nu(\Pi'\theta,\Pi_S))$.  To show that this derivation is
reducible, there are two cases to consider.  If $C_2$ is non-atomic
then it is easy to see that $mc(\Psi, \nu(\Pi'\theta, \Pi_S))$ is
equivalent to $\nu(mc(\Psi, \Pi'\theta), \Pi_S)$, which is reducible
by the outer induction hypothesis.  If, however, $C_2 = p\,\vec u$ for
some $\vec u$, then $mc(\Psi, \nu(\Pi'\theta, \Pi_S))$ is the
derivation (supposing that the end sequent of $\Psi$ is
$\Seq{\Delta}{C_1\theta}$):
$$
\infer[mc] {\Seq{\Delta, \Gamma\theta}{p\,\vec u}} {
  \deduce{\Seq{\Delta}{C_1\theta}}{\Psi} & \infer[\coindR]
  {\Seq{C_1\theta, \Gamma\theta}{p\,\vec u}} {\deduce{\Seq{C_1\theta,
        \Gamma\theta}{S\,\vec u}}{\Pi'\theta} & \deduce{\Seq{S\,\vec
        x}{B\,S\,\vec x}}{\Pi_S} } }
$$
To show that this derivation is reducible, we must show that all its
reducts are reducible.  There is only one reduction rule that is
applicable in this case, i.e., the $-/ \coindR$-case, which leads to
the following derivation:
$$
\infer[\coindR .]  {\Seq{\Delta, \Gamma\theta}{p\,\vec u}} {
  \infer[mc] {\Seq{\Delta, \Gamma\theta}{S\,\vec u}} {
    \deduce{\Seq{\Delta}{C_1\theta}}{\Psi} & \deduce{\Seq{C_1\theta,
        \Gamma\theta}{S\,\vec u}}{\Pi'\theta} } & \deduce{\Seq{S\,\vec
      x}{B\,S\,\vec x}}{\Pi_S} }
$$
But notice that this is exactly the derivation $\nu(mc(\Psi,
\Pi'\theta), \Pi_S)$, which is reducible by the outer induction
hypothesis.

Having shown that $\nu(\Pi',\Pi_S)$ and $mc(\Psi,
\nu(\Pi',\Pi_S)\theta)$ are reducible, we have sufficient conditions
to conclude that $\nu(\Pi, \Pi_S)$ is indeed reducible.

\item For the inductive cases, $\Pi$ ends either with $mc$ or a
  left-rule.  We show the former case here (the other cases are
  straightforward). Suppose $\Pi$ is
$$
\infer[\mc]{\Seq{\Delta_1,\ldots,\Delta_n,\Gamma'}{C[S/p]}}
{\deduce{\Seq{\Delta_1}{D_1}} {\Pi_1} & \cdots &
  \deduce{\Seq{\Delta_n}{D_m}} {\Pi_n} &
  \deduce{\Seq{D_1,\ldots,D_m,\Gamma'}{C[S/p]}} {\Pi'}} \enspace
$$
Then $\nu(\Pi,\Pi_S)$ is the derivation
$$
\infer[\mc]{\Seq{\Delta_1,\ldots,\Delta_n,\Gamma'}{C}}
{\deduce{\Seq{\Delta_1}{D_1}} {\Pi_1} & \cdots &
  \deduce{\Seq{\Delta_n}{D_m}} {\Pi_n} &
  \deduce{\Seq{D_1,\ldots,D_m,\Gamma'}{C}} {\nu(\Pi',\Pi_S)}} \enspace
$$
The derivation $\nu(\Pi,\Pi_S)$ is reducible if every reduct of
$\nu(\Pi,\Pi_S)$ is also reducible.  From Lemma\ref{lm:unfold-reduct-coind}, it follows that every reduct of
$\nu(\Pi,\Pi_S)$ is of the form $\nu(\Xi,\Pi_S)$ where $\Xi$ is a
reduct of $\Pi$. Since all reducts of $\Pi$ are predecessors of $\Pi$
in the parametric reducibility ordering, we can apply the inductive
hypothesis to show that every reduct of $\nu(\Pi,\Pi_S)$ is reducible,
hence $\nu(\Pi, \Pi_S)$ is also reducible.

\end{enumerate}

\end{enumerate}
\qed
\end{proof}


\subsection{Cut elimination}

Most cases in the cut elimination proof for $\Linc$ in the following
are similar to those of $\FOLDN$. The crucial differences are in the
handling of the essential cut reductions for inductive and
co-inductive rules.\footnote{We also note that McDowell and Miller's
  proof of cut elimination for $\FOLDN$ given in \cite{mcdowell00tcs}
  appears to contain a small gap in the proof of a main technical
  lemma. More specifically, they use a similar technical lemma as
  Lemma~\ref{lm:comp}, but without the extra assumptions about the
  substitutions $\delta_1,\ldots,\delta_n,\theta$.  The problem with
  their formulation of the lemma appears in the case involving the
  $\eqL/ \circL$ reduction rule. This problem is fixed in our cut
  elimination proof with the more general statement of
  Lemma~\ref{lm:comp}.  See
  \url{http://www.lix.polytechnique.fr/~dale/papers/tcs00.errata.html}
  for details of the errata in their paper.  }  In the case of
derivations of inductive predicates, a crucial part of the proof is in
establishing that the $S$-indexed set of reducible derivations (where
$S$ is an inductive invariant) satisfies the conditions of
Lemma~\ref{lm:ind unfold} (in effect, demonstrating that the said set
forms a pre-fixed point).  Dually, in the case for co-inductive
proofs, one must show that the $S$-indexed set of reducible
derivations, where $S$ is a co-inductive invariant, forms a saturated
set (i.e., a post fixed point of the co-inductive definition
involved).

\begin{lemma}
  \label{lm:comp}
  For any derivation $\Pi$ of $\Seq{B_1,\ldots,B_n,\Gamma}{C}$, for
  any reducible derivations
$$
\deduce{\Seq{\Delta_1}{B_1}}{\Pi_1} , \quad \ldots \quad , ~
\deduce{\Seq{\Delta_n}{B_n}}{\Pi_n}
$$
where $n \geq 0$, and for any substitutions
$\delta_1,\dots,\delta_n,\gamma$ such that $B_i\delta_i = B_i\gamma$
for every $i \in \{1,\dots,n\}$, the derivation $\Xi$
\begin{displaymath}
  \infer[\mc]{\Seq{\Delta_1\delta_1,\ldots,\Delta_n\delta_n,\Gamma\gamma}{C\gamma}}
  {\deduce{\Seq{\Delta_1\delta_1}{B_1\delta_1}}
    {\Pi_1\delta_1}
    & \cdots
    & \deduce{\Seq{\Delta_n\delta_n}{B_n\delta_n}}
    {\Pi_n\delta_n}
    & \deduce{\Seq{B_1\gamma,\ldots,B_n\gamma,\Gamma\gamma}{C\gamma}}
    {\Pi\gamma}}
\end{displaymath}
is reducible.
\end{lemma}
\begin{proof}
  The proof is by induction on $\indm{\Pi}$ with subordinate induction
  on $\measure{\Pi}$, on $n$ and on the reductions of $\Pi_1, \ldots,
  \Pi_n$.  The proof does not rely on the order of the inductions on
  reductions.  Thus when we need to distinguish one of the $\Pi_i$, we
  shall refer to it as $\Pi_1$ without loss of generality.  The
  derivation $\Xi$ is reducible if all its reducts are reducible.
 
  If $n=0$, then $\Xi$ reduces to $\Pi\gamma$, thus in this case we
  show that $\Pi\gamma$ is reducible. Since reducibility is preserved
  by substitution (Lemma~\ref{lm:red-subst}), it is enough to show
  that $\Pi$ is reducible.  This is proved by a case analysis of the
  last rule in $\Pi$.  For each case, the result follows easily from
  the induction hypothesis on $\measure \Pi$ and
  Definition~\ref{def:red}.  The $\oimpR$ case requires that
  substitution for variables does not increase the measures of a
  derivation.  In the cases for $\oimpL$ and $\indL$ we need the
  additional information that reducibility implies normalizability
  (Lemma~\ref{lm:red-norm}).  The case for $\coindR$ requires special
  attention.  Let $p\,\vec{x} \defnu D\,p\,\vec{x}$ be a co-inductive
  definition.  Suppose $\Pi$ is the derivation
$$
\infer[\coindR] {\Seq{\Gamma}{p\,\vec{t}}} {
  \deduce{\Seq{\Gamma}{S\,\vec{t}}}{\Pi'} &
  \deduce{\Seq{S\,\vec{x}}{D\,S\,\vec{x}}}{\Pi_S} }
$$
for some invariant $S$. Let $\Rscr = \bigcup \{\RED_j \mid j < \level
p \}$.  To show that $\Pi$ is reducible we must show that its premises
are normalizable and that there exists a $(\Rscr,\Pi_S)$-saturated set
$\Sscr$ such that $\Pi' \in \Sscr$.  The former follows from the outer
induction hypothesis and Lemma~\ref{lm:red-norm}. For the latter, the
set $\Sscr$ is defined as follows:
$$\
\Sscr = \{\Psi \mid \hbox{$\Psi$ is a reducible derivaiton of type
  $S\,\vec u$, for some $\vec u$} \}.
$$
Since $\Pi'$ is reducible by induction hypothesis, we have $\Pi' \in
\Sscr$.  It remains to show that $\Sscr$ is a
$(\Rscr,\Pi_S)$-saturated set. More specifically, we show that $\Sscr$
has the following properties.
\begin{enumerate}
\item Every derivation in $\Sscr$ is normalizable.
\item If $\Psi \in \Sscr$ then $\Psi\theta \in \Sscr$ for any
  $\theta$.
\item If $\Psi \in \Sscr$ and $\Psi$ is of type $S\,\vec u$ for some
  $\vec u$, then $mc(\Psi, \Pi_S[\vec u/\vec x]) \in \RED_{B\,p\,\vec
    u}^p[\Rscr,\Sscr]$
\end{enumerate}
Property (1) follows from the fact that reducibility implies
normalizability (Lemma~\ref{lm:red-norm}).  Property (2) follows from
the fact that reducibility is closed under substitution
(Lemma~\ref{lm:red-subst}).  To prove (3), first notice that by
Lemma~\ref{lm:subst-mu}, $\indm{\Pi_S[\vec u/\vec x]} \leq
\indm{\Pi_S} = \indm{\Pi}$ and $\measure{\Pi_S[\vec u/ \vec x]} \leq
\measure{\Pi_S} < \measure{\Pi}$. Therefore, by the outer induction
hypothesis, we have that $mc(\Psi, \Pi_S[\vec u/ \vec x])$ is
reducible. By Lemma~\ref{lm:red equal param red}, we have that
$mc(\Psi, \Pi_S[\vec u/\vec x]) \in \RED_{B\,p\,\vec
  u}^p[\Rscr,\Sscr]$.  Therefore, $\Sscr$ is a
$(\Rscr,\Pi_S)$-saturated set containing $\Pi'$, hence $\Pi$ is
reducible.

For $n > 0$, we analyze all possible cut reductions and show for each
case the reduct is reducible. Some cases follow immediately from
inductive hypothesis. We show here the non-trivial cases.
\begin{description}
\item[$\oimpR/\oimpL$:] Suppose $\Pi_1$ and $\Pi$ are
$$
\infer[\oimpR] {\Seq{\Delta_1}{B_1'\oimp B_1''}} {
  \deduce{\Seq{\Delta_1,B_1'}{B_1''}}{\Pi_1'} } \qquad \qquad
\infer[\oimpL] {\Seq{B_1'\oimp B_1'',B_2,\dots, B_n,\Gamma}{C}} {
  \deduce{\Seq{B_2,\dots,\Gamma}{B_1'}}{\Pi'} &
  \deduce{\Seq{B_1'',B_2,\dots,\Gamma}{C}}{\Pi''} } \enspace 
$$
The derivation $\Xi_1$
$$
\infer[\mc]
{\Seq{\Delta_2\delta_2,\dots,\Delta_n\delta_n,\Gamma\gamma}{B_1'\gamma}}
{\deduce{\Seq{\Delta_2\delta_2}{B_2\delta_2}}{\Pi_2\delta_2} & \ldots
  & \deduce{\Seq{\Delta_n\delta_n}{B_n\delta_n}}{\Pi_n\delta_n} &
  \deduce{\Seq{B_2\gamma,\dots,B_n\gamma, \Gamma\gamma}{B_1'\gamma}}
  {\Pi'\gamma} }
$$
is reducible by induction hypothesis since $\indm{\Pi'} \leq
\indm{\Pi}$ and $\measure{\Pi'} < \measure{\Pi}$. Since $\Pi_1$ is
reducible, by Definition~\ref{def:red} the derivation $\Xi_2$
$$
\infer[\mc]
{\Seq{\Delta_1\delta_1,\dots,\Delta_n\delta_n,\Gamma\gamma}{B_1''\delta_1}}
{
  \deduce{\Seq{\Delta_2\delta_2,\dots,\Gamma\gamma}{B_1'\gamma}}{\Xi_1}
  & \deduce{\Seq{B_1'\delta_1,\Delta_1\delta_1}{B_1''\delta_1}}
  {\Pi_1\delta_1} }
$$
is a predecessor of $\Pi_1$ and therefore is reducible.  The reduct of
$\Xi$ in this case is the following derivation \settowidth{\infwidthi}
{$\Seq{\Delta_1\delta_1,\ldots,\Delta_n\delta_n,\Gamma\gamma,
    \Delta_1\delta_1,\ldots,\Delta_n\delta_n,\Gamma\gamma} {C\gamma}$}
\begin{displaymath}
  \infer{\Seq{\Delta_1\delta_1,\ldots,\Delta_n\delta_n,\Gamma\gamma}{C\gamma}}
  {\infer[\cL]{\makebox[\infwidthi]{}}
    {\infer[\mc]{\Seq{\Delta_1\delta_1,\ldots,
          \Delta_n\delta_n,\Gamma\gamma,
          \Delta_2\delta_2,\ldots,
          \Delta_n\gamma,\Gamma\gamma}
        {C\gamma}}
      {\raisebox{-2.5ex}{\deduce{\Seq{\ldots}{B_1''\delta_1}}
          {\Xi_2}}
        & \left\{\raisebox{-1.5ex}{\deduce{\Seq{\Delta_i\delta_i}{B_i\delta_i}}
            {\Pi_i\delta_i}}\right\}_{i \in \{2..n\}}
        & \raisebox{-2.5ex}{\deduce{\Seq{B_1''\gamma, \dots, B_n\gamma,
              \Gamma\gamma}
            {C\gamma}}
          {\Pi''\gamma}}}}}
\end{displaymath}
which is reducible by induction hypothesis and Definition~\ref{def:red}.

\item[$\forallL/\forallR$:] Suppose $\Pi_1$ and $\Pi$ are
$$
\infer[\forallR] {\Seq{\Delta_1}{\forall x.B_1'}} {
  \deduce{\Seq{\Delta_1}{B_1'[y/x]}}{\Pi_1'} } \qquad \qquad
\infer[\forallL] {\Seq{\forall x.B_1',B_2,\dots,B_n,\Gamma}{C}} {
  \deduce{\Seq{B_1'[t/x],B_2,\dots,B_n,\Gamma}{C}}{\Pi'} }
$$
Since we identify derivations that differ only in the choice of
intermediate eigenvariables that are not free in the end sequents, we
can choose a variable $y$ such that it is not free in the domains and
ranges of $\delta_1$ and $\gamma$. We assume without loss of
generality that $x$ is chosen to be fresh with respect to the free
variables in the substitutions so we can push the substitutions under
the binder.  The derivation $\Xi$ is thus
$$
\infer[\mc]
{\Seq{\Delta_1\delta_1,\dots,\Delta_n\delta_n,\Gamma\gamma}{C\gamma}}
{\infer[\forallR] {\Seq{\Delta_1\delta_1}{\forall x.B_1'\delta_1}} {
    \deduce{\Seq{\Delta_1\delta_1}{B_1'\delta_1[y /x]}}
    {\Pi_1'\delta_1} } & \ldots & \infer[\forallL] {\Seq{\forall
      x.B_1'\gamma, \dots,\Gamma\gamma}{C\gamma}} {
    \deduce{\Seq{B_1'\gamma[t\gamma/x],\dots,\Gamma\gamma}{C\gamma}}
    {\Pi'\gamma} } }
$$
Let $\delta_1' = \delta_1 \circ [t\gamma/y]$.  The reduct of $\Xi$ in
this case is
$$
\infer[\mc]
{\Seq{\Delta_1\delta_1,\dots,\Delta_n\delta_n,\Gamma\gamma}{C\gamma}}
{\deduce{\Seq{\Delta_1\delta_1} {B_1'\delta_1[t\gamma/x]}}
  {\Pi_1'\delta_1'} & \ldots &
  \deduce{\Seq{B_1'\gamma[t\gamma/x],\dots,\Gamma\gamma}
    {C\gamma}}{\Pi'\gamma} }
$$
which is reducible by induction hypothesis.

\item[$\eqR/\eqL$:] Suppose $\Pi_1$ and $\Pi$ are
$$
\infer[\eqR] {\Seq{\Delta_1}{s = t}} {} \qquad \qquad \infer[\eqL]
{\Seq{s = t,\dots,B_n,\Gamma}{C}} {\left\{\raisebox{-1.5ex}
    {\deduce{\Seq{B_2\rho,\dots,B_n\rho,\Gamma\rho}{C\rho}} {\Pi^\rho}
    } \right\}_\rho }
$$
Then $\Xi$ is the derivation
$$
\infer[\mc]
{\Seq{\Delta_1\delta_1,\dots,\Delta_n\delta_n,\Gamma\gamma}{C\gamma}}
{\infer[\eqR] {\Seq{\Delta_1\delta_1}{(s=t)\delta_1}} {} & \cdots &
  \infer[\eqL]
  {\Seq{(s=t)\gamma,\dots,B_n\gamma,\Gamma\gamma}{C\gamma}}
  {\left\{\raisebox{-1.5ex}
      {\deduce{\Seq{B_2\gamma\rho',\dots,B_n\gamma\rho',\Gamma\rho'}
          {C\gamma\rho'}} {\Pi^{\gamma \circ \rho'}}} \right\}_{\rho'}
  } }
$$
The $\eqR$ tells us that $s$ and $t$ are unifiable via empty
substitution (i.e., they are the same normal terms).  The reduct of
$\Xi$
$$
\infer=[\wL] {\Seq{\Delta_1\delta_1, \Delta_2\delta_2,
    \dots,\Delta_n\delta_n,\Gamma\gamma}{C\gamma}} {\infer[\mc]
  {\Seq{\Delta_2\delta_2,\dots,\Delta_n\delta_n,\Gamma\gamma}{C\gamma}}
  {\deduce{\Seq{\Delta_2\delta_2}{B_2\delta_2}}{\Pi_2\delta_2} &
    \ldots & \deduce{\Seq{B_2\gamma,\dots,\Gamma\gamma}{C\gamma}}
    {\Pi^\gamma} } }
$$
is therefore reducible by induction hypothesis.

\item[$*/\indL$:] Suppose $\Pi$ is the derivation
$$
\infer[\indL] {\Seq{p\,\vec{t}, \Gamma}{C}} {
  \deduce{\Seq{D\,S\,\vec{x}}{S\,\vec{x}}}{\Pi_S} &
  \deduce{\Seq{S\,\vec{t}, \Gamma}{C}}{\Pi'} }
$$
where $p\,\vec x \defmu D\,p\,\vec x$.  Let $p\,\vec{u}$ be the result
of applying $\delta_1$ to $p\,\vec{t}$. Then $\Xi$ is the derivation
$$
\infer[\mc]
{\Seq{\Delta_1\delta_1,\ldots,\Delta_n\delta_n,\Gamma\gamma}{C\gamma}}
{\deduce{\Seq{\Delta_1\delta_1}{p\,\vec{u}}} {\Pi_1\delta_1} & \cdots
  & \deduce{\Seq{\Delta_n\delta_n}{B_n\delta_n}} {\Pi_n\delta_n} &
  \infer[\indL] {\Seq{p\,\vec{u}, \ldots, \Gamma\gamma}{C\gamma}} {
    \deduce{\Seq{D\,S\,\vec{x}}{S\,\vec{x}}}{\Pi_S} &
    \deduce{\Seq{S\,\vec{u}, \ldots, \Gamma\gamma}{C\gamma}}
    {\Pi'\gamma} } }
$$
The derivation $\Xi$ reduces to the derivation $\Xi'$
$$
\infer[\mc]
{\Seq{\Delta_1\delta_1,\ldots,\Delta_n\delta_n,\Gamma\gamma}{C\gamma}}
{\deduce {\Seq{\Delta_1\delta_1}{S\,\vec{u}}}
  {\mu(\Pi_1,\Pi_S)\delta_1} & \cdots &
  \deduce{\Seq{\Delta_n\delta_n}{B_n\delta_n}} {\Pi_n\delta_n} &
  \deduce {\Seq{S\,\vec{u}, \Gamma\gamma}{C\gamma}} {\Pi'\gamma} }
\enspace 
$$
Notice that we have used the fact that
$$
\mu(\Pi_1\delta_1,\Pi_S) = \mu(\Pi_1,\Pi_S)\delta_1
$$
in the derivation above, which follows from Lemma~\ref{lm:ind unfold
  subst}.  Therefore, in order to prove that $\Xi'$ is reducible, it
remains to show that the unfolding of $\Pi_1$ produces a reducible
derivation. This will be proved using Lemma~\ref{lm:ind unfold}, but
we shall first prove the following properties, which are the
conditions for applying Lemma~\ref{lm:ind unfold}:
\begin{enumerate}
\item For every derivation $\Psi$ of $\Seq{\Delta}{D\,p\,\vec{s}}$, if
  $\mu(\Psi,\Pi_S)$ is reducible, then the derivation
  $mc(\mu(\Psi,\Pi_S), \Pi_S[\vec s/ \vec x])$ is reducible.

\item For every reducible derivation $\Psi$ of $\Seq{\Delta}{S\,\vec
    u}$ the derivation $mc(\Psi,\idrv_{S\,\vec u})$ is reducible.

\item The derivation $ind(\Pi_S, \idrv_{S\,\vec u})$ is reducible, for
  every $\vec u$ of the appropriate types.

\end{enumerate}
To prove (1), we observe that $\indm{\Pi_S[\vec u/\vec x]} \leq
\indm{\Pi_S} < \indm{\Pi}$, so by the outer induction hypothesis, the
derivation $mc(\mu(\Xi,\Pi_S), \Pi_S[\vec u/ \vec x])$ is reducible.
Property (2) is proved similarly, by observing that $\indm{\idrv_{S\,\vec u}} <
\indm{\Pi}$ (since identity derivations do not use the $\indL$ rule;
c.f. Lemma~\ref{lm:idrv}).  Property (3) follows from the fact that
$\idrv_{S\,\vec u}$ is reducible and
that $\Pi_S$ is reducible (hence, also normalizable).  Having shown
these three properties, using Lemma~\ref{lm:ind unfold} we conclude
that $\mu(\Pi_1,\Pi_S)$ is reducible, hence, by the outer induction
($\Pi'$ is smaller than $\Pi$), the reduct $\Xi'$ is reducible.

\item[$\coindR/\coindL$:] Suppose $\Pi_1$ and $\Pi$ are
$$
\infer[\coindR] {\Seq{\Delta_1}{p\,\vec{t}}} {
  \deduce{\Seq{\Delta_1}{S\,\vec{t}}}{\Pi_1'} &
  \deduce{\Seq{S\,\vec{x}}{D\,S\,\vec{x}}}{\Pi_S} } \qquad \qquad
\infer[\coindL] {\Seq{p\,\vec{t}, B_2, \dots, \Gamma}{C}}
{\deduce{\Seq{D\,p\,\vec{t},B_2,\dots, \Gamma}{C}}{\Pi'}}
$$
where $p\,\vec{x}\defnu D\,p\,\vec{x}$.  Suppose $(p\,\vec{t})\delta_1
= (p\,\vec{t})\gamma = p\,\vec{u}$. Then $\Xi$ is the derivation
$$
\infer[\mc]
{\Seq{\Delta_1\delta_1,\dots,\Delta_n\delta_n,\Gamma\gamma}{C\gamma}}
{\infer[\coindR] {\Seq{\Delta_1\delta_1}{p\,\vec{u}}} {
    \deduce{\Seq{\Delta_1\delta_1}{S\,\vec{u}}} {\Pi_1'\delta_1} &
    \deduce{\Seq{S\,\vec{x}}{D\,S\,\vec{x}}}{\Pi_S} } & \cdots &
  \infer[\coindL] {\Seq{p\,\vec{u}, \dots, \Gamma\gamma}{C\gamma}} {
    \deduce{\Seq{D\,p\,\vec{u},\dots,\Gamma\gamma}{C\gamma}}
    {\Pi'\gamma} } }
$$
Let $\Rscr = \bigcup \{\RED_F \mid \level F < \level p\}$. Since
$\Pi_1$ is reducible, there exists a $(\Rscr,\Pi_S)$-saturated set
$\Sscr$ such that $\Pi_1' \in \Sscr$.  Let $\Xi_1$ be the derivation
$$
\infer[\mc] {\Seq{\Delta_1\delta_1}{D\,S\,\vec{u}}} {
  \deduce{\Seq{\Delta_1\delta_1}{S\,\vec{u}}} {\Pi_1'\delta_1} &
  \deduce{\Seq{S\,\vec{u}} {D\,S\,\vec{u}}}{\Pi_S[\vec u/ \vec x]} }
\enspace 
$$
Since $\Sscr$ is a $(\Rscr,\Pi_S)$-saturated set, by
Definition~\ref{def:saturated sets}, $\Xi_1 \in \RED_{D\,p\,\vec
  u}^p[\Rscr,\Sscr]$. It then follows from Lemma~\ref{lm:coind unfold}
that $\nu(\Xi_1,\Pi_S)$ is reducible.

The reduct of $\Xi$ is the derivation
$$
\infer[\mc .]{\Seq{\Delta_1\delta_1,\ldots,\Delta_n\delta_n,\Gamma\gamma}{C\gamma}}
{\deduce{\Seq{\Delta_1\delta_1}{D\,p\,\vec{u}}} {\nu(\Xi_1,\Pi_S)} &
  \cdots & \deduce{\Seq{\Delta_n\delta_n}{B_n\delta_n}}
  {\Pi_n\delta_n} &
  \deduce{\Seq{D\,p\,\vec{u},\ldots,B_n\gamma,\Gamma\gamma} {C\gamma}}
  {\Pi'\gamma}}
$$
Its reducibility follows from the reducibility of $\nu(\Xi_1,\Pi_S)$
and the outer induction hypothesis.

\item[$\oimpL/\circ\Lscr$:] Suppose $\Pi_1$ is
$$
\infer[\oimpL] {\Seq{D_1'\oimp D_1'',\Delta_1'}{B_1}} {
  \deduce{\Seq{\Delta_1'}{D_1'}}{\Pi_1'} &
  \deduce{\Seq{D_1'',\Delta_1'}{B_1}}{\Pi_1''} }
$$
Since $\Pi_1$ is reducible, it follows from Definition~\ref{def:red}
that $\Pi_1'$ is normalizable and $\Pi_1''$ is reducible.  Let $\Xi_1$
be the derivation
$$
\infer[\mc]
{\Seq{D_1''\delta_1,\Delta_1'\delta_1,\Delta_2\delta_2,\dots,\Gamma\gamma}{C\gamma}}
{\deduce{\Seq{D_1''\delta_1,\Delta_1'\delta_1}{B_1\delta_1}}
  {\Pi_1''\delta_1} &
  \deduce{\Seq{\Delta_2\delta_2}{B_2\delta_2}}{\Pi_2\delta_2} & \cdots
  & \deduce{\Seq{B_1\delta_1,\dots,\Gamma\gamma}{C\gamma}} {\Pi\gamma}
} \enspace 
$$
$\Xi_1$ is reducible by induction hypothesis on the reduction of
$\Pi_1$ ($\Pi_1''$ is a predecessor of $\Pi_1$).  The reduct of $\Xi$
in this case is the derivation \settowidth{\infwidthii}
{$\Seq{\Delta_1'\delta_1,\Delta_2\delta_2,\dots,\Gamma\gamma}{D_1'\delta_1}$}
\begin{displaymath}
  \infer[\oimpL]
  {\Seq{(D_1' \oimp D_1'')\delta_1,\Delta_1'\delta_1,\Delta_2\delta_2,\ldots,\Gamma\gamma}{C\gamma}}
  {
    \infer{\Seq{\Delta_1'\delta_1,\Delta_2\delta_2,\ldots,\Gamma\gamma}{D_1'\delta_1}}
    {\infer[\wL]{\makebox[\infwidthii]{}}
      {\deduce{\Seq{\Delta_1'\delta_1}{D_1'\delta_1}}{\Pi_1'\delta_1}} }
    & \deduce{\Seq{D_1''\delta_1,\Delta_1'\delta_1,\Delta_2\delta_2,\ldots,\Gamma\gamma}{C\gamma}}
    {\Xi_1}}
  \enspace 
\end{displaymath}
Since $\Pi_1'$ is normalizable and substitutions preserve
normalizability, by Definition~\ref{def:norm} the left premise of the
reduct is normalizable, and hence the reduct is reducible.

\item[$\eqL/\circL$:] Suppose $\Pi_1$ is
$$
\infer[\eqL]{\Seq{s=t,\Delta_1'}{B_1}} {\left\{\raisebox{-1.5ex}
    {\deduce{\Seq{\Delta_1'\rho}{B_1\rho}} {\Pi^{\rho}}}
  \right\}_{\rho}}
$$
Then $\Xi$ is the derivation
$$
\infer[\mc]
{\Seq{(s=t)\delta_1,\Delta_1'\delta_1,\Delta_2\delta_2,\dots,\Gamma\gamma}
  {C\gamma}} {\infer[\eqL]{\Seq{(s=t)\delta_1,
      \Delta_1'\delta_1}{B_1\delta_1}} {\left\{\raisebox{-1.5ex}
      {\deduce{\Seq{\Delta_1'\delta_1\rho'}{B_1\delta_1\rho'}}
        {\Pi^{\delta_1\circ \rho'}}} \right\}_{\rho'}} &
  \deduce{\Seq{\Delta_2\delta_2}{B_2\delta_2}}{\Pi_2\delta_2} & \cdots
  & \deduce{\Seq{B_1\gamma,\dots,\Gamma\gamma}{C\gamma}} {\Pi\gamma} }
$$
Notice that each premise derivation $\Pi^{\delta_1\circ\rho'}$ of
$\Pi_1\delta_1$ is a also a premise derivation of $\Pi_1$, since for
every unifier $\rho'$ of $(s = t) \delta_1$, there is a unifier of $s
= t$, i.e., the substitution $\delta_1 \circ \rho'$.  Therefore every
$\Pi^{\delta_1 \circ \rho'}$ is a predecessor of $\Pi_1$.  Let
$\Xi^{\rho'}$ be the derivation
$$
\infer[\mc.]  {\Seq{\Delta_1'\delta_1\rho',
    \Delta_2\delta_2\rho',\dots,\Gamma\gamma\rho' } {C\gamma\rho'}} {
  \deduce{\Seq{\Delta_1'\delta_1\rho'}{B_1\delta_1\rho'}}
  {\Pi_1^{\delta_1\circ\rho'}} &
  \deduce{\Seq{\Delta_2\delta_2\rho'}{B_2\delta_2\rho'}}
  {\Pi_2\delta_2\rho'} & \raisebox{1.5ex}{\ldots} &
  \deduce{\Seq{B_1\gamma\rho',\dots,\Gamma\gamma\rho'}{C\gamma\rho'}}
  {\Pi\gamma\rho'} }
$$
The reduct of $\Xi$
$$
\infer[\eqL]{\Seq{(s=t)\delta_1,\Delta_1'\delta_1,\dots,\Gamma\gamma}
  {C\gamma}} {\left\{\raisebox{-1.5ex}
    {\deduce{\Seq{\Delta_1'\delta_1\rho',\dots,\Gamma\gamma\rho'}{C\gamma\rho'}}
      {\Xi^{\rho'}}} \right\}_{\rho'}}
$$
is then reducible by Definition~\ref{def:red}.

\item[$\indL/\circL$:] Suppose $\Pi_1$ is
$$
\infer[\indL] {\Seq{p\,\vec{t}, \Delta_1'}{B_1}} {
  \deduce{\Seq{D\,S\,\vec{x}}{S\,\vec{x}}}{\Pi_S} &
  \deduce{\Seq{S\,\vec{t}, \Delta_1'}{B_1}}{\Pi_1'} } \enspace 
$$
Since $\Pi_1$ is reducible, it follows from the definition of
reducibility that $\Pi_1'$ is reducible predecessor of $\Pi_1$ and
$\Pi_S$ is normalizable.  Suppose $p\,\vec{u} = (p\,\vec{t})\delta_1 =
(p\,\vec{t})\gamma$.  Let $\Xi_1$ be the derivation
$$
\infer[\mc]{\Seq{S\,\vec{u},
    \Delta_1'\delta_1,\ldots,\Delta_n\delta_n,\Gamma\gamma}{C\gamma}}
{\deduce{\Seq{S\,\vec{u}, \Delta_1'\delta_1}{B_1\delta_1}}
  {\Pi_1'\delta_1} & \cdots &
  \deduce{\Seq{\Delta_n\delta_n}{B_n\delta_n}} {\Pi_n\delta_n} &
  \deduce{\Seq{B_1\gamma,\ldots,B_n\gamma,\Gamma\gamma}
    {C\gamma}}{\Pi\gamma}} \enspace 
$$
$\Xi_1$ is reducible by induction on the reduction of $\Pi_1$,
therefore the reduct of $\Xi$
$$
\infer[\indL] {\Seq{p\,\vec{u},\Delta_1'\delta_1,
    \dots,\Delta_n\delta_n,\Gamma\gamma}{C\gamma}} {
  \deduce{\Seq{D\,S\,\vec{x}}{S\,\vec{x}}}{\Pi_S} &
  \deduce{\Seq{S\,\vec{u},\Delta_1'\delta_1,
      \ldots,\Delta_n\delta_n,\Gamma\gamma}{C\gamma}}{\Xi_1} }
$$
is reducible.
 
\item[$-/\oimpL$:] Suppose $\Pi$ is
  \begin{displaymath}
    \infer[\oimpL]{\Seq{B_1,\ldots,B_n,D' \oimp D'',\Gamma'}{C}}
    {\deduce{\Seq{B_1,\ldots,B_n,\Gamma'}{D'}}
      {\Pi'}
      & \deduce{\Seq{B_1,\ldots,B_n,D'',\Gamma'}{C}}
      {\Pi''}}
    \enspace 
  \end{displaymath}
  Let $\Xi_1$ be
  \begin{displaymath}
    \infer[\mc]{\Seq{\Delta_1\delta_1,\ldots,\Delta_n\delta_n,\Gamma'\gamma}{D'\gamma}}
    {\deduce{\Seq{\Delta_1\delta_1}{B_1\delta_1}}
      {\Pi_1\delta_1}
      & \cdots
      & \deduce{\Seq{\Delta_n\delta_n}{B_n\delta_n}}
      {\Pi_n\delta_n}
      & \deduce{\Seq{B_1\gamma,\ldots,B_n\gamma,\Gamma'\gamma}
        {D'\gamma}}
      {\Pi'\gamma}}
  \end{displaymath}
  and $\Xi_2$ be
  \begin{displaymath}
    \infer[\mc]{\Seq{\Delta_1\delta_1,\ldots,\Delta_n\delta_n,
        D''\gamma,\Gamma'\gamma}{C\gamma}}
    {\deduce{\Seq{\Delta_1\delta_1}{B_1\delta_1}}
      {\Pi_1\delta_1}
      & \cdots
      & \deduce{\Seq{\Delta_n\delta_n}{B_n\delta_n}}
      {\Pi_n\delta_n}
      & \deduce{\Seq{B_1\gamma,\ldots,B_n\gamma,D''\gamma,\Gamma'\gamma}{C\gamma}}
      {\Pi''\gamma}}
    \enspace 
  \end{displaymath}
  Both $\Xi_1$ and $\Xi_2$ are reducible by induction hypothesis.
  Therefore the reduct of $\Xi$
  \begin{displaymath}
    \infer[\oimpL]{\Seq{\Delta_1\delta_1,\ldots,\Delta_n\delta_n,(D' \oimp D'')\gamma,
        \Gamma'\gamma}{C\gamma}}
    {\deduce{\Seq{\Delta_1\delta_1,\ldots,\Delta_n\delta_n,\Gamma'\gamma}{D'\gamma}}
      {\Xi_1}
      & \quad &
      \deduce{\Seq{\Delta_1\delta_1,\ldots,\Delta_n\delta_n,
          D''\gamma,\Gamma'\gamma}{C\gamma}}
      {\Xi_2}}
    \enspace 
  \end{displaymath}
  is reducible (reducibility of $\Xi_1$ implies its normalizability by
  Lemma~\ref{lm:subst-norm}).

\item[$-/\coindR$:] Suppose $\Pi$ is
$$
\infer[\coindR] {\Seq{B_1,\dots,B_n,\Gamma}{p\,\vec{t}}} {
  \deduce{\Seq{B_1,\dots,B_n,\Gamma}{S\,\vec{t}}}{\Pi'} &
  \deduce{\Seq{S\,\vec{x}}{D\,S\,\vec{x}}}{\Pi_S} } \enspace ,
$$
where $p\,\vec{x} \defnu D\,p\,\vec{x}$.  Suppose $p\,\vec{u} =
(p\,\vec{t})\delta_1 = (p\,\vec{t})\gamma$.  Let $\Xi_1$ be the
derivation
\begin{displaymath}
  \infer[\mc]{\Seq{\Delta_1\delta_1,\ldots,\Delta_n\delta_n,\Gamma\gamma}
    {S\,\vec{u}}}
  {\deduce{\Seq{\Delta_1\delta_1}{B_1\delta_1}}
    {\Pi_1\delta_1}
    & \cdots
    & \deduce{\Seq{\Delta_n\delta_n}{B_n\delta_n}}
    {\Pi_n\delta_n}
    & \deduce{\Seq{B_1\gamma,\ldots,B_n\gamma,\Gamma\gamma}
      {S\,\vec{u}}}{\Pi'\gamma}}
  \enspace .
\end{displaymath}
The derivations $\Pi'\gamma$, $\Pi_S$, $\Xi_1$ and the derivation
$$
\infer[\mc] {\Seq{\Delta'}{D\,S\,\vec{w}}} {
  \deduce{\Seq{\Delta'}{S\,\vec{w}}}{\Psi} &
  \deduce{\Seq{S\,\vec{w}}{D\,S\,\vec{w}}}{\Pi_S[\vec w/ \vec x]} }
\enspace,
$$
where $\Psi$ is any reducible derivation, are all reducible by
induction hypothesis on the length of $\Pi$.  Again, we use the same
arguments as in the case where $n=0$ to construct a
$(\Rscr,\Pi_S)$-saturated set $\Sscr$ such that $\Xi_1 \in \Sscr$.
Therefore by Definition~\ref{def:red}, the reduct of $\Xi$:
$$
\infer[\coindR]
{\Seq{\Delta_1\delta_1,\dots,\Delta_n\delta_n,\Gamma\gamma}
  {p\,\vec{u}}} {
  \deduce{\Seq{\Delta_1\delta_1,\dots,\Delta_n\delta_n,\Gamma\gamma}
    {S\,\vec{u}}}{\Xi_1} &
  \deduce{\Seq{S\,\vec{x}}{D\,S\,\vec{x}}}{\Pi_S} }
$$
is reducible.

\item[$\mc/\circL$:] Suppose $\Pi_1$ ends with a $\mc$. Then any
  reduct of $\Pi_1\delta_1$ corresponds to a predecessor of $\Pi_1$ by
  Lemma~\ref{lm:reduct subst}.  Therefore the reduct of $\Xi$ is
  reducible by induction on the reduction of $\Pi_1$.

\item[$-/\init$:] $\Xi$ reduces to $\Pi_1\delta_1$.  Since $\Pi_1$ is
  reducible, by Lemma~\ref{lm:red-subst}, $\Pi_1\delta_1$ is reducible
  and hence $\Xi$ is reducible.
\end{description}
\qed
\end{proof}

\begin{corollary}
  Every derivation is reducible.
\end{corollary}
\begin{proof}
  The proof follows from Lemma~\ref{lm:comp}, by setting $n=0$. \qed
\end{proof}
Since reducibility implies cut-elimination, it follows that every
proof can be transformed into a cut-free proof.

\begin{corollary}
  \label{cor:cut-elimination}
  Given a fixed stratified definition, a sequent has a proof in
  $\Linc$ if and only if it has a cut-free proof.
\end{corollary}

The consistency of $\Linc$ is an immediate consequence of
cut-elimination.  By consistency we mean the following: given a fixed
stratified definition and an arbitrary formula $C$, it is not the case
that both $C$ and $C\oimp \bot$ are provable.

\begin{corollary}
  \label{cor:consistency}
  The logic $\Linc$ is consistent.
\end{corollary}
\begin{proof}
  Suppose otherwise, that is, there is a formula $C$ such that there
  is a proof $\Pi_1$ of $C$ and another proof $\Pi_2$ for $C\oimp
  \bot$.  Since cut elimination holds, we can assume, without loss of
  generality, that $\Pi_1$ and $\Pi_2$ are cut free. By inspection of
  the inference rules of $\Linc$, we see that $\Pi_2$ must end with
  $\oimpR$, that is, $\Pi_2$ is
$$
\infer[\oimpR] {\Seq{}{C\oimp \bot}} {\deduce{\Seq{C}{\bot}}{\Pi_2'}
}
$$
Cutting $\Pi_1$ with $\Pi_2'$ we get a derivation of
$\Seq{\cdot}{\bot}$, and applying the cut-elimination procedure we get
a cut-free derivation of $\Seq{\cdot}{\bot}$. But there cannot be such
a derivation since there is no right-introduction rule for $\bot$,
contradiction.  \qed
\end{proof}

\section{Related Work}
\label{sec:lrel}
 
Of course, there is a long association between mathematical logic and
inductive definitions~\cite{Acz77} and in particular with
proof-theory, starting with the  Takeuti's conjecture, the earliest relevant
entry for our purposes being Martin-L\"of's original formulation of
the theory of \emph{iterated inductive
  definitions}~\cite{martin-lof71sls}.  From the impredicative
encoding of inductive types~\cite{Bohm85} and the introduction of
(co)recursion~\cite{Geuvers92,MENDittcslc} in system F, (co)inductive
types became common and made it into type-theoretic proof assistants
such as Coq~\cite{PaulinMohring93}, first via a primitive recursive
operator, but eventually in the let-rec style of functional
programming languages, as in Gimenez's \emph{Calculus of Infinite
  Constructions}~\cite{Gim96phd}; here termination
(resp.~productivity) is ensured by a syntactic check known as
\emph{guarded by destructors} \cite{Gim95types}. Note that 
 Coq forbids altogether the introduction
of blocks of mutually dependent types containing both inductive and
co-inductive ones, even though they could be stratified.
Moreover, while a syntactic check has obvious advantages, it tends to
be too restrictive, as observed and improved upon in \cite{BartheFGPU04} 
by using type based termination. The same can be said about
\emph{Agda} \cite{norell:thesis}, where size types termination will
eventually supersede guardedness \cite{Mugda}.

Baelde and Miller have recently introduced an extension of linear
logic with least and greatest fixed points~\cite{baelde07lpar}. 
However, cut elimination is proved indirectly via a second-order 
encoding of the least and the greatest fixed point operators 
into higher-order linear logic and via an appeal to completeness of
focused proofs for higher-order linear logic.

Circular proofs are also connected with the emerging proof-theory of
of {fixed point logics} and {process
  calculi}~\cite{Santocanale02,Sprenger03}, as well as in traditional
sequent calculi such as in \cite{BrotherstonS07}.  The issue is the
equivalence between systems with local vs\ global induction, that is,
between fixed point rules vs.\ well-founded and guarded induction
(\ie circular proofs). In the sequent calculus it is unknown whether
every inductive proof can be obtained via global induction.




In higher order logic (co)inductive definitions are obtained via the
usual Tarski fixed point constructions, as realized for example in
Isabelle/HOL~\cite{Paulson97}.  As we mentioned before, those
approaches are at odd with HOAS even at the level of the syntax. This
issue has originated a research field in its own that we can only try
to mention the main contenders: in the {Twelf} system~\cite{Pfenning99cade}
 the LF type theory is used to encode deductive systems as judgments
and to specify meta-theorems as relations (type families) among them;
a logic programming-like interpretation provides an operational
semantics to those relations, so that an external check for totality
(incorporating termination, well-modedness and
coverage~\cite{SchurmannP03,Pientka05}) verifies that the given
relation is indeed a realizer for that theorem. Coinduction is still
unaccounted for and may require a switch to a different operational
semantics for LF.
There exists a second approach to reasoning in LF that is built on the
idea of devising an explicit (meta-)meta-logic ($\Momega$) for reasoning
(inductively) about the framework, in a fully automated
way~\cite{S00}. It can be seen as a constructive first-order
inductive type theory, whose quantifiers range over possibly open LF
objects over a signature.  In this calculus it is possible to express
and inductively prove meta-logical properties of an object level system.  
$\Momega$ can be also seen as a dependently-typed functional
programming language, and as such it has been refined first into the
\emph{Elphin} programming language~\cite{SPS:TLCA2005} and more
recently in \emph{Delphin} \cite{PosSch08}.  In a similar vein the
context modal logic of Pientka, Pfenning and
Naneski~\cite{NanevskiTOCL} provides a basis for a different
foundation for programming with HOAS and dependent types based on
hereditary substitutions, see the programming language \emph{Beluga}
(\cite{Pientka08,PientkaPPDP08}). Because all of these systems are
programming languages, we refrain from a deeper discussion. We only
note that
systems like Delphin or Beluga separate data from computations. This
means they are always based on eager evaluation, whereas co-recursive
functions should be interpreted lazily.  Using standard techniques
such as \emph{thunks} to simulate lazy evaluation in such a context
seems problematic (Pientka, personal
communication).

\emph{Weak higher-order abstract syntax}~\cite{DFHtlca95} is an
approach that strives to co-exist with an inductive setting, where the
positivity condition for datatypes and hypothetical judgments must be
obeyed.  The problem of negative occurrences in datatypes is handled
by replacing them with a new type.  The approach is extended to
hypothetical judgments by introducing distinct predicates for the
negative occurrences.  Some axioms are needed to reason about
hypothetical judgments, to mimic what is inferred by the cut rule in
our architecture.  Miculan \etal's framework~\cite{HonsellMS01}
embraces this \emph{axiomatic} approach extending Coq with the ``theory of
contexts'' (ToC).  The theory includes axioms for the the reification
of key properties of names akin to \emph{freshness}. Furthermore,
higher-order induction and recursion schemata on expressions are also
assumed.  \emph{Hybrid}~\cite{Ambler02} is a $\lambda$-calculus on top
of Isabelle/HOL which provides the user  with a \emph{Full} HOAS syntax,
compatible with a classical (co)-inductive setting. $\Linc$ improves
on the latter on several counts. First it disposes of Hybrid notion of
\emph{abstraction}, which is used to carve out the ``parametric''
function space from the full HOL space.  Moreover it is not restricted
to second-order abstract syntax, as the current Hybrid version is (and
as ToC cannot escape from being). Finally, at higher types, reasoning
via $\defL$ is more powerful than inversion, which does not exploit
higher-order unification.

ToC can be seen as a stepping stone towards Gabbay and Pitts
\emph{nominal logic}, which aims to be a foundation of programming and
reasoning with \emph{names}. It can be presented as a first-order
theory~\cite{pitts03ic}, which includes primitives for variable
renaming and variable freshness, and a (derived) new ``freshness''
quantifier.  Using this theory, it is possible to prove properties by
structural induction and also to define functions by recursion over
syntax~\cite{Pitts06}.  Urban \etal's have engineered a \emph{nominal
  datatype package} inside Isabelle/HOL~\cite{Nominal}
analogous to the standard datatype package but defining equivalence
classes of term constructors. In more recent versions, principles of
primitive recursion and strong induction have been
added~\cite{UrbanB06}. Coinduction on nominal datatypes is not
available, but to be fair it is also absent from Isabelle/HOL due to
some technical limitations in the automation of the inductive package



\section{Conclusion and Future Work}
\label{sec:conc}

We have presented a proof theoretical treatment of both induction and
co-induction in a sequent calculus compatible with HOAS encodings. The
proof principle underlying the explicit proof rules is basically fixed
point (co)induction.  We have shown some examples where informal
(co)inductive proofs using invariants and simulations are reproduced
formally in $\Linc$.

Consistency of the logic is an easy consequence of cut-elimination.
Our proof system is, as far as we know, the first which incorporates a
co-induction proof rule with a direct cut elimination proof. This
schema can be used as a springboard towards cut elimination procedures
for more expressive (conservative) extensions of $\Linc$, for example
in the direction of $\FOLNb$~\cite{miller05tocl}, or more recently,
the logic $LG^\omega$~\cite{Tiu07} by Tiu and the logic $\Gscr$ by
Gacek \etal \cite{gacek08lics}.


  

As far as future work, we may investigate loosening the stratification
condition for example in the sense of \emph{local} stratification,
possibly allowing to encode proofs such as type preservation in
operational semantics directly in $\Linc$ rather than with the 2-level
approach~\cite{mcdowell02tocl,Momigliano03fos}. More general notions
of stratifications are already allowed in practice, see the proof by
logical relations in \cite{AbellaSOS}, but not formally justified.

Another interesting problem is the connection with \emph{circular
  proofs}, which is particularly attractive from the viewpoint of
proof search, both inductively and co-inductively. This could be
realized by directly proving a cut-elimination result for a logic
where circular proofs, under termination and guardedness conditions
completely replace (co)inductive rules. Indeed, the question whether
``global'' proofs are equivalent to ``local'' proofs
\cite{BrotherstonS07} is still unsettled.

\medskip

\textbf{Acknowledgements} The $\Linc$ logic was developed in
collaboration with Dale Miller. Alberto Momigliano has been supported
by EPSRC grant GR/M98555 and partly by the MRG project
(IST-2001-33149), funded by the EC under the FET proactive initiative
on Global Computing.  


\begin{thebibliography}{10}

\bibitem{Ong}
S.~Abramsky and C.-H.~L. Ong.
\newblock Full abstraction in the lazy lambda calculus.
\newblock {\em Inf. Comput.}, 105(2):159--267, 1993.

\bibitem{Acz77}
P.~Aczel.
\newblock An introduction to inductive definitions.
\newblock In J.~Barwise, editor, {\em Handbook of Mathematical Logic},
  volume~90 of {\em Studies in Logic and the Foundations of Mathematics},
  chapter C.7, pages 739--782. North-Holland, Amsterdam, 1977.

\bibitem{Ambler02}
S.~Ambler, R.~Crole, and A.~Momigliano.
\newblock Combining higher order abstract syntax with tactical theorem proving
  and (co)induction.
\newblock In V.~A. Carre{\~n}o, editor, {\em Proceedings of the 15th
  International Conference on Theorem Proving in Higher Order Logics, Hampton,
  VA, 1-3 August 2002}, volume 2342 of {\em LNCS}. Springer Verlag, 2002.

\bibitem{BaaderS01}
F.~Baader and W.~Snyder.
\newblock Unification theory.
\newblock In J.~A. Robinson and A.~Voronkov, editors, {\em Handbook of
  Automated Reasoning}, pages 445--532. Elsevier and MIT Press, 2001.

\bibitem{baelde07lpar}
D.~Baelde and D.~Miller.
\newblock Least and greatest fixed points in linear logic.
\newblock In {\em LPAR}, Lecture Notes in Computer Science, pages 92--106.
  Springer, 2007.

\bibitem{BartheFGPU04}
G.~Barthe, M.~J. Frade, E.~Gim{\'e}nez, L.~Pinto, and T.~Uustalu.
\newblock Type-based termination of recursive definitions.
\newblock {\em Mathematical Structures in Computer Science}, 14(1):97--141,
  2004.

\bibitem{Bohm85}
C.~Bohm and A.~Berarducci.
\newblock Automatic synthesis of typed lambda -programs on term algebras.
\newblock {\em Theoretical Computer Science}, 39(2-3):135--153, Aug. 1985.

\bibitem{BrotherstonS07}
J.~Brotherston and A.~Simpson.
\newblock Complete sequent calculi for induction and infinite descent.
\newblock In {\em LICS}, pages 51--62. IEEE Computer Society, 2007.

\bibitem{clark78}
K.~L. Clark.
\newblock Negation as failure.
\newblock In J.~Gallaire and J.~Minker, editors, {\em Logic and Data Bases},
  pages 293--322. Plenum Press, New York, 1978.

\bibitem{DeBruijn91lf}
N.~de~Bruijn.
\newblock A plea for weaker frameworks.
\newblock In G.~Huet and G.~Plotkin, editors, {\em Logical Frameworks}, pages
  40--67. Cambridge University Press, 1991.

\bibitem{DFHtlca95}
J.~Despeyroux, A.~Felty, and A.~Hirschowitz.
\newblock Higher-order abstract syntax in {Coq}.
\newblock In {\em Second International Conference on Typed Lambda Calculi and
  Applications}, pages 124--138. Springer, {\em Lecture Notes in Computer
  Science}, Apr. 1995.

\bibitem{despeyroux94lpar}
J.~Despeyroux and A.~Hirschowitz.
\newblock Higher-order abstract syntax with induction in {Coq}.
\newblock In {\em Fifth International Conference on Logic Programming and
  Automated Reasoning}, pages 159--173, June 1994.

\bibitem{eriksson91elp}
L.-H. Eriksson.
\newblock A finitary version of the calculus of partial inductive definitions.
\newblock In L.-H. Eriksson, L.~Halln{\"{a}}s, and P.~Schroeder-Heister,
  editors, {\em Proceedings of the Second International Workshop on Extensions
  to Logic Programming}, volume 596 of {\em Lecture Notes in Artificial
  Intelligence}, pages 89--134. Springer-Verlag, 1991.

\bibitem{gacek08lics}
A.~Gacek, D.~Miller, and G.~Nadathur.
\newblock Combining generic judgments with recursive definitions.
\newblock In {\em LICS}, pages 33--44. IEEE Computer Society, 2008.

\bibitem{AbellaSOS}
A.~Gacek, D.~Miller, and G.~Nadathur.
\newblock Reasoning in {Abella} about structural operational semantics
  specifications.
\newblock In A.~Abel and C.~Urban, editors, {\em Informal proceedings of
  LFMTP'08}. To appearin ENTCS, 2008.

\bibitem{Geuvers92}
H.~Geuvers.
\newblock Inductive and coinductive types with iteration and recursion.
\newblock In B.~Nordstr{\"o}m, K.~Pettersson, and G.~Plotkin, editors, {\em
  Informal Proceedings Workshop on Types for Proofs and Programs, B{\aa}stad,
  Sweden, 8--12 June 1992}, pages 193--217. Dept.\ of Computing Science,
  Chalmers Univ.\ of Technology and G{\"o}teborg Univ., 1992.

\bibitem{Gim95types}
E.~Gim{\'e}nez.
\newblock Codifying guarded definitions with recursion schemes.
\newblock In P.~Dybjer and B.~Nordstr{\"o}m, editors, {\em Selected Papers 2nd
  Int.\ Workshop on Types for Proofs and Programs, TYPES'94, B{\aa}stad,
  Sweden, 6--10 June 1994}, volume 996 of {\em Lecture Notes in Computer
  Science}, pages 39--59. Springer-Verlag, Berlin, 1994.

\bibitem{Gim96phd}
E.~Gim{\'e}nez.
\newblock {\em Un Calcul de Constructions Infinies et son Application a la
  Verification des Systemes Communicants}.
\newblock {PhD} thesis {PhD} 96-11, Laboratoire de l'Informatique du
  Parall{\'e}lisme, Ecole Normale Sup{\'e}rieure de Lyon, Dec. 1996.

\bibitem{girard89book}
J.-Y. Girard, P.~Taylor, and Y.~Lafont.
\newblock {\em Proofs and Types}.
\newblock Cambridge University Press, 1989.

\bibitem{PID}
L.~Halln\"{a}s.
\newblock Partial inductive definitions.
\newblock {\em Theor. Comput. Sci.}, 87(1):115--142, 1991.

\bibitem{harper93jacm}
R.~Harper, F.~Honsell, and G.~Plotkin.
\newblock A framework for defining logics.
\newblock {\em Journal of the ACM}, 40(1):143--184, 1993.

\bibitem{HonsellMS01}
F.~Honsell, M.~Miculan, and I.~Scagnetto.
\newblock An axiomatic approach to metareasoning on nominal algebras in {HOAS}.
\newblock In F.~Orejas, P.~G. Spirakis, and J.~van Leeuwen, editors, {\em
  ICALP}, volume 2076 of {\em Lecture Notes in Computer Science}, pages
  963--978. Springer, 2001.

\bibitem{Jacobs97}
B.~Jacobs and J.~Rutten.
\newblock A tutorial on (co)algebras and (co)induction.
\newblock {\em Bulletin of the European Association for Theoretical Computer
  Science}, 62:222--259, June 1997.
\newblock Surveys and Tutorials.

\bibitem{martin-lof71sls}
P.~Martin-L{\"{o}}f.
\newblock Hauptsatz for the intuitionistic theory of iterated inductive
  definitions.
\newblock In J.~Fenstad, editor, {\em Proceedings of the Second Scandinavian
  Logic Symposium}, volume~63 of {\em Studies in Logic and the Foundations of
  Mathematics}, pages 179--216. North-Holland, 1971.

\bibitem{mcdowell00tcs}
R.~McDowell and D.~Miller.
\newblock Cut-elimination for a logic with definitions and induction.
\newblock {\em Theoretical Computer Science}, 232:91--119, 2000.

\bibitem{mcdowell02tocl}
R.~McDowell and D.~Miller.
\newblock Reasoning with higher-order abstract syntax in a logical framework.
\newblock {\em ACM Transactions on Computational Logic}, 3(1):80--136, January
  2002.

\bibitem{mcdowell03tcs}
R.~McDowell, D.~Miller, and C.~Palamidessi.
\newblock Encoding transition systems in sequent calculus.
\newblock {\em TCS}, 294(3):411--437, 2003.

\bibitem{Mugda}
K.~Mehltretter.
\newblock Termination checking for a dependently typed language.
\newblock Master's thesis, LMU, Dec. 2007.
\newblock Diplomarbeit.

\bibitem{MENDittcslc}
N.~P. Mendler.
\newblock Inductive types and type constraints in the second order lambda
  calculus.
\newblock {\em Annals of Pure and Applied Logic}, 51(1):159--172, 1991.

\bibitem{Miller91elp}
D.~Miller.
\newblock A logic programming language with lambda-abstraction, function
  variables, and simple unification.
\newblock In P.~Schroeder-Heister, editor, {\em {Extensions of Logic
  Programming: International Workshop, T\"ubingen}}, volume 475 of {\em LNAI},
  pages 253--281. Springer-Verlag, 1991.

\bibitem{miller05tocl}
D.~Miller and A.~Tiu.
\newblock A proof theory for generic judgments.
\newblock {\em ACM Trans. Comput. Logic}, 6(4):749--783, 2005.

\bibitem{Momigliano03fos}
A.~Momigliano and S.~Ambler.
\newblock Multi-level meta-reasoning with higher order abstract syntax.
\newblock In A.~Gordon, editor, {\em FOSSACS'03}, volume 2620 of {\em LNCS},
  pages 375--392. Springer Verlag, 2003.

\bibitem{Momigliano03TYPES}
A.~Momigliano and A.~Tiu.
\newblock Induction and co-induction in sequent calculus.
\newblock In S.~Berardi, M.~Coppo, and F.~Damiani, editors, {\em TYPES}, volume
  3085 of {\em Lecture Notes in Computer Science}, pages 293--308. Springer,
  2003.

\bibitem{NanevskiTOCL}
A.~Nanevski, B.~Pientka, and F.~Pfenning.
\newblock Contextual modal type theory.
\newblock {\em {ACM Transactions on Computational Logic}}, 200?
\newblock To appear.

\bibitem{Nominal}
{Nominal Methods Group}.
\newblock Nominal {Isabelle}.
\newblock isabelle.in.tum.de/nominal/, 2008, Accessed 2 July 2008.

\bibitem{norell:thesis}
U.~Norell.
\newblock {\em Towards a practical programming language based on dependent type
  theory}.
\newblock PhD thesis, Department of Computer Science and Engineering, Chalmers
  University of Technology, SE-412 96 G\"{o}teborg, Sweden, September 2007.

\bibitem{PaulinMohring93}
C.~Paulin-Mohring.
\newblock Inductive definitions in the system {Coq}: Rules and properties.
\newblock In M.~Bezem and J.~Groote, editors, {\em Proceedings of the
  International Conference on Typed Lambda Calculi and Applications}, pages
  328--345, Utrecht, The Netherlands, Mar. 1993. Springer-Verlag LNCS 664.

\bibitem{Paulson97}
L.~C. Paulson.
\newblock Mechanizing coinduction and corecursion in higher-order logic.
\newblock {\em Journal of Logic and Computation}, 7(2):175--204, Mar. 1997.

\bibitem{pfenning01handbook}
F.~Pfenning.
\newblock Logical frameworks.
\newblock In A.~Robinson and A.~Voronkov, editors, {\em Handbook of Automated
  Reasoning}, chapter~17, pages 1063--1147. Elsevier Science Publisher and MIT
  Press, 2001.

\bibitem{PfenningE88}
F.~Pfenning and C.~Elliott.
\newblock Higher-order abstract syntax.
\newblock In {\em PLDI}, pages 199--208, 1988.

\bibitem{Pfenning99cade}
F.~Pfenning and C.~Sch{\"u}rmann.
\newblock System description: Twelf --- a meta-logical framework for deductive
  systems.
\newblock In H.~Ganzinger, editor, {\em Proceedings of the 16th International
  Conference on Automated Deduction (CADE-16)}, pages 202--206, Trento, Italy,
  July 1999. Springer-Verlag LNAI 1632.

\bibitem{Pientka05}
B.~Pientka.
\newblock Verifying termination and reduction properties about higher-order
  logic programs.
\newblock {\em J. Autom. Reasoning}, 34(2):179--207, 2005.

\bibitem{Pientka08}
B.~Pientka.
\newblock A type-theoretic foundation for programming with higher-order
  abstract syntax and first-class substitutions.
\newblock In G.~C. Necula and P.~Wadler, editors, {\em POPL}, pages 371--382.
  ACM, 2008.

\bibitem{PientkaPPDP08}
B.~Pientka and J.~Dunfield.
\newblock Programming with proofs and explicit contexts.
\newblock In {\em PPDP}. ACM Press, 2008.

\bibitem{pitts03ic}
A.~M. Pitts.
\newblock Nominal logic, a first order theory of names and binding.
\newblock {\em Information and Computation}, 186(2):165--193, 2003.

\bibitem{Pitts06}
A.~M. Pitts.
\newblock Alpha-structural recursion and induction.
\newblock {\em J. ACM}, 53(3):459--506, 2006.

\bibitem{PosSch08}
A.~Poswolsky and C.~Sch{\"u}rmann.
\newblock Practical programming with higher-order encodings and dependent
  types.
\newblock In S.~Drossopoulou, editor, {\em ESOP}, volume 4960 of {\em Lecture
  Notes in Computer Science}, pages 93--107. Springer, 2008.

\bibitem{Santocanale02}
L.~Santocanale.
\newblock A calculus of circular proofs and its categorical semantics.
\newblock In M.~Nielsen and U.~Engberg, editors, {\em FoSSaCS}, volume 2303 of
  {\em Lecture Notes in Computer Science}, pages 357--371. Springer, 2002.

\bibitem{schroeder-heister92nlip}
P.~Schroeder-Heister.
\newblock Cut-elimination in logics with definitional reflection.
\newblock In D.~Pearce and H.~Wansing, editors, {\em Nonclassical Logics and
  Information Processing}, volume 619 of {\em LNCS}, pages 146--171. Springer,
  1992.

\bibitem{SchroederHeister93elp}
P.~Schroeder-Heister.
\newblock Definitional reflection and the completion.
\newblock In R.~Dyckhoff, editor, {\em Proceedings of the 4th International
  Workshop on Extensions of Logic Programming}, pages 333--347. Springer-Verlag
  LNAI 798, 1993.

\bibitem{schroeder-heister93lics}
P.~Schroeder-Heister.
\newblock Rules of definitional reflection.
\newblock In M.~Vardi, editor, {\em Eighth {Annual Symposium on Logic in
  Computer Science}}, pages 222--232. IEEE Computer Society Press, IEEE, June
  1993.

\bibitem{S00}
C.~Sch{\"u}rmann.
\newblock {\em Automating the Meta-Theory of Deductive Systems}.
\newblock PhD thesis, Carnegie-Mellon University, 2000.
\newblock CMU-CS-00-146.

\bibitem{SchurmannP03}
C.~Sch{\"u}rmann and F.~Pfenning.
\newblock A coverage checking algorithm for {LF}.
\newblock In D.~A. Basin and B.~Wolff, editors, {\em TPHOLs}, volume 2758 of
  {\em Lecture Notes in Computer Science}, pages 120--135. Springer, 2003.

\bibitem{SPS:TLCA2005}
C.~Sch{\"u}rmann, A.~Poswolsky, and J.~Sarnat.
\newblock The $\bigtriangledown$-calculus. {F}unctional programming with
  higher-order encodings.
\newblock In {\em Seventh International Conference on Typed Lambda Calculi and
  Applications}, pages 339--353. Springer, {\em Lecture Notes in Computer
  Science}, Apr. 2005.

\bibitem{Sprenger03}
C.~Spenger and M.~Dams.
\newblock On the structure of inductive reasoning: Circular and tree-shaped
  proofs in the $\mu$-calculus.
\newblock In A.~Gordon, editor, {\em FOSSACS'03}, volume 2620 of {\em LNCS},
  pages 425--440,. Springer Verlag, 2003.

\bibitem{tiu04phd}
A.~Tiu.
\newblock {\em A Logical Framework for Reasoning about Logical Specifications}.
\newblock PhD thesis, Pennsylvania State University, May 2004.

\bibitem{Tiu07}
A.~Tiu.
\newblock A logic for reasoning about generic judgments.
\newblock {\em Electr. Notes Theor. Comput. Sci.}, 174(5):3--18, 2007.

\bibitem{tiu05entcs}
A.~Tiu and D.~Miller.
\newblock A proof search specification of the pi-calculus.
\newblock {\em Electr. Notes Theor. Comput. Sci.}, 138(1):79--101, 2005.

\bibitem{tiu05concur}
A.~F. Tiu.
\newblock Model checking for pi-calculus using proof search.
\newblock In {\em Proceedings of CONCUR 2005}, volume 3653 of {\em Lecture
  Notes in Computer Science}, pages 36--50. Springer, 2005.

\bibitem{UrbanB06}
C.~Urban and S.~Berghofer.
\newblock A recursion combinator for nominal datatypes implemented in
  {Isabelle/HOL}.
\newblock In U.~Furbach and N.~Shankar, editors, {\em IJCAR}, volume 4130 of
  {\em Lecture Notes in Computer Science}, pages 498--512. Springer, 2006.

\end{thebibliography}

\end{document}